\documentclass[a4paper,10.5pt]{article}
\usepackage[utf8]{inputenc}      
\usepackage[T1]{fontenc}           
\usepackage[french,english]{babel}

\usepackage{booktabs}
\usepackage{amssymb}               
\usepackage{amsmath}              
\usepackage{amsthm}                
\usepackage{graphicx}              
\usepackage{color}                 
\usepackage{calc}
\usepackage{subfigure}
\usepackage{fancyhdr}
\usepackage{rotating}
\usepackage{lscape}
\usepackage{pdflscape}
\usepackage{imakeidx}
\usepackage{bbm}
\usepackage{mwe}
\usepackage{dsfont}
\usepackage{mathrsfs}
\usepackage{calligra}
\usepackage[all]{xy}
\usepackage{multicol}
\usepackage[table]{xcolor}
\usepackage{tikz}
\usepackage{epstopdf} 
\usepackage{lmodern,textcomp}
\usepackage[export]{adjustbox}
\usepackage[justification=centering]{caption}
\usepackage{float}
\usepackage{verbatim} 
\usepackage{setspace}
\usepackage{bigints}
\usepackage[colorlinks=true,linkcolor=black,urlcolor=black,citecolor=black]{hyperref}
\usepackage{textcomp}
\usepackage{accents}
\usepackage{enumerate}
\usepackage{mathtools}
\usepackage{nicematrix}
\usepackage{appendix}
\usepackage{framed}
\usepackage{cite}

\usepackage{changes}

\newcommand{\mG}{\mathcal{G}}

\newcommand{\vect}[1]{\operatorname{vec}(#1)}

\newcommand{\nucl}[2]{{\nu_{#1, #2}}}
\newcommand{\nuclst}{{\nu_{\mG_1, \mG_2}}}

\newcommand{\diffXclst}{\mathbf{X}^T\nu}
\newcommand{\diffxclst}{\mathbf{x}^T\nu}

\newcommand{\spanned}{\operatorname{span}}
\newcommand{\setC}[1]{{\mathcal{C}_{#1}}}
\newcommand{\mP}{\mathcal{P}}

\newcommand{\nuspace}{\Lambda}

\newcommand{\bPiV}{\bPi_{\nuspace}}
\newcommand{\bPiVperp}{\bPi_{\nuspace^\perp}}
\newcommand{\R}{\mathbb{R}}

\newcommand{\bA}{\mathbf{A}}

\newcommand{\bI}{\mathbf{I}}
\newcommand{\bJ}{\boldsymbol{1}_{n\times n}}
\newcommand{\bM}{\mathbf{M}}
\newcommand{\bP}{\mathbf{P}}
\newcommand{\bU}{\mathbf{U}}

\newcommand{\bx}{\mathbf{x}}
\newcommand{\bX}{\mathbf{X}}
\newcommand{\by}{\mathbf{y}}
\newcommand{\bY}{\mathbf{Y}}

\newcommand{\bOne}{\boldsymbol{1}}
\newcommand{\bmu}{\boldsymbol{\mu}}
\newcommand{\bnu}{\boldsymbol{\nu}}
\newcommand{\bSigma}{\boldsymbol{\Sigma}}
\newcommand{\bGamma}{{\boldsymbol{\Gamma}}}
\newcommand{\bPi}{\boldsymbol{\Pi}}
\newcommand{\normal}{\mathcal{N}}
\newcommand{\range}{\operatorname{Range}}

\newcommand{\knu}{\Lambda}
\newcommand{\Piknu}{\bPi_{\knu}}

\newcommand{\Piknuperp}{\bPi_{\knu^\perp}}

\def\R{\mathbb R}

\DeclareMathAlphabet{\mathcalligra}{T1}{calligra}{m}{n}

\newcommand{\abs}[1]{\left\lvert#1\right\rvert}
\newcommand{\norm}[1]{||#1||}

\usepackage{titlesec}
\titleformat*{\section}{\large\bfseries}
\titleformat*{\subsection}{\bfseries}
\titleformat*{\subsubsection}{\itshape}

\usepackage[english]{cleveref}

\begin{document}

\newtheorem{definition}{Definition}[section]
\newtheorem{theorem}{Theorem}[section]
\newtheorem{lemma}[theorem]{Lemma}
\newtheorem{cor}[theorem]{Corollary}
\newtheorem{prop}[theorem]{Proposition}
\newtheorem{remark}{Remark}[section]
\newtheorem{assumption}{Assumption}[section]

\newcommand\independent{\protect\mathpalette{\protect\independenT}{\perp}}
\def\independenT#1#2{\mathrel{\rlap{$#1#2$}\mkern2mu{#1#2}}}
\newcommand{\notindependent}{\not\!\perp\!\!\!\perp}
\setcounter{section}{0}

\begin{center}
\Large\textbf{Post-clustering Inference under Dependence}\\
\vspace{0.2cm}
\large{Javier González-Delgado$^{1,2,3}$, Mathis Deronzier$^{2}$, Juan Cortés$^{3}$ and Pierre Neuvial$^{2}$}
\end{center}
{$^1$\textit{\scriptsize Universit\'e de Rennes, ENSAI, CNRS, CREST-UMR 9194, F-35000 Rennes, France.}}\\
{$^2$\textit{\scriptsize Institut de Math\'ematiques de Toulouse, UMR5219, Universit\'e de Toulouse, CNRS, UPS, F-31062 Toulouse Cedex 9, France.}}
{$^3$\textit{\scriptsize LAAS-CNRS, Universit{\'e} de Toulouse, CNRS, F-31400 Toulouse, France.}}

\begin{abstract}
    Recent work by Gao \textit{et al.}~\cite{Gao} has laid the foundations for post-clustering inference, establishing a theoretical framework allowing to test for differences between means of estimated clusters. Additionally, they studied the estimation of unknown parameters while controlling the selective type I error. However, their theory was developed for independent observations identically distributed as $p$-dimensional Gaussian variables, where the parameter estimation could only be performed for spherical covariance matrices. Here, we aim at extending this framework to a more convenient scenario for practical applications, where arbitrary dependence structures between observations and features are allowed. We establish sufficient conditions for extending the setting presented in~\cite{Gao} to the general dependence framework. Moreover, 
    we assess theoretical conditions allowing the compatible estimation of a covariance matrix. The theory is developed for hierarchical agglomerative clustering algorithms with several types of linkages, and for the $k$-means algorithm. We illustrate our method with synthetic data and real data of protein structures.
\end{abstract}

\section{Introduction}

Post-selection inference has gained substantial attention in recent years due to its potential to address practical problems in diverse fields. The issue of using data to answer a question that has been chosen based on the same data was formalized in \cite{fithian2017optimal}, where the basis of selective hypothesis testing was rigorously set with the definition of the selective type I error. This paved the way to perform selective testing when null hypotheses are chosen through clustering algorithms, bypassing the naive data splitting that reveals unsuitable in this context. However, their proposed approach, referred to as \textit{data carving}, as well as more recent approaches like \textit{data fission}~\cite{datafission} are difficult to implement in practice  because they require knowledge of the covariance structure between variables. Moreover, they often involve the non-trivial calibration of a tuning parameter that controls the proportion of information allocated for model selection and for inference. The seminal work by Gao \textit{et al.}~\cite{Gao} established a theoretical framework allowing selective testing after clustering using all the information in the data set. Their method is defined for independent observations identically distributed as $p$-dimensional Gaussian random variables with a spherical covariance matrix. This corresponds to the following matrix normal model~\cite{Gupta2018}:
\begin{equation}\label{model_gao}\tag{ind-MN}
    \mathbf{X}\sim\mathcal{MN}_{n\times p}(\boldsymbol{\mu},\mathbf{I}_n,\sigma^2\mathbf{I}_p),
\end{equation}
where $\boldsymbol\mu\in\mathcal{M}_{n\times p}(\mathbb{R})$ and $\sigma>0$. Under \eqref{model_gao}, the authors in \cite{Gao} defined a $p$-value that controls the selective type I error when testing for a difference in means between a pair of estimated clusters. This $p$-value can be efficiently computed for hierarchical clustering algorithms with common linkage functions. Moreover, the authors in \cite{Gao} made another remarkable contribution by addressing the estimation of $\sigma$ while controlling the selective type I error, which had not been addressed in previous works~\cite{datafission, rasines} despite its major importance in applications. They showed that if $\sigma$ is asymptotically over-estimated, the $p$-value is asymptotically super-uniform under the null, and provided an estimator $\hat{\sigma}$ that can be used in practice. They also proposed an extension of their testing procedure to known arbitrary covariance structures between features, still assuming i.i.d. observations. However, the estimation of the covariance between features remained unaddressed. 

Despite the notable contribution of \cite{Gao}, model \eqref{model_gao} is somewhat limited in view of more complex applications. In real problems, features describing observations are unlikely to be independent with identical variance, but rather present more general covariance structures $\mathbf{\Sigma}$. In the same way, observations may present non-negligible dependence structures when, for instance, they can be drawn from time series models or simulated with physical models involving time evolution. Note that ignoring dependence between features and observations implies the loss of selective type I error control. This can be illustrated by a simple simulation scenario based on matrix normal samples with non-diagonal covariance matrices, accounting for the dependence structures between observations and/or features. If model assumptions are not satisfied, the approach presented in \cite{Gao} does not control the selective type I error, as illustrated in Figure~\ref{fig:ignoredep} by the fact that the distribution of the corresponding $p$-values is above the diagonal (which corresponds to uniform $p$-values). This deviation from uniformity increases with the dimension of the feature space. Details about the corresponding simulation are given in Appendix~\ref{sec:ignoredep}.

\begin{figure}
    \centering
    \includegraphics[width=0.85\textwidth]{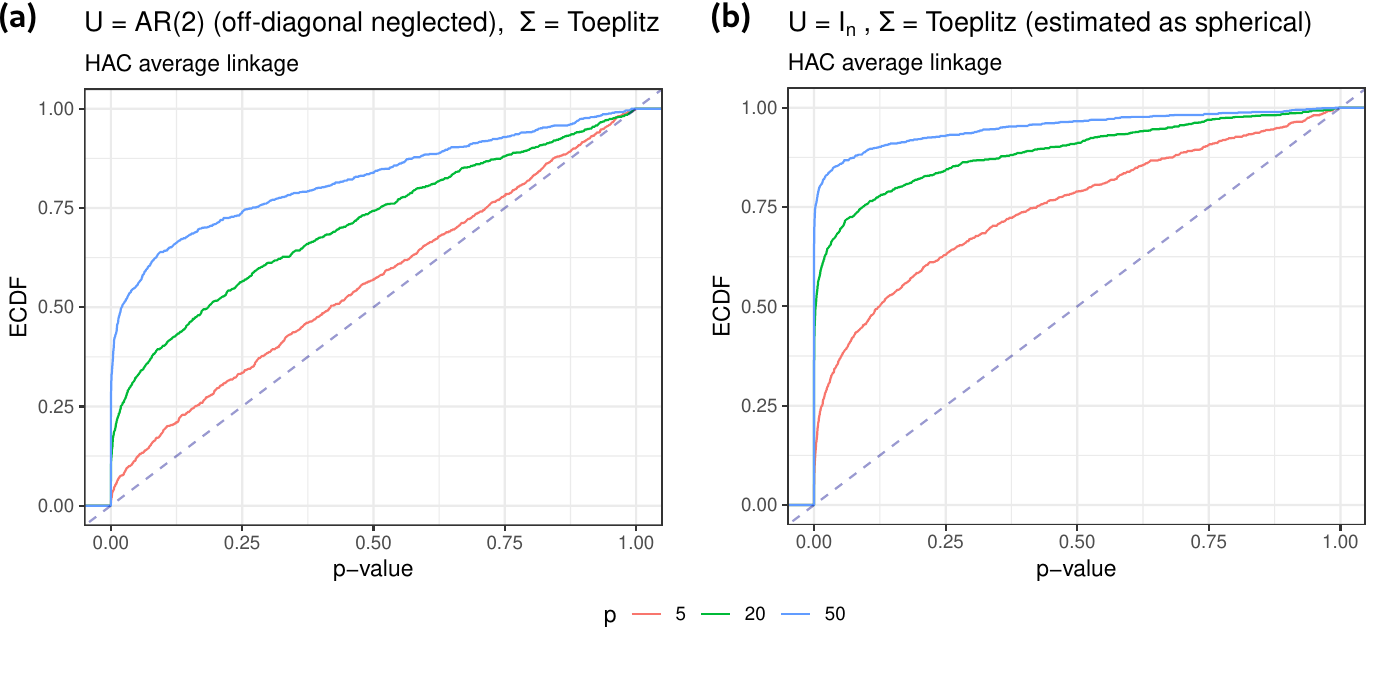}
    \caption{Empirical cumulative distribution functions (ECDF) of $p$-values defined in \cite{Gao} testing for the difference in cluster means after performing a hierarchical clustering algorithm (HAC) with average linkage. The ECDF were computed from $M=2000$ realizations of a matrix normal model with $\boldsymbol{\mu}=\mathbf{0}_{n\times p}$ and non-diagonal $\mathbf{U}$ and $\boldsymbol{\Sigma}$. For each realization, the test compared the means of two randomly selected clusters after setting the HAC to choose three clusters. We set $n=100$ and $p\in\lbrace 5,20,50\rbrace$. In (a), dependence between observations is ignored. In (b), the covariance between features is assumed to be spherical to allow its estimation using the approach in~\cite{Gao}.}
    \label{fig:ignoredep}
\end{figure}

The practical motivation of the present work is to perform inference after clustering protein conformations. Protein structures are non-static and their conformational variability is essential to understand the relationship between sequence, structural properties and function~\cite{Kessel2018}. Due to the high complexity of the conformational space, clustering techniques have emerged as powerful tools to characterize the structural variability of proteins, by extracting families of representative states~\cite{engens, Appadurai2022, Shao2007, Pearce2021}. Usually, Euclidean distances between pairs of amino acids are considered as $p$-dimensional descriptors of protein conformations~\cite{engens, Lazar, CamachoZarco2020}. These distances are highly correlated and hardly match the model~\eqref{model_gao}. Moreover, protein data is often simulated with Molecular Dynamics approaches that simulate the time-evolution of the protein according to physical models~\cite{allen}. In that case, independence between observations cannot be assumed. 

Accordingly, our aim is to go one step further and extend the framework introduced in~\cite{Gao} to a more general setting where arbitrary dependence structures between both observations and features are admitted, allowing for the estimation of one of them. We present a generalization of~\cite{Gao} where the model~\eqref{model_gao} is extended to
\begin{equation}\label{model}\tag{gen-MN}
    \mathbf{X}\sim\mathcal{MN}_{n\times p}(\boldsymbol\mu, \mathbf{U},\mathbf{\Sigma}),
\end{equation}
where $\boldsymbol\mu\in\mathcal{M}_{n\times p}(\mathbb{R})$, $\mathbf{U}\in\mathcal{M}_{n\times n}(\mathbb{R})$ and $\mathbf{\Sigma}\in\mathcal{M}_{p\times p}(\mathbb{R})$. Our techniques follow the same reasoning steps as the ones in~\cite{Gao}, establishing sufficient conditions that allow an extension to \eqref{model}.

The reader might wonder whether it is necessary to develop a new framework for \eqref{model} given that the matrix normal data can be whitened to fit into the Gao \textit{et al.} model. Indeed, as we have
\begin{equation*}
    \mathbf{X}\sim\mathcal{MN}_{n\times p}(\boldsymbol\mu, \mathbf{U},\mathbf{\Sigma})\Leftrightarrow \textrm{vec}(\mathbf{X})\sim\mathcal{N}_{np}\left(\textrm{vec}(\boldsymbol{\mu}), \mathbf{\Sigma}\otimes\mathbf{U}\right),
\end{equation*}
the transformed random vector $\left(\mathbf{\Sigma}\otimes\mathbf{U}\right)^{-\frac{1}{2}}\textrm{vec}(\mathbf{X})$ has covariance matrix $\mathbb{I}_{np}$ and can be de-vectorized to fit~\eqref{model_gao}. However, clustering the original and whitened data often leads to different partitions, and thus to different null hypotheses. In some cases, de-correlating the observations and features of $\mathbf{X}$ might yield a misleading impression of the underlying class structure. This is illustrated in Figure~\ref{fig:whitening}, where we show that whitening a sample drawn from~\eqref{model} and performing a selective test defined for~\eqref{model_gao} might substantially alter the significance of the differences between cluster means, as well as the overall clustering partition. Details on this numerical analysis are provided in Appendix~\ref{sec:whitening}. Note also that whitening a $n\times p$ matrix normal sample involves the inversion of a $np\times np$ matrix. These considerations, together with the unsuitability of the whitening approach when any of the covariance matrices is unknown, justifies the need of developing a new framework for the general model~\eqref{model}.

\begin{figure}[t]
    \centering
    \includegraphics[width=0.95\textwidth]{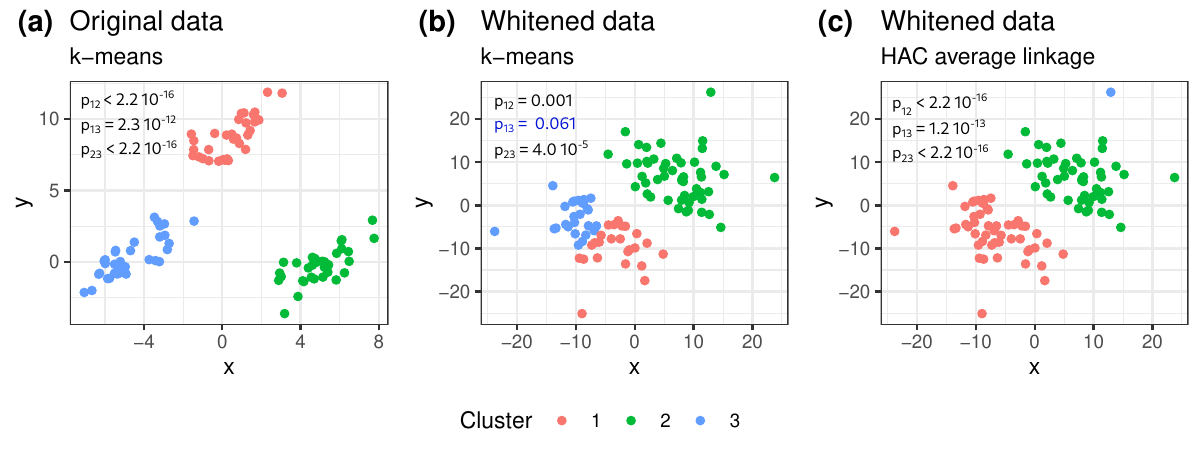}
    \caption{(a): Sample drawn from \eqref{model} with $n=100$, $p=2$ and non-diagonal covariance matrices $\mathbf{U}$ and $\mathbf{\Sigma}$. The mean matrix $\boldsymbol{\mu}$ is divided into three clusters. Observations are classified into three groups by the $k$-means algorithm and the difference between cluster means is tested using the approach proposed in this work. Classification using hierarchical agglomerative clustering (HAC) with average linkage yielded the same partition. (b): Data in (a) whitened and classified into three groups by the $k$-means algorithm. The differences between cluster means are tested assuming \eqref{model_gao} and using the approach presented in~\cite{chen2022selective}. (c): Data in (a) whitened and classified into three groups using HAC with average linkage. The differences between cluster means are tested assuming~\eqref{model_gao} and using the approach proposed in~\cite{Gao}. In all panels, $p_{ij}$ denotes the $p$-value for the difference between the means of clusters $i$ and $j$, for $i,j=1,2,3$.}
    \label{fig:whitening}
\end{figure}

The paper is organized as follows:
\begin{itemize}
    \item Section~\ref{sec:sel_clustering} presents our extension of~\cite{Gao} to the general model \eqref{model} when both covariance matrices are known.
    \item In Section~\ref{sec:unknown_sigma}, we explore the scenarios that allow the asymptotic over-estimation of either $\mathbf{U}$ or $\mathbf{\Sigma}$ while respecting the asymptotic control of the selective type I error. We provide an estimator that can be used in several common practical scenarios.
    \item Section~\ref{sec:numerical_experiments} illustrates all the results through numerical experiments on synthetic data, and evaluates the robustness of the presented approach to model misspecification. Finally, Section~\ref{sec:proteins} shows how this theory can be applied to perform inference after clustering protein structures. 
\end{itemize}  

\section{Selective inference for clustering under general dependence}\label{sec:sel_clustering}

In \cite{Gao}, the authors consider the problem of selective inference after hierarchical clustering in the case of independent observations and features (with an extension to arbitrary known dependence between features). Here, we aim to extend the method to allow for general dependence structures between both observations and features. We consider $n$ observations of $p$ features drawn from the matrix normal distribution \eqref{model}, where $\mathbf{U}$ and $\mathbf{\Sigma}$ are required to be positive definite. Each row of $\mathbf{X}$ is a vector of features in $\mathbb{R}^p$. The dependence between such features is given by $\mathbf{\Sigma}$, and $\mathbf{U}$ encodes the dependence between observations. If observations are independent with unit variance, we have $\mathbf{U}=\mathbf{I}_n$, and if features are independent with equal variance we can write $\mathbf{\Sigma}=\sigma^2\mathbf{I}_p$ for a given $\sigma>0$. These two assumptions define the model in \cite{Gao}, which we aim to extend to the most general $\mathbf{U}$ and $\mathbf{\Sigma}$.

\subsection{Problem setting and Gao \textit{et al.}'s approach}\label{sec:gao_setting}

Let us first recall the setting originally introduced in \cite{Gao}. We will denote by $X_i$ (resp. $\mu_i$) the $i$-th row of $\mathbf{X}$ (resp. $\boldsymbol\mu$) and, for a group of observations $\mathcal{G}\subseteq[n]=\lbrace 1,\ldots,n\rbrace$, $X_\mathcal{G}$ will denote the submatrix of $\mathbf{X}$ with rows $X_i$ for $i\in\mathcal{G}$. We also consider the mean of $\mathcal{G}$ in $\mathbf{X}$ and its empirical counterpart, denoted respectively by
\begin{equation}\label{empirical_mean}
    \bar{\mu}_\mathcal{G}=\frac{1}{\abs{\mathcal{G}}}\sum_{i\in\mathcal{G}}\mu_i
\quad\textrm{and}\quad
    \bar{X}_\mathcal{G}=\frac{1}{\abs{\mathcal{G}}}\sum_{i\in\mathcal{G}}X_i.
\end{equation}
Letting
\begin{equation}
   \mathcal{C}_{[n]}=\lbrace (\mathcal{G}_1,\mathcal{G}_2),\, \mathcal{G}_1,\mathcal{G}_2 \subset [n]\,:\,\mathcal{G}_1\cap\mathcal{G}_2=\emptyset\rbrace,
\end{equation}
be the set of all pairs of non-overlapping groups of observations, for any  $(\mathcal{G}_1,\mathcal{G}_2)\in\mathcal{C}_{[n]}$ we can define the column vector $\nu_{\mathcal{G}_1,\mathcal{G}_2}$ having as components
\begin{equation}
    [\nu_{\mathcal{G}_1,\mathcal{G}_2}]_i=\mathds{1}\lbrace i\in\mathcal{G}_1\rbrace/|\mathcal{G}_1|-
    \mathds{1}\lbrace i\in\mathcal{G}_2\rbrace/|\mathcal{G}_2|,
\end{equation}
for $i\in[n]$. This allows the difference between the (empirical) group means to be written compactly as
\begin{equation}\label{nu}
    \bar{\mu}_{\mathcal{G}_1}-\bar{\mu}_{\mathcal{G}_2}=\mathbf{\boldsymbol{\mu}}^T\nu_{\mathcal{G}_1,\mathcal{G}_2},\quad\textrm{and}\quad \bar{X}_{\mathcal{G}_1}-\bar{X}_{\mathcal{G}_2}=\mathbf{X}^T\nu_{\mathcal{G}_1,\mathcal{G}_2}.
\end{equation}
For the sake of a clearer notation, we will simply write $\nu=\nu_{\mG_1,\mG_2}$ when the context is clear. Let $\mathcal{C}$ be a clustering algorithm, $\mathbf{x}$ a realization of $\mathbf{X}$ and $\mG_1,\mG_2$ an arbitrary pair of clusters in $\mathcal{C}(\mathbf{x})$. The goal of post-clustering inference is to assess the null hypothesis
\begin{equation}\label{h0}\tag{H0}
H_0^{\lbrace \mG_1,\mG_2\rbrace}:\mathbf{\boldsymbol{\mu}}^T\nu_{\mG_1,\mG_2}=0,
\end{equation}
by controlling the \textit{selective type I error for clustering} at level $\alpha$, i.e. by ensuring that
\begin{equation}\label{sel_typeI}
    \mathbb{P}_{H_0^{\lbrace\mG_1,\mG_2\rbrace}}\left(\textrm{reject }H_0^{\lbrace \mG_1,\mG_2\rbrace}\textrm{
     based on }\mathbf{X}\textrm{ at level }\alpha\,\,\biggr\rvert\,\,\mG_1,\mG_2\in\mathcal{C}(\mathbf{X})\right)\leq \alpha\quad \forall\,\alpha\in(0,1).
\end{equation}
If the inequality in the previous equation can be replaced by an equality, we will say that the selective type I error is controlled \textit{exactly at level $\alpha$}. The ideal scenario to define a $p$-value for \eqref{h0} satisfying \eqref{sel_typeI} would be to only condition on the event $\lbrace \mG_1,\mG_2\in\mathcal{C}(\mathbf{X})\rbrace$, which is the broader conditioning set that allows selective type I error control. However, making the $p$-value analytically tractable often needs the refinement of the conditioning set by adding extra technical events (see also Appendix~\ref{sec:finer_cond}). 

The idea in~\cite{Gao} is to decompose $\mathbf{X}$ using the projection onto the orthogonal complement of $\nu$, that is, $\boldsymbol{\pi}_\nu^{\perp}=\mathbf{I}_n-\nu\nu^T/\norm{\nu}_2^2$. This naturally brings out the difference between empirical cluster means $\mathbf{X}^T\nu$, which can be used as a test statistic to evaluate~\eqref{h0}:
\begin{equation}\label{decomposition_gao}
    \mathbf{X}=\boldsymbol\pi_\nu^\perp\mathbf{X}+(\mathbf{I}_n-\boldsymbol\pi_\nu^\perp\mathbf{X})=\boldsymbol\pi_\nu^\perp\mathbf{X}+\left(\frac{\norm{\mathbf{X}^T\nu}_2}{\norm{\nu}_2^2}\right)\nu\,\mathrm{dir}(\mathbf{X}^T\nu)^T,
\end{equation}
where $\mathrm{dir}(v)=v/\norm{v}_2\mathds{1}\lbrace v\neq 0\rbrace$ for all $v\in\mathbb{R}^p$. The previous decomposition depends on the quantities $\boldsymbol{\pi}_\nu^\perp\mathbf{X}$ and $\mathrm{dir}(\mathbf{X}^T\nu)$, whose null distributions remain unknown. As a consequence, the authors in~\cite{Gao} condition on their values for a realization $\mathbf{x}$ of $\mathbf{X}$, defining the following quantity:
\begin{align}\label{pvalue_gao}
    p(\mathbf{x};\lbrace \mG_1,\mG_2\rbrace)=\mathbb{P}_{H_0^{\lbrace \mG_1,\mG_2\rbrace}}
    \Bigl(
    \norm{\mathbf{X}^T\nu}_2\geq \norm{\mathbf{x}^T\nu}_2\,\, \Bigr\lvert \,\, \mG_1,\mG_2\in\mathcal{C}(\mathbf{X}), \nonumber\\ \boldsymbol\pi_{\nu}^\perp \mathbf{X}=\boldsymbol\pi_{\nu}^\perp \mathbf{x} \, ,\, \mathrm{dir}\left(\mathbf{X}^T\nu\right)=\mathrm{dir}\left(\mathbf{x}^T\nu\right)
     \Bigr),\tag{p-GBW}
\end{align}
as a $p$-value for~\eqref{h0}~\cite[Theorem 1]{Gao}. 

The key challenge in proposing~\eqref{pvalue_gao} is finding an efficient characterization suitable for practical application. The idea in~\cite{Gao} involves two steps. The first is the definition of a test statistic based on the norm induced by the null covariance matrix of $\mathbf{X}^T\nu$ (up to a positive multiplicative factor). More precisely, if $\mathbf{A}$ is the covariance matrix of a non-degenerated, centered $p$-dimensional Gaussian vector $y$, then $\norm{y}_{\mathbf{A}}^2 = y^T \mathbf{A}^{-1} y$ follows a $\chi^2_p$ distribution. This implies that $\norm{\mathbf{X}^T \nu}_2$ follows a $\sigma\norm{\nu}_2\cdot\chi_p$ distribution under~\eqref{h0}, thereby justifying the choice of the $\ell^2$-norm. The second step is to show that $\norm{\mathbf{X}^T\nu}_2$ is independent of both the direction and the projection in~\eqref{pvalue_gao}. Consequently, the $p$-value~\eqref{pvalue_gao} can be expressed in terms of a $\chi_p$ distribution truncated to a set $\hat{\mathcal{S}}$ that accounts for the event $\lbrace\mathcal{C}_1,\mathcal{C}_2\in\mathcal{C}(\mathbf{X})\rbrace$. If $\mathcal{C}$ is a hierarchical clustering algorithm, the set $\hat{\mathcal{S}}$ -and thus~\eqref{pvalue_gao}-  can be efficiently computed for several types of linkages. Otherwise, it can be approximated with a Monte Carlo procedure. 

\subsection{Extension to the general matrix normal model}

\subsubsection{Feasibility of a straightforward extension of Gao \textit{et al.}}

Here, we aim at extending \eqref{pvalue_gao} for the general model \eqref{model}, following the same strategy to ensure the tractability of the $p$-value. Noticing that, under~\eqref{h0}, $\mathbf{X}^T\nu\sim\mathcal{N}_p(0,\mathbf{V}_{\mG_1,\mG_2})$, where
\begin{equation}\label{def_V}
    \mathbf{V}_{\mathcal{G}_1,\mathcal{G}_2} = \nu^T \mathbf{U} \nu\mathbf{\Sigma},
\end{equation}     
a natural extension corresponds to replace $||\cdot||_2$ by the more general norm 
\begin{equation}\label{norm_V}
\norm{v}_{\mathbf{V}_{\mathcal{G}_1,\mathcal{G}_2}}=\sqrt{v^T\mathbf{V}_{\mathcal{G}_1,\mathcal{G}_2}^{-1}v},\quad\forall\,v\in\mathbb{R}^p,
\end{equation}
which satisfies $||\mathbf{X}^T\nu||_{\mathbf{V}_{\mG_1,\mG_2}}\sim \chi_p$ under the null. This choice leads us to consider the quantity
\begin{align}\label{pvalue_V}
    p_{\mathbf{V}_{\mG_1,\mG_2}}
    (\mathbf{x};\lbrace \mG_1,\mG_2\rbrace)=\mathbb{P}_{H_0^{\lbrace \mG_1,\mG_2 \rbrace}}\Bigl(\norm{\diffXclst}_{\mathbf{V}_{\mG_1,\mG_2}}\geq \norm{\diffxclst}_{\mathbf{V}_{\mG_1,\mG_2}}\,\, \Bigr\rvert \,\, \mG_1,\mG_2\in\mathcal{C}(\mathbf{X}), \nonumber\\ 
    \boldsymbol\pi_{\nu}^\perp \mathbf{X}=\boldsymbol\pi_{\nu}^\perp \mathbf{x} \, ,\, \mathrm{dir}_{\mathbf{V}_{\mG_1,\mG_2}}\left(\diffXclst\right)=\mathrm{dir}_{\mathbf{V}_{\mG_1,\mG_2}}\left(\diffxclst\right)\Bigr),\tag{p-gen}
\end{align}
as a candidate $p$-value to extend Theorem 1 in~\cite{Gao}, where $\mathrm{dir}_{\mathbf{V}_{\mG_1,\mG_2}}(v)=v/\norm{v}_{\mathbf{V}_{\mG_1,\mG_2}}\mathds{1}\lbrace v\neq 0\rbrace$ for all $v\in\mathbb{R}^p$. 

A straightforward generalization of~\cite{Gao} to~\eqref{model} can be obtained when the test statistic is independent of the projection and direction in~\eqref{pvalue_V}, using the same argument as in the second step of the approach of~\cite{Gao}. However, the following result shows that the independence on the extra conditioning events holds if and only if the norm~\eqref{norm_V} is chosen to define the test statistic and $\mathbf{U}$ belongs to the class of positive definite compound symmetry matrices:
\begin{equation}
    \mathcal{CS}(n)=\left\lbrace (a-b)\mathbf{I}_n+b\mathbf{1}_{n\times n }\,:\,a\geq 0,\,-\frac{a}{n-1}< b <a \right\rbrace,
\end{equation}
where $\mathbf{1}_{n\times n}$ is a $n\times n$ matrix of ones. Note that $\mathcal{CS}(n)$ is the set of covariance matrices of the vectors $(y_1+\epsilon,\ldots,y_n+\epsilon)$, where the $y_i$ are centered i.i.d. Gaussian variables and $\epsilon$ is a centered noise independent of the $y_i$.

\begin{prop}\label{prop_conditions}  Let $\mathcal{C}$ be a clustering algorithm and $\mathbf{x}$ a realization of $\mathbf{X}\sim\mathcal{MN}_{n\times p}(\boldsymbol{\mu},\mathbf{U},\mathbf{\Sigma})$. For any $p\times p$ symmetric positive definite matrix $\mathbf{A}$, let $\mathrm{dir}_{\mathbf{A}}(v)=v/\norm{v}_{\mathbf{A}}\mathds{1}\lbrace v\neq 0\rbrace$ for all $v\in\mathbb{R}^p$. Then,
\begin{enumerate}
    \item[$(i)$] $\mathbf{U}\in\mathcal{CS}(n)\Leftrightarrow \mathbf{X}^T\nu_{\mG_1,\mG_2}\independent\boldsymbol{\pi}^\perp_{\nu_{\mG_1,\mG_2}}\mathbf{X}$ for all $(\mG_1,\mG_2)\in\mathcal{C}_{[n]}$,
    \item[$(ii)$] For any $(\mG_1,\mG_2)\in\mathcal{C}_{[n]}$, $\mathbf{A}=c\mathbf{V}_{\mG_1,\mG_2}$ for some $c>0$ $\overset{H_0^{\lbrace \mG_1,\mG_2 \rbrace}}{\Leftrightarrow} ||\mathbf{X}^T\nu_{\mG_1,\mG_2}||_{\mathbf{A}}\independent \mathrm{dir}_{\mathbf{A}}\left(\mathbf{X}^T\nu_{\mG_1,\mG_2}\right)$.
\end{enumerate}
\end{prop}
The previous result is proved in Appendix~\ref{proofs_1_lemma}. 
The first equivalence is established by showing that both conditions are simultaneously equivalent to $\nucl{\mG_1}{\mG_2}$ being an eigenvector of $\bU$. The second follows from the fact that a Gaussian vector is independent of its direction if and only if it is centered with spherical covariance, a condition that we rewrite in terms of the matrix $\bA$.

Proposition~\ref{prop_conditions} shows that the choice of the norm~\eqref{norm_V} to define the test statistic not only ensures the tractability of its null distribution but also its independence with respect to the direction in~\eqref{pvalue_V}. Furthermore, the independence $\mathbf{X}^T\nu\independent \boldsymbol{\pi}_\nu^{\perp}\mathbf{X}$ is equivalent to $\mathbf{U}\in\mathcal{CS}(n)$. In other words, the direct extension of the strategy in~\cite{Gao} to the general model~\eqref{model} imposes a compound symmetry constraint on the dependence between observations. We develop the framework $\mathbf{U}\in\mathcal{CS}(n)$ in Section~\ref{sec:CS_known_U}. In Section~\ref{sec:general_known_U}, we explore the extension of the same strategy to arbitrary $\mathbf{U}$, focusing on the characterization of quantities of the form~\eqref{pvalue_gao} when the extra conditioning events are not independent of the test statistic.

\subsubsection{Compound symmetry dependence between observations}\label{sec:CS_known_U}

If the dependence between observations $\mathbf{U}$ has a compound symmetry structure, the quantity~\eqref{pvalue_V} can be efficiently written in terms of a truncated $\chi_p$ distribution, and used as a $p$-value for~\eqref{h0}. This is stated in the next result, that extends Theorem 1 in~\cite{Gao}. 

\begin{theorem}\label{th:pvalue_V} Let $\mathcal{C}$ be a clustering algorithm and $\mathbf{x}$ a realization of $\mathbf{X}\sim\mathcal{MN}_{n\times p}(\boldsymbol{\mu},\mathbf{U},\mathbf{\Sigma})$ with $\mathbf{U}\in\mathcal{CS}(n)$. Then,
\begin{equation}\label{pvalue_F}\tag{p-tract}
    p_{\mathbf{V}_{\mathcal{G}_1,\mathcal{G}_2}}(\mathbf{x};\lbrace \mathcal{G}_1,\mathcal{G}_2\rbrace) = 1-\mathbb{F}_p\left(\norm{\mathbf{x}^T\nu}_{\mathbf{V}_{\mathcal{G}_1,\mathcal{G}_2}},\,\mathcal{S}_{\mathbf{V}_{\mathcal{G}_1,\mathcal{G}_2}}(\mathbf{x};\lbrace \mathcal{G}_1,\mathcal{G}_2\rbrace)\right),
\end{equation}
where $\mathbb{F}_p(t,\mathcal{S})$ is the cumulative distribution function of a $\chi_p$ random variable truncated to the set $\mathcal{S}$ and
\begin{equation}\label{set_S}
    \mathcal{S}_{\mathbf{V}_{\mathcal{G}_1,\mathcal{G}_2}}(\mathbf{x};\lbrace \mathcal{G}_1,\mathcal{G}_2\rbrace) = \left\lbrace \phi\geq 0\,:\,\mathcal{G}_1,\mathcal{G}_2\in\mathcal{C}\Bigl(
    \boldsymbol\pi_{\nu}^\perp \mathbf{x}+\phi\,\frac{\nu}{||\nu||_2^2}\,\mathrm{dir}_{\mathbf{V}_{\mathcal{G}_1,\mathcal{G}_2}}(\mathbf{x}^T\nu)
    \Bigr)\right\rbrace,
\end{equation}
for any $(\mG_1,\mG_2)\in\mathcal{C}_{[n]}$. Furthermore,  $p_{\mathbf{V}_{\mathcal{G}_1,\mathcal{G}_2}}(\mathbf{x};\lbrace \mathcal{G}_1,\mathcal{G}_2\rbrace)$ is a $p$-value for~\eqref{h0} that controls the selective type I error for clustering~\eqref{sel_typeI} exactly at level $\alpha$.
\end{theorem}

The proof of the previous result is presented in Appendix~\ref{proofs_1_CS}. One can easily verify that setting $\mathbf{U}=\mathbf{I}_n$ and $\mathbf{\Sigma}=\sigma^2\mathbf{I}_p$ in Theorem~\ref{th:pvalue_V} yields exactly Theorem 1 in \cite{Gao}. In this general version, the information about the variance has been extracted from the statistic null distribution, which now remains the same independently of $\mathbf{U},\mathbf{\Sigma}$, and moved it \textit{into} the test statistic itself by making it dependent on the scale matrices. More precisely, $||\mathbf{X}^T\nu||_{\mathbf{V}_{\mG_1,\mG_2}}$ is the \textit{Mahalanobis distance}~\cite{mahalanobis} between the empirical group means with respect to the null distribution of their difference. This distance generalizes to multiple dimensions the idea of quantifying how many standard deviations away a point is from the mean of its distribution, and therefore integrates the dependence structure between columns and rows in $\mathbf{X}$.

Following \eqref{pvalue_F}, the computation of \eqref{pvalue_V} only depends on the characterization of the one-dimensional set 
\begin{equation}\label{set_S_hat}
\hat{\mathcal{S}}_{\mathbf{V}_{\mathcal{G}_1,\mathcal{G}_2}}=\mathcal{S}_{\mathbf{V}_{\mathcal{G}_1,\mathcal{G}_2}}(\mathbf{x};\lbrace \mG_1,\mG_2\rbrace)=\left\lbrace \phi\geq 0\,:\,\mG_1,\mG_2\in\mathcal{C}\bigl(\mathbf{x}'_{\mathbf{V}_{\mathcal{G}_1,\mathcal{G}_2}}(\phi)\bigr)\right\rbrace, 
\end{equation}
where
\begin{equation}\label{perturbed_x}
    \mathbf{x}'_{\mathbf{V}_{\mathcal{G}_1,\mathcal{G}_2}}(\phi)= \boldsymbol\pi_{\nu}^\perp \mathbf{x}+\phi\,\frac{\nu}{||\nu||_2^2}\,\mathrm{dir}_{\mathbf{V}_{\mathcal{G}_1,\mathcal{G}_2}}(\mathbf{x}^T\nu).
\end{equation}

The data set \eqref{perturbed_x} is analogous to $\mathbf{x}'(\phi)$ in \cite[Equation (13)]{Gao} for the norm \eqref{norm_V}, and its interpretation is equivalent. Indeed, we can rewrite both $\mathbf{x}'(\phi)$ and \eqref{perturbed_x} as
\begin{align}
\mathbf{x}'(\phi)=
\mathbf{x} + \frac{\nu}{\norm{\nu}_2^2}\left(\phi-\norm{\mathbf{x}^T\nu}_2\right)\mathrm{dir}(\mathbf{x}^T\nu), \label{perturbed_x_bis}\\
\mathbf{x}'_{\mathbf{V}_{\mathcal{G}_1,\mathcal{G}_2}}(\phi)=
\mathbf{x} + \frac{\nu}{\norm{\nu}_2^2}\left(\phi-\norm{\mathbf{x}^T\nu}_{\mathbf{V}_{\mathcal{G}_1,\mathcal{G}_2}}\right)\mathrm{dir}_{\mathbf{V}_{\mathcal{G}_1,\mathcal{G}_2}}(\mathbf{x}^T\nu).\label{perturbed_2}
\end{align}
Consequently, we can interpret \eqref{perturbed_x} as a perturbed version of $\mathbf{x}$, but where the perturbation is based on the norm $\norm{\cdot}_{\mathbf{V}_{\mathcal{G}_1,\mathcal{G}_2}}$ instead of $\norm{\cdot}_2$. Thus, \eqref{set_S_hat} is the set of non-negative $\phi$ for which applying the clustering algorithm $\mathcal{C}$ to the perturbed data set $\mathbf{x}'_{\mathbf{V}_{\mathcal{G}_1,\mathcal{G}_2}}(\phi)$ yields $\mG_1$ and $\mG_2$. As shown in \cite{Gao}, the set 
\begin{equation}\label{set_S_Gao}
    \hat{\mathcal{S}}=\lbrace \phi \geq 0\,:\,\mG_1,\mG_2\in\mathcal{C}(\mathbf{x}'(\phi))\rbrace,
\end{equation}
can be explicitly characterized for hierarchical agglomerative clustering with several types of linkages. The next Lemma shows that we do not need to re-adapt the work in \cite{Gao} to the set \eqref{set_S_hat}, as its points are given by a scale transformation of the points in $\hat{\mathcal{S}}$.

\begin{lemma}\label{equivalence_sets} Let $\mathbf{x}$ be a realization of $\mathbf{X}$ and $\mG_1,\mG_2$ an arbitrary pair of clusters in $\mathcal{C}(\mathbf{x})$. Let $\hat{\mathcal{S}}$ denote the set \eqref{set_S_Gao} defined in \cite[Equation (12)]{Gao}. Then,
\begin{equation}
    \hat{\mathcal{S}}_{\mathbf{V}_{\mathcal{G}_1,\mathcal{G}_2}} = \frac{\norm{\mathbf{x}^T\nu}_{\mathbf{V}_{\mathcal{G}_1,\mathcal{G}_2}}}{\norm{\mathbf{x}^T\nu}_2}\,\hat{\mathcal{S}},
\end{equation}
where $\hat{\mathcal{S}}_{\mathbf{V}_{\mathcal{G}_1,\mathcal{G}_2}}$ is defined in \eqref{set_S_hat}.
\end{lemma}

Consequently, the work in \cite[Section 3]{Gao} can be applied here to characterize the set \eqref{set_S_hat} and, therefore, to compute the $p$-value defined in \eqref{pvalue_V}. An explicit characterization of \eqref{set_S_hat} is possible when $\mathcal{C}$ is a hierarchical clustering algorithm with squared Euclidean distance, along with either single linkage or a linkage satisfying a linear Lance-Williams update \cite[Equation 20]{Gao}, e.g. average, weighted, Ward, centroid or median linkage.
 The efficient computation of $\eqref{set_S_hat}$ can also be extended to $k$-means clustering using the work in~\cite{chen2022selective}, as shown in Appendix~\ref{sec:finer_cond}. Otherwise, the $p$-value \eqref{pvalue_V} can be approximated with a Monte Carlo procedure, adapting the importance sampling approach presented in \cite[Section 4.1]{Gao}. Following the same notation, we sample $\omega_1,\ldots,\omega_N\sim\mathcal{N}\left(\norm{\mathbf{x}^T\nu}_{\mathbf{V}_{\mathcal{G}_1,\mathcal{G}_2}},1\right)$ i.i.d. and approximate \eqref{pvalue_V} as
\begin{equation}\label{monte_carlo_pv}
    p_{\mathbf{V}_{\mathcal{G}_1,\mathcal{G}_2}}(\mathbf{x};\lbrace \mG_1,\mG_2\rbrace) \approx  \frac{\sum_{i=1}^N \,\pi_i\,\mathds{1}\left\lbrace \omega_i\geq \norm{\mathbf{x}^T\nu}_{\mathbf{V}_{\mathcal{G}_1,\mathcal{G}_2}}, \mG_1,\mG_2\in\mathcal{C}(\mathbf{x}'(\omega_i))\right\rbrace }{ \sum_{i=1}^N \,\pi_i\,\mathds{1}\left\lbrace \mG_1,\mG_2\in\mathcal{C}(\mathbf{x}'(\omega_i))\right\rbrace},
\end{equation}
for $\pi_i=f_1(\omega_i)/f_2(\omega_i)$, where $f_1$ is the density of a $\chi_p$ random variable, and $f_2$ is the density of a $\mathcal{N}(\norm{\mathbf{x}^T\nu}_{\mathbf{V}_{\mathcal{G}_1,\mathcal{G}_2}},1)$ random variable.

\subsubsection{General dependence between observations}\label{sec:general_known_U}

The goal of this section is to study whether a quantity of the form~\eqref{pvalue_gao} can be $(i)$ efficiently characterized in terms of a known distribution and/or $(ii)$ used as a $p$-value for~\eqref{h0}, in the case where the restriction $\mathbf{U}\in\mathcal{CS}(n)$ is not necessarily satisfied. Following from Proposition~\ref{prop_conditions}(i), this means that the projection $\boldsymbol{\pi}_\nu^\perp\mathbf{X}$ is not independent of $\mathbf{X}^T\nu$ in general and, therefore, that the distribution of interest for the definition of the test statistic is not that of $\mathbf{X}^T\nu$, but rather that of the conditioned vector:
\begin{equation}\label{cond_mean_vector}
    \bar{\bX}_\nu(\mathbf{x}):=\mathbf{X}^T\nu\,\bigl\vert\,\lbrace \boldsymbol{\pi}_\nu^\perp\mathbf{X}=\boldsymbol{\pi}_\nu^\perp\mathbf{x},\,\mathrm{dir}(\mathbf{X}^T\nu)=\pm\mathrm{dir}(\mathbf{x}^T\nu)\rbrace,\quad\textrm{for }\mathbf{x}\in\mathbb{R}^{n\times p}.
\end{equation}
Adding the $\pm$ symbol allows to express the conditioning set as a linear constraint, which is more suitable for Gaussian processes. Then, the distribution of $\bX^T\nu$ conditioned on the original conditioning set can be recovered by truncating the density function of $\bar{\bX}_\nu(\mathbf{x})$ to the half space $\{\by \in \R^{p}\,:\, \langle \by,\, \bx^T\nu \rangle\, \geq 0\}$. The null distribution of~\eqref{cond_mean_vector} is derived in the next result. In what follows, we will denote by $\mathbf{A}^\dagger$ the Moore-Penrose pseudo-inverse of a matrix $\mathbf{A}$.
\begin{theorem}\label{th:null_dist_XC} Let $\mathcal{C}$ be a clustering algorithm and $\mathbf{x}$ a realization of $\mathbf{X}\sim\mathcal{MN}_{n\times p}(\boldsymbol{\mu},\mathbf{U},\mathbf{\Sigma})$. Then, under~\eqref{h0},
\begin{equation}\label{null_dist_xcond}
    \bar{\bX}_\nu(\mathbf{x})\sim\mathcal{N}_p\left(0,\boldsymbol{\Gamma}_{\mathbf{x}}\right),
\end{equation}
with
\begin{equation}\label{gamma_x}
    \boldsymbol{\Gamma_{\mathbf{x}}} = (\mathbf{I}_p\otimes \nu^T)(\bSigma \otimes \bU-(\bSigma \otimes \bU) {\bA_\bx^T}({\bA_\bx}( \bSigma  \otimes \bU) {\bA_\bx^T})^\dagger{\bA_\bx}(\bSigma \otimes \bU))(\mathbf{I}_p\otimes \nu),
\end{equation}
where $\mathbf{A}_\mathbf{x}$ is a $2np\times p$ matrix given by:
\begin{equation}
{\bA_\bx} =
\begin{bNiceArray}{c}
\boldsymbol{\pi}_{\mathbf{x}_\nu}^\perp
(\bI_p \otimes \boldsymbol{\pi}_\nu) \\
\bI_p\otimes\boldsymbol{\pi}_\nu^\perp
\end{bNiceArray},     
\end{equation}
with $\boldsymbol{\pi}_{\nu}=\mathbf{I}_n-\boldsymbol{\pi}_\nu^\perp$, $\mathbf{x}_\nu=\mathrm{vec}(\boldsymbol{\pi}_{\nu}\mathbf{x})$ and $\boldsymbol{\pi}_{\mathbf{x}_\nu}^\perp=\mathbf{I}_{np}-\mathbf{x}_\nu^T\mathbf{x}_\nu/\norm{\mathbf{x}_\nu}_2^2$. 
\end{theorem}

The proof of Theorem~\ref{th:null_dist_XC}, presented in Appendix~\ref{proofs_1_general}, proceeds in two main steps. First, we express the conditioning set in~\eqref{cond_mean_vector} as a linear constraint, which allows us to apply Proposition 3.13 in~\cite{Eaton2007}, characterizing the distribution of Gaussian vectors $z$ conditioned to events of the form $\lbrace\bA z = y\rbrace$. The structure of the conditioned covariance and mean matrices motivates a detailed analysis of some specific matrix families (Lemma~\ref{lemma:proj:prod:spaces:general} and Corollary~\ref{cor:inclusion:prod:space}). The corresponding results allow us to prove that $\bar{\bX}_\nu(\bx)$ is centered under~\eqref{h0} for all $\bx\in\mathbb{R}^{n\times p}$.

Following from~\eqref{null_dist_xcond}, in order to define a quantity of the form~\eqref{pvalue_gao}, the same reasoning as in the previous sections leads us to consider the following norm based on the covariance matrix~\eqref{gamma_x}:
\begin{equation}\label{norm_gamma_x}
    \norm{v}_{\boldsymbol{\Gamma}_{\mathbf{x}}}=\sqrt{v^T\boldsymbol{\Gamma}_\bx^\dagger v},\quad\forall\,v\in\mathbb{R}^p,
\end{equation}
for any $\mathbf{x}\in\mathbb{R}^{n\times p}$, where we have considered the generalized inverse of~\eqref{gamma_x} as this matrix is not full-rank. This leads us to define the quantity:
\begin{align}\label{pvalue_gamma}
    p_{\bGamma}(\mathbf{x};\lbrace \mG_1,\mG_2\rbrace)=\mathbb{P}_{H_0^{\lbrace \mG_1,\mG_2 \rbrace}}\Bigl(\norm{\diffXclst}_{\bGamma_\bx}\geq \norm{\diffxclst}_{\bGamma_\bx}\,\, \Bigr\rvert \,\, \mG_1,\mG_2\in\mathcal{C}(\mathbf{X}), \nonumber\\ \boldsymbol\pi_{\nu}^\perp \mathbf{X}=\boldsymbol\pi_{\nu}^\perp \mathbf{x} \, ,\, \mathrm{dir}\left(\diffXclst\right)=\pm\mathrm{dir}\left(\diffxclst\right)\Bigr),\tag{p-Gamma}
\end{align}
as a candidate $p$-value for~\eqref{h0} under~\eqref{model}. The previous quantity has an unusual form for a $p$-value, since the test statistic $\norm{\diffXclst}_{\bGamma_\bx}$ depends on the realization $\mathbf{x}$. However, the following result shows that its distribution under~\eqref{h0} \textit{is almost surely independent of} $\mathbf{x}$, yielding an efficient characterization of~\eqref{pvalue_gamma}. Its proof is presented in Appendix~\ref{proofs_1_general}.

\begin{prop}\label{prop:gamma_chi} In the conditions of Theorem~\ref{th:null_dist_XC}, the quantity $\norm{\bar{\bX}_\nu(\mathbf{x})}_{\bGamma_\bx}$ follows $\mathbf{x}$-a.s. a $\chi_1$ distribution under~\eqref{h0}.  Moreover, the quantity~\eqref{pvalue_gamma} can be written as:
\begin{equation}\label{pvalue_gamma_chi1}
    p_{\bGamma}(\mathbf{x};\lbrace \mG_1,\mG_2\rbrace) = 1-\mathbb{F}_1\left(\norm{\mathbf{x}^T\nu}_{\bGamma_\bx},\,\mathcal{S}_{\bGamma_\bx}(\mathbf{x};\lbrace \mathcal{G}_1,\mathcal{G}_2\rbrace)\right),
\end{equation}
for any $(\mG_1,\mG_2)\in\mathcal{C}_{[n]}$, where $\mathbb{F}_1(t,\mathcal{S})$ is the cumulative distribution function of a $\chi_1$ random variable truncated to the set $\mathcal{S}$ and
\begin{equation}\label{S_gamma}
 \mathcal{S}_{\bGamma_\bx}(\mathbf{x};\lbrace \mathcal{G}_1,\mathcal{G}_2\rbrace)) = \frac{\norm{\mathbf{x}^T\nu}_{\bGamma_\bx}}{\norm{\mathbf{x}^T\nu}_2}\,\lbrace \phi \geq 0\,:\,\mG_1,\mG_2\in\mathcal{C}(\mathbf{x}'(\pm\phi))\rbrace,
\end{equation}
where $\mathbf{x}'(\phi)$ is the perturbed data set defined in~\eqref{perturbed_2}.
\end{prop}

The previous result allows the efficient computation of~\eqref{pvalue_gamma} in terms of a $\chi_1$ distribution. Equation~\eqref{pvalue_gamma_chi1} is the counterpart of~\cite[Equation (9)]{Gao} and~\eqref{pvalue_V} for the most general case, where~\eqref{model} holds with arbitrary $\bU$. However, although the quantity~\eqref{pvalue_gamma} is the natural extension of~\eqref{pvalue_gao} in this context, and can be efficiently characterized via~\eqref{pvalue_gamma_chi1}, assessing whether it controls the selective type I error is a challenging problem. More precisely, the null distribution of the conditioned vector~\eqref{cond_mean_vector} depends on the realization $\bx$ and, consequently, the norm required to ensure that the test statistic is distribution-free is also dependent on $\bx$. As a consequence, \eqref{pvalue_gamma} compares two quantities that behave differently under~\eqref{h0}. Indeed, in order to assess whether~\eqref{sel_typeI} is satisfied, it is necessary to understand the behavior of the null distribution of $\norm{\bar{\bX}_\nu(\bx)}_{\Gamma_\bX}^2=\bar{\bX}_\nu(\bx)^T\Gamma_{\bX}^\dagger\bar{\bX}_\nu(\bx)$, which is a nontrivial problem. Nevertheless, since Proposition~\ref{prop:gamma_chi} allows the computation of~\eqref{pvalue_gamma} in practice, we are able to illustrate numerically that the quantity~\eqref{pvalue_gamma_chi1} does not control the selective type I error for several $\bU\notin{CS}(n)$ structures. We present these simulations in Appendix~\ref{app:simulations_gamma}.

As we further discuss in Section~\ref{sec:discussion}, the analyses presented above suggest that defining a tractable $p$-value of the form~\eqref{pvalue_gao} that ensures the selective type I error control requires the conditioning on events that are \textit{independent} of the test statistic. Following from Proposition~\ref{prop_conditions}, this is ensured if and only if the covariance structure between observations has a compound symmetry structure and the norm~\eqref{norm_V} is used to the define the $p$-value.

\section{Unknown dependence structures}\label{sec:unknown_sigma}

The selective inference framework introduced for~\eqref{model} in Section~\ref{sec:sel_clustering} assumes that both scale matrices $\mathbf{U}$ and $\mathbf{\Sigma}$ are known, which is a quite unrealistic scenario. Under the independence assumption made in \cite{Gao}, where $\mathbf{\Sigma}=\sigma^2\mathbf{I}_p$ and $\mathbf{U}=\mathbf{I}_n$, the authors showed in Theorem 4 that over-estimating $\sigma$ yields asymptotic control of the selective type I error, and provided such an estimator $\hat{\sigma}$ that can be used in practice.

The simultaneous estimation of $\mathbf{U}$ and $\boldsymbol{\Sigma}$ from a single copy of $\mathbf{X}$ is a challenging task due to the intrinsic limitations of the matrix normal model. The non-identifiability of both matrices under~\eqref{model} makes their existing estimators interdependent. Besides, multiple realizations of $\mathbf{X}$ are needed to ensure their existence and uniqueness~\cite{MLE}. The same goes for the estimation of $\mathbf{U}\otimes\boldsymbol{\Sigma}$, that fully determines the covariance structure of $\mathbf{X}$~\cite{Soloveychik2016Jul,Drton2021Oct,Derksen2021, Drton2024Jan}, even when $\mathbf{U}$ is restricted to the class $\mathcal{CS}(n)$~\cite{Ahmed2015Oct}. All of this hinders the estimation of both covariance matrices from a single copy of $\mathbf{X}$. Furthermore, we not only require \textit{any} estimator of $\mathbf{U}$ and $\boldsymbol{\Sigma}$ but one which is compatible with the selective type I error control. Consequently, we opt to investigate the situation where only one of the scale matrices is known, and assess theoretical conditions that allow asymptotic control of the selective type I error when estimating the other one. We also provide an estimator that satisfies these conditions for some common dependence models.

Let us recall that, for the model \eqref{model}, we have
\begin{equation}\label{transposing_X}
     \mathbf{X}\sim\mathcal{M}\mathcal{N}_{n\times p}(\boldsymbol\mu,\mathbf{U},\mathbf{\Sigma})\Leftrightarrow \mathbf{X}^T\sim\mathcal{M}\mathcal{N}_{p\times n}(\boldsymbol\mu^T,\mathbf{\Sigma},\mathbf{U}).
\end{equation}
Therefore, the methods presented in this section can be equally applied to estimate $\mathbf{U}$ or $\mathbf{\Sigma}$ when the other is known, by transposing $\mathbf{X}$ if needed. From now on, we assume that the dependence structure between observations $\mathbf{U}$ is known, and study under which conditions we can suitably estimate $\mathbf{\Sigma}$. In Section~\ref{sec:estimation_CS}, we focus on the case where a computationally tractable $p$-value can be defined according to Theorem~\ref{th:pvalue_V}, assessing the applicability of~\eqref{pvalue_F} when $\boldsymbol{\Sigma}$ is estimated with $\mathbf{U}\in\mathcal{CS}(n)$. Since the robustness of~\eqref{pvalue_F} to $\mathbf{U} \notin\mathcal{CS}(n)$ will be numerically studied, in Section~\ref{sec:estimation_not_CS} we explore the theoretical guarantees that can be provided in that case regarding the estimation of $\boldsymbol{\Sigma}$.

\subsection{Compound symmetry covariance between observations}\label{sec:estimation_CS}

Let $\hat{\mathbf{\Sigma}}(\mathbf{x})$ be an estimate of $\mathbf{\Sigma}$ for a given realization $\mathbf{x}$ of $\mathbf{X}$. Following from Theorem~\ref{th:pvalue_V}, the $p$-value~\eqref{pvalue_V} has the closed form~\eqref{pvalue_F} if $\mathbf{U}\in\mathcal{CS}(n)$. In that case, the estimation of $\mathbf{\Sigma}$ comes down to studying under which conditions the $p$-value
\begin{equation}\label{hat_pvalue}\tag{hat-p-tract}
p_{\hat{\mathbf{V}}_{\mathcal{G}_1,\mathcal{G}_2}}(\mathbf{x};\lbrace \mathcal{G}_1,\mathcal{G}_2\rbrace)= 1-\mathbb{F}_p\left(\norm{\mathbf{x}^T\nu}_{\hat{\mathbf{V}}_{\mathcal{G}_1,\mathcal{G}_2}};\,\mathcal{S}_{\hat{\mathbf{V}}_{\mathcal{G}_1,\mathcal{G}_2}}(\mathbf{x};\lbrace \mathcal{G}_1,\mathcal{G}_2\rbrace)\right), 
\end{equation}
where $\hat{\mathbf{V}}_{\mathcal{G}_1,\mathcal{G}_2}=\nu^T\mathbf{U}\nu \hat{\mathbf{\Sigma}}(\mathbf{x})$, controls the selective type I error. Theorem~\ref{th:over_estimate} below generalizes Theorem~4 in \cite{Gao} for the estimation of $\mathbf{\Sigma}$ under the model \eqref{model} by relying on the Loewner partial order, defined below. The proof is included in Appendix~\ref{proofs_2}.

\begin{definition}[Definition 7.7.1 in \cite{horn2013matrix}]\label{def:loewner} For two square matrices of equal size $A,B$, we write $A\succeq B$ if and only if $A,B$ are Hermitian and $A-B$ is positive semidefinite. This binary relation between square matrices is called the Loewner partial order.
\end{definition}

 \begin{theorem}\label{th:over_estimate} For $n\in\mathbb{N}$, let $\mathbf{X}^{(n)}\sim\mathcal{M}\mathcal{N}_{n\times p}(\boldsymbol\mu^{(n)},\mathbf{U}^{(n)},\mathbf{\Sigma})$ with $\mathbf{U}^{(n)} = (a-b)\mathbf{I}_n+b\mathbf{1}_{n\times n}$ for some $a>b>0$. Let $\mathbf{x}^{(n)}$ be a realization of $\mathbf{X}^{(n)}$ and $\mG_1^{(n)},\mG_2^{(n)}$ a pair of clusters estimated from $\mathbf{x}^{(n)}$. If $\hat{\mathbf{\Sigma}}\left(\mathbf{X}^{(n)}\right)$ is a positive definite estimator of $\mathbf{\Sigma}$ such that
\begin{equation}\label{sigma_overestimate_condition}\tag{over-est}
    \underset{n\rightarrow\infty}{\lim}\mathbb{P}_{H_0^{\left\lbrace \mG_1^{(n)}, \mG_2^{(n)}\right\rbrace}}\left(\hat{\mathbf{\Sigma}}\Bigl(\mathbf{X}^{(n)}\Bigr)\succeq \mathbf{\Sigma}\,\bigg\rvert\,\mG_1^{(n)},\mG_2^{(n)}\in\mathcal{C}\Bigl(\mathbf{X}^{(n)}\Bigr)\right)=1,
\end{equation}
then,
\begin{equation}\label{superuniformity}
    \underset{n\rightarrow\infty}{\limsup}\,\mathbb{P}_{H_0^{\left\lbrace \mG_1^{(n)}, \mG_2^{(n)}\right\rbrace}}\left(p_{\hat{\mathbf{V}}_{\mG_1^{(n)},\mG_2^{(n)}}}\left(\mathbf{X}^{(n)};\left\lbrace \mG_1^{(n)},\mG_2^{(n)}\right\rbrace\right)\leq\alpha\,\bigg\rvert\,\mG_1^{(n)},\mG_2^{(n)}\in\mathcal{C}\Bigl(\mathbf{X}^{(n)}\Bigr)\right)\leq\alpha,
\end{equation}
for any $\alpha\in[0,1]$.
\end{theorem}

Note that the Loewner partial order is a natural extension to Hermitian matrices of the usual order in $\mathbb{R}$. If we replace $\mathbf{\Sigma}$ by $\sigma^2\mathbf{I}_p$ in Theorem~\ref{th:over_estimate}, the condition $\hat{\mathbf{\Sigma}}\succeq\mathbf{\Sigma}$ becomes $\hat{\sigma}\geq\sigma$, as in \cite[Theorem 4]{Gao}. We aim now at providing an estimator of $\mathbf{\Sigma}$ satisfying condition~\eqref{sigma_overestimate_condition}. The asymptotic properties of such an estimator strongly depend on the asymptotic dependence structure between observations, given by the sequence of matrices $\lbrace\mathbf{U}^{(n)}\rbrace_{n\in\mathbb{N}}$ of Theorem~\ref{th:over_estimate}. First, let us consider
\begin{equation}\label{estSigma}\tag{hat-Sigma}
    \hat{\mathbf{\Sigma}} = \hat{\mathbf{\Sigma}}\left(\mathbf{X}\right)=\frac{1}{n-1}\left(\mathbf{X}-\bar{\mathbf{X}}\right)^T\mathbf{U}^{-1}\left(\mathbf{X}-\bar{\mathbf{X}}\right),
\end{equation}
where $\bar{\mathbf{X}}$ is a $n\times p$ matrix having as rows the mean across rows of $\mathbf{X}$, i.e.
\begin{equation}
    \bar{\mathbf{X}} = \mathbf{1}_n \otimes \frac{1}{n}\sum_{k=1}^n X_{k},
\end{equation}
where $\mathbf{1}_n$ is a column $n$-vector of ones. Note that \eqref{estSigma} is constructed by first de-correlating the observations using $\mathbf{U}$, then subtracting off the column means and finally taking the sample covariance matrix. Following \cite[Corollary 2.3.10.2]{Gupta2018}, subtracting off the true mean matrix $\boldsymbol{\mu}$ instead of $\bar{\mathbf{X}}$ would lead to a consistent estimator without making any assumption on $\mathbf{U}$, as the rows of $\mathbf{U}^{-\frac{1}{2}}(\mathbf{X}-\boldsymbol{\mu})$ are $n$ i.i.d. copies of a $p$-dimensional centered Gaussian vector of covariance matrix $\mathbf{\Sigma}$. However, $\boldsymbol{\mu}$ needs to be considered unknown in the context of clustering analysis. Note also that the estimator $\hat{\mathbf{\Sigma}}$ is a positive definite matrix if the matrix $\mathbf{X}-\Bar{\mathbf{X}}$ has full rank. In order to ensure that \eqref{estSigma} satisfies condition~\eqref{sigma_overestimate_condition}, some additional assumptions regarding the asymptotic behavior of the matrices $\boldsymbol{\mu}^{(n)}$ are required.

\begin{assumption}[Assumptions 1 and 2 in \cite{Gao}]\label{as_1} For all $n\in\mathbb{N}$, there are exactly $K^*$ distinct mean vectors among the first $n$ observations, i.e.
\begin{equation}
 \left\lbrace \mu_i^{(n)}\right\rbrace_{i=1,\ldots,n}=\lbrace \theta_1,\ldots,\theta_{K^*}\rbrace.   
\end{equation}
Moreover, the proportion of the first $n$ observations that have mean vector $\theta_k$ converges to $\pi_k>0$, i.e.
\begin{equation}\label{convergence_indep}
    \underset{n\rightarrow\infty}{\lim}\frac{1}{n}\sum_{i=1}^n \mathds{1}\lbrace \mu_i^{(n)}=\theta_k\rbrace = \pi_k,
\end{equation}
for all $k\in\lbrace1,\ldots,K^{*}\rbrace$, where $\sum_{k=1}^{K^*}\pi_k=1$.
\end{assumption}
If observations are independent, Assumption~\ref{as_1} is the only requirement for \eqref{estSigma} to asymptotically over-estimate $\mathbf{\Sigma}$ in the sense of Theorem~\ref{th:over_estimate}. 
For non-diagonal $\mathbf{U}^{(n)}$, the following condition on $\lbrace\boldsymbol\mu^{(n)}\rbrace_{n\in\mathbb{N}}$ needs to be assumed.

\begin{assumption}\label{as_3} If $\mathbf{U}^{(n)}$ is non-diagonal for all $n\in\mathbb{N}$, for any $k,k'\in\lbrace 1,\ldots,K^*\rbrace$, the proportion of the first $n$ observations at distance $r\geq 1 $ in $\mathbf{X}^{(n)}$ having means $\theta_k$ and $\theta_{k'}$ converges, and its limit converges to $\pi_{k}\pi_{k'}$ when the lag $r$ tends to infinity. More precisely,
    \begin{equation}\label{mixing}
        \underset{n\rightarrow\infty}{\lim}\frac{1}{n}\sum_{i=1}^{n-r}\mathds{1}\lbrace\mu_i^{(n)}=\theta_k\rbrace\,\mathds{1}\lbrace\mu_{i+r}^{(n)}=\theta_{k'}\rbrace = \pi_{kk'}^r\underset{r\rightarrow\infty}{\longrightarrow}\pi_k\,\pi_{k'}.
    \end{equation}
\end{assumption}
Note that we are requiring the proportion of pairs of observations having a given a pair of means to approach the product of individual proportions \eqref{convergence_indep} when both observations are far away in $\mathbf{X}^{(n)}$. Assumption \ref{as_3} can be alternatively formulated in terms of strong mixing of measure-preserving dynamical systems~\cite[Chapter 20]{Klenke}. This is proved in Appendix~\ref{proofs_2}.

If $\mathbf{U}^{(n)}$ is compound symmetry for fixed $a>b>0$ and Assumptions \ref{as_1} and \ref{as_3} hold for a given sequence $\lbrace \boldsymbol\mu^{(n)}\rbrace_{n\in\mathbb{N}}$, the following result ensures that $\hat{\mathbf{\Sigma}}$ asymptotically over-estimates (in the sense of the Loewner partial order) the dependence structure $\mathbf{\Sigma}$ between features.

\begin{prop}\label{prop_overestimate} Let $\mathbf{X}^{(n)}\sim\mathcal{M}\mathcal{N}_{n\times p}(\boldsymbol\mu^{(n)},\mathbf{U}^{(n)},\mathbf{\Sigma})$, where $\mathbf{U}^{(n)}=(a-b)\mathbf{I}_n+b\mathbf{1}_{n\times n}$ for some $a>b>0$ and $\boldsymbol\mu^{(n)}$ satisfies Assumptions~\ref{as_1} and \ref{as_3} for some $K^*>1$. Let $\hat{\mathbf{\Sigma}}$ be the estimator defined in \eqref{estSigma}. Then,
\begin{equation}\label{asymp_overest_CS}
    \underset{n\rightarrow\infty}{\lim}\,\mathbb{P}\left(\hat{\mathbf{\Sigma}}\Bigl(\mathbf{X}^{(n)}\Bigr)\succeq \mathbf{\Sigma}\right)=1.
\end{equation}
\end{prop}

Finally, it suffices to estimate $\mathbf{\Sigma}$ using an independent and identically distributed copy of $\mathbf{X}^{(n)}$ to have \eqref{sigma_overestimate_condition} provided~\eqref{asymp_overest_CS} holds. Such a copy is sometimes available in practical applications, as the one we present in Section~\ref{sec:proteins}. Combining this observation with Proposition~\ref{prop_overestimate}, we obtain our final result:

\begin{prop}\label{prop_indcopy} Let $\mathbf{X}^{(n)}\sim\mathcal{M}\mathcal{N}_{n\times p}(\boldsymbol\mu^{(n)},\mathbf{U}^{(n)},\mathbf{\Sigma})$, where $\mathbf{U}^{(n)}=(a-b)\mathbf{I}_n+b\mathbf{1}_{n\times n}$ for some $a>b>0$ and $\boldsymbol\mu^{(n)}$ satisfies Assumptions~\ref{as_1} and \ref{as_3} for some $K^*>1$. Let $\mathbf{x}^{(n)}$ be a realization of $\mathbf{X}^{(n)}$ and $\mG_1^{(n)}$, $\mG_2^{(n)}$ a pair of clusters estimated from $\mathbf{x}^{(n)}$. Let $\mathbf{Y}^{(n)}$ be an independent and identically distributed copy of $\mathbf{X}^{(n)}$. Then, the estimator $\hat{\mathbf{\Sigma}}\left(\mathbf{Y}^{(n)}\right)$ defined in \eqref{estSigma} satisfies the conditions of Theorem~\ref{th:over_estimate}, i.e.
\begin{equation}
       \underset{n\rightarrow\infty}{\lim}\mathbb{P}_{H_0^{\left\lbrace \mG_1^{(n)}, \mG_2^{(n)}\right\rbrace}}\left(\hat{\mathbf{\Sigma}}\Bigl(\mathbf{Y}^{(n)}\Bigr)\succeq \mathbf{\Sigma}\,\bigg\rvert\,\mG_1^{(n)},\mG_2^{(n)}\in\mathcal{C}\Bigl(\mathbf{X}^{(n)}\Bigr)\right)=1.
\end{equation}
\end{prop}

Following from the previous result and from Theorem~\ref{th:over_estimate}, if Assumptions~\ref{as_1} and~\ref{as_3} hold and $\mathbf{U}\in\mathcal{CS}(n)$, selective type I error is asymptotically controlled when using~\eqref{estSigma} to estimate $\mathbf{\Sigma}$. This constitutes an extension of the over-estimation framework presented in~\cite{Gao}, which holds under model~\eqref{model_gao}, to the more general~\eqref{model} with compound symmetry dependence between observations.

\subsection{Arbitrary covariance between observations}\label{sec:estimation_not_CS}

In Section~\ref{sec:estimation_CS}, we proved that $p$-values~\eqref{hat_pvalue} are asymptotically super-uniform under~\eqref{h0} if, besides Assumptions~\ref{as_1} and~\ref{as_3}, the following conditions hold:
\begin{itemize}
    \item[(a)] The $p$-value~\eqref{pvalue_F} (for known $\mathbf{\Sigma}$) is uniformly distributed under~\eqref{h0} or, equivalently, $\mathbf{U}\in\mathcal{CS}(n)$,
    \item[(b)] The estimator~\eqref{estSigma} satisfies~\eqref{sigma_overestimate_condition}.
\end{itemize}

However, as it will be numerically illustrated in Section~\ref{sec:robustness}, the null uniformity of~\eqref{pvalue_F} is robust to $\mathbf{U}$ structures that do not fit in $\mathcal{CS}(n)$. Consequently, the null super-uniformity of~\eqref{hat_pvalue} will be robust to $\mathbf{U}\notin\mathcal{CS}(n)$ as long as $(b)$ is satisfied. In this section, we investigate the theoretical conditions that need to be imposed to an arbitrary sequence $\lbrace \mathbf{U}^{(n)}\rbrace_{n\in\mathbb{N}}$ so that the estimator~\eqref{estSigma} satisfies~\eqref{sigma_overestimate_condition}. To that end, besides Assumptions~\ref{as_1} and~\ref{as_3}, the quantities
\begin{equation}\label{convergence_U}
    \frac{1}{n}\sum_{l,s=1}^n \Bigl(U^{(n)}\Bigr)^{-1}_{ls}\,\mathds{1}\lbrace \mu_l^{(n)}=\theta_k\rbrace\,\mathds{1}\lbrace \mu_s^{(n)}=\theta_{k'}\rbrace
\end{equation}
are also required to converge. Furthermore, we need to know their limit explicitly to assess~\eqref{sigma_overestimate_condition}. Below, we state sufficient conditions on the sequence $\lbrace \mathbf{U}^{(n)}\rbrace_{n\in\mathbb{N}}$ that -together with Assumptions~\ref{as_1} and \ref{as_3}- ensure the convergence of \eqref{convergence_U} to a tractable limit. Note that these technical assumptions can be difficult to verify for a given model of dependence, and other unknown sufficient conditions might guarantee that \eqref{estSigma} asymptotically over-estimates $\mathbf{\Sigma}$. This point is investigated numerically in Section~\ref{sec:no_ad_U}.

\begin{assumption}\label{as_2} Let $\lbrace \mathbf{U}^{(n)}\rbrace_{n\in\mathbb{N}}$ be a sequence of real positive definite matrices, and let $\left(U^{(n)}\right)^{-1}_{ij}$ denote the $i,j$ entry of $\left(\mathbf{U}^{(n)}\right)^{-1}$ for any $n\in\mathbb{N}$. Then, every superdiagonal of $\left(\mathbf{U}^{(n)}\right)^{-1}$ defines asymptotically a convergent sequence, whose limits sum up to a real value. More precisely, for any $i\in\mathbb{N}$ and any $r\geq 0$,
\begin{equation}\label{as_2_eq}
\underset{n\rightarrow\infty}{\lim}\,\Bigl(U^{(n)}\Bigr)^{-1}_{i\,i+r}=\Lambda_{i\,i+r},\,\quad\mathrm{where}\quad\underset{i\rightarrow \infty}{\lim}\Lambda_{i\,i+r} = \lambda_r\quad\mathrm{and}\quad\sum_{r = 0}^\infty\lambda_r=\lambda\in\mathbb{R}.
\end{equation}
Moreover, for each $r\geq 0$ any of the following conditions are satisfied:
\begin{itemize}
    \item[$(i)$] It exists a sequence $\lbrace\alpha_i\rbrace_{i=1}^{\infty}\in\ell_1$ such that $\left\lvert \left(U^{(n)}\right)^{-1}_{i\,i+r} - \Lambda_{i\,i+r} \right\rvert\leq \alpha_i$ for all $n\in\mathbb{N}$,
    \item[$(ii)$] For each $i\in\mathbb{N}$, the sequence $\lbrace (U^{(n)})^{-1}_{i\,i+r}\rbrace_{n\in\mathbb{N}}$ is non-decreasing or  non-increasing.
\end{itemize}
\end{assumption}

Note that Assumptions~\ref{as_3} and \ref{as_2} implicitly require an ordering of the observations in $\mathbf{X}$. More precisely, they require the existence of a permutation of the rows in $\mathbf{X}$ such that their conditions are satisfied. The following result generalizes Proposition~\ref{prop_overestimate} to arbitrary sequences $\lbrace \mathbf{U}^{(n)}\rbrace_{n\in\mathbb{N}}$.

\begin{prop}\label{prop_overestimate_general} Let $\mathbf{X}^{(n)}\sim\mathcal{M}\mathcal{N}_{n\times p}(\boldsymbol\mu^{(n)},\mathbf{U}^{(n)},\mathbf{\Sigma})$, where $\boldsymbol\mu^{(n)}$ and $\mathbf{U}^{(n)}$ satisfy Assumptions~\ref{as_1}, \ref{as_3} and \ref{as_2} for some $K^*>1$. Let $\hat{\mathbf{\Sigma}}$ be the estimator defined in \eqref{estSigma}. Then,
\begin{equation}\label{asymp_overest}
    \underset{n\rightarrow\infty}{\lim}\,\mathbb{P}\left(\hat{\mathbf{\Sigma}}\Bigl(\mathbf{X}^{(n)}\Bigr)\succeq \mathbf{\Sigma}\right)=1.
\end{equation}
\end{prop}
As a consequence, Proposition~\ref{prop_indcopy} directly holds for arbitrary $\lbrace \mathbf{U}^{(n)}\rbrace_{n\in\mathbb{N}}$ if Assumption~\ref{as_2} is added to its hypotheses. Our proof of Proposition~\ref{prop_overestimate_general} relies on the following Lemma, which makes use of Assumptions \ref{as_1}, \ref{as_3} and \ref{as_2} explicitly. Both results are proved in Appendix~\ref{proofs_2}.

\begin{lemma}\label{lemma_prop} Let $\mathbf{X}^{(n)}\sim\mathcal{M}\mathcal{N}_{n\times p}(\boldsymbol\mu^{(n)},\mathbf{U}^{(n)},\mathbf{\Sigma})$, where $\boldsymbol\mu^{(n)}$ and $\mathbf{U}^{(n)}$ satisfy Assumptions~\ref{as_1}, \ref{as_3} and \ref{as_2} for some $K^*>1$. Then,
\begin{equation}\label{lemma_eq}
    \underset{n\rightarrow\infty}{\lim}\,\frac{1}{n}\sum_{l,s=1}^n \left(U^{(n)}\right)^{-1}_{ls}\,\mathds{1}\lbrace \mu_l^{(n)}=\theta_k\rbrace\,\mathds{1}\lbrace \mu_s^{(n)}=\theta_{k'}\rbrace=2(\lambda-\lambda_0)\pi_{k}\pi_{k'}+\lambda_0\pi_k\delta_{kk'},
\end{equation}
for any $k,k'\in\lbrace1,\ldots,K'\rbrace$, and where $\pi_k,\pi_{k'}$ and $\lambda_0,\lambda$ are defined in Assumptions \ref{as_1} and \ref{as_2} respectively.  
\end{lemma}

Assessing whether a model of dependence satisfies Assumption~\ref{as_2} is not trivial as it requires full knowledge of how the inverse matrices $\left(\mathbf{U}^{(n)}\right)^{-1}$ grow up when the sample size increases. However, we are able to show that Assumption~\ref{as_2} is satisfied for some specific dependence models and, consequently, that selective type I error can be controlled when $\mathbf{\Sigma}$ is over-estimated in such cases. The following remarks are proved in Appendix~\ref{proofs_2}.

\begin{remark}[Compound symmetry]\label{remark_compound} Let $\mathbf{U}^{(n)}=(a-b)\mathbf{I}_n+b\mathbf{1}_{n\times n}$ for some $a>b>0$. Then, $\lbrace\mathbf{U}^{(n)}\rbrace_{n\in\mathbb{N}}$ satisfies Assumption~\ref{as_2}.   
\end{remark}

The compatibility of compound symmetry structures with the over-estimation of $\bSigma$ can be explained within this more general framework: Remark~\ref{remark_compound} and Proposition~\ref{prop_overestimate_general} imply Proposition~\ref{prop_overestimate}. Therefore, we do not provide a direct proof of the latter result. We can also consider the case of independent observations with different variances along features. Note that, if the matrix $\mathbf{X}$ is transposed, any general dependence structure between observations $\mathbf{U}$ can be estimated if independent features with known variances are provided, which is already an important generalization of~\cite{Gao}.

\begin{remark}[Diagonal]\label{remark_diagonal} Let $\mathbf{U}^{(n)}=\mathrm{diag}(\lambda_1,\ldots,\lambda_n)$. If the sequence $\lbrace\lambda_n\rbrace_{n\in\mathbb{N}}$ is convergent, the sequence $\lbrace \mathbf{U}^{(n)}\rbrace_{n\in\mathbb{N}}$ satisfies Assumption~\ref{as_2}.
\end{remark}

We can extend the complexity of $\mathbf{U}^{(n)}$ to auto-regressive covariance structures of any lag. This is mainly thanks to the fact that the inverses of such matrices are tractable and banded, i.e. their non-zero entries are confined to a centered diagonal band. Under model \eqref{model}, assuming that $\mathbf{U}^{(n)}$ is the covariance matrix of an auto-regressive process of order $P$ means that
\begin{equation}\label{ar_p}
    \frac{1}{\sqrt{\Sigma_{jj}}}\,X_{ij}^{(n)}=\frac{1}{\sqrt{\Sigma_{jj}}}\,\sum_{s=1}^P\beta_s\,X_{i-s\,j}^{(n)}+\varepsilon_{i},\quad\forall\,j\in\lbrace1,\ldots,p\rbrace,
\end{equation}
where $\lbrace\varepsilon_{i}\rbrace_{i=1,\ldots,n}$ are i.i.d univariate centered normal variables and $\lbrace\beta_s\rbrace_{s=1,\ldots,P}\subset\mathbb{R}$ are the model coefficients. Then, for any $j\in\lbrace1,\ldots,p\rbrace$, the entries of $\mathbf{U}^{(n)}$ would be given by
\begin{equation}
    U_{ii'}=\mathrm{Cov}\left(\frac{X_{ij}}{\sqrt{\Sigma_{jj}}},\frac{X_{i'j}}{\sqrt{\Sigma_{jj}}}\right),\quad\forall i,i'\in[n],\,\,\forall\,j\in\lbrace1,\ldots,p\rbrace.
\end{equation}
If model \eqref{ar_p} is assumed, the covariance matrix $\mathbf{U}^{(n)}$ and its inverse have a tractable structure. For example, for the simplest auto-regressive process where $P=1$, and the $i$-th observation depends linearly only on the $(i-1)$-th one, the entries of $\mathbf{U}^{(n)}$ have the form $U^{(n)}_{ij}=\sigma^2\rho^{\abs{i-j}}$, for $\sigma>0$. To ensure the the positive definiteness of $\mathbf{U}^{(n)}$, we need $\abs{\rho}<1$ (see the form of eigenvalues in \cite{TRENCH}). This is equivalent to ask the the process to be stationary. Then, the inverse of $\mathbf{U}^{(n)}$ is a tridiagonal matrix of the form

\begin{equation}\label{mat_ar1}
\left(\mathbf{U}^{(n)}\right)^{-1}= \frac{1}{\sigma^2(1-\rho^2)}
\begin{pmatrix}
1 & -\rho &  \\
-\rho & 1+ \rho^2 & -\rho &  \\
  & -\rho & \ddots & \ddots  \\
  & & \ddots & 1 + \rho^2 & -\rho \\
  & & & -\rho & 1
\end{pmatrix}.
\end{equation}

The super and sub-diagonals trivially satisfy condition $(i)$ in Assumption~\ref{as_2} with $\lambda_{\pm 1}=-\rho/(1-\rho^2)$. Then, the entries of the main diagonal define the sequences
\begin{equation*}
    \sigma^2(1-\rho^2)\,\left\lbrace \left(U^{(n)}\right)^{-1}_{ii}\right\rbrace_{n\in\mathbb{N}}=
    \begin{cases} 
      \lbrace 1,1,\ldots\rbrace& \mathrm{if }\,\,i = 1,\\
      \lbrace \xi_1, \ldots, \xi_{i-1}, 1, 1+\rho^2, 1+\rho^2,\ldots\rbrace & \mathrm{if }\,\,i>1,
   \end{cases}   
\end{equation*}
for every $i\in\mathbb{N}$, where the entries $\sigma^2(1-\rho^2)\,(U^{(n)})^{-1}_{ii}=\xi_n$ for $i>n$ can be chosen as needed. Note that these sequences do not satisfy condition $(i)$ in Assumption~\ref{as_2}, but they are non-decreasing (choosing appropriately the $\xi_k$). Consequently, Assumption~\ref{as_2} holds and we have $\Lambda_{11}=1/(\sigma^2((1-  \rho^2))$, $\Lambda_{ii}=\lambda_{0}=(1+\rho^2)/(\sigma^2((1-  \rho^2))$ for all $i>1$ and, finally, $\lambda = (1-\rho)^2/(\sigma^2((1-\rho^2))$. For any $P\geq1$, the inverse matrices are banded with $2P+1$ non-zero diagonals and we can follow the same reasoning. However, for $P>2$, we need to require the coefficients $\beta_1,\ldots,\beta_P$ to have the same sign.

\begin{remark}[Auto-regressive]\label{remark_ar} Let $\mathbf{U}^{(n)}$ be the covariance matrix of an auto-regressive process of order $P\geq 1$ such that, if $P>2$, $\beta_k\beta_{k'}\geq 0$ for all $k,k'\in\lbrace 1,\ldots,P\rbrace$. Then, the sequence $\lbrace \mathbf{U}^{(n)}\rbrace_{n\in\mathbb{N}}$ satisfies Assumption~\ref{as_2}.    
\end{remark}

The above remarks imply that~\eqref{estSigma} satisfies~\eqref{sigma_overestimate_condition} in the above-studied compound symmetry, diagonal and auto-regressive models. Consequently, the asymptotic null super-uniformity of~\eqref{hat_pvalue} will be robust to $\mathbf{U}$ being diagonal or auto-regressive as long as the null uniformity of~\eqref{pvalue_F} is robust to $\mathbf{U}$ belonging to such models (and Assumptions~\ref{as_1} and~\ref{as_3} hold).

\section{Numerical experiments}\label{sec:numerical_experiments}

In this section, we assess the numerical performance of the proposed approach in several scenarios simulated with synthetic data. We start by simulating settings that satisfy condition $(ii)$ in Theorem~\ref{th:pvalue_V}, that is, choosing $\mathbf{U}\in\mathcal{CS}(n)$ and using the $p$-value~\eqref{pvalue_F}. The following three cases are considered for the scale matrices $\mathbf{U}$ and $\mathbf{\Sigma}$:
\begin{itemize}
    \item[$(D1)$] $\mathbf{U}=\mathbf{I}_n$ and $\mathbf{\Sigma}$ is the covariance matrix of an AR(1) model, i.e. $\Sigma_{ij}=\sigma^2\rho^{\abs{i-j}}$, with $\sigma = 1$ and $\rho=0.5$. 
    \item[$(D2)$] $\mathbf{U}=b\mathbf{1}_{n\times n}+(a-b)\mathbf{I}_n$, with $a = 0.5$ and $b=1$. $\mathbf{\Sigma}$ is a Toeplitz matrix, i.e. $\Sigma_{ij}=t(\abs{i-j})$, with $t(s)=1+1/(1+s)$ for $s\in\mathbb{N}$.
    \item[$(D3)$] $\mathbf{U}=b\mathbf{1}_{n\times n}+(a-b)\mathbf{I}_n$, with $a = 0.2$ and $b=2$. $\mathbf{\Sigma}$ is a diagonal matrix with diagonal entries given by $\Sigma_{ii}=1+1/i$. 
\end{itemize}

We simulated matrix normal data in settings $(D1)$, $(D2)$ and $(D3)$ and performed $k$-means and hierarchical agglomerative clustering (HAC) with average, centroid, single and complete linkages. In Section~\ref{sec:global_h0} we illustrate the uniformity of the $p$-values \eqref{pvalue_V} under a global null hypothesis, assuming that both scale matrices are known. In Section~\ref{sec:sigma_est}, we consider the case where $\mathbf{U}$ is known and the covariance between features $\mathbf{\Sigma}$ is estimated. We show, as proved in Section~\ref{sec:estimation_CS}, that $p$-values are super-uniform for large enough sample sizes. In Section~\ref{sec:power}, we assess the relative efficiency of the considered algorithms in terms of power, for the three dependence scenarios. Finally, in Section~\ref{sec:robustness}, we study the robustness of the proposed approach to model misspecification.

\subsection[Uniform p-values under a global null hypothesis]{Uniform $p$-values under a global null hypothesis}\label{sec:global_h0}

To illustrate the null distribution of $p$-values, we followed the same steps as in \cite[Section 5.1]{Gao}. For $n=100$ and $p\in\lbrace 5,20,50\rbrace$, we simulated $M=2000$ samples drawn from model \eqref{model} in settings $(D1)$, $(D2)$ and $(D3)$ with $\boldsymbol{\mu}=\mathbf{0}_{n\times p}$ a zero matrix, so that the null hypothesis \eqref{h0} holds for any pair of clusters in $\mathcal{C}(\mathbf{X})$. For each simulated sample, we used $k$-means and HAC to estimate three clusters and tested \eqref{h0} for two randomly selected clusters. Results for HAC with average linkage are displayed in Figure~\ref{fig:average_global_h0}, where the empirical cumulative distribution functions (ECDF) of the simulated $p$-values are shown. The results for $k$-means and HAC with centroid, single and complete linkage are analogous to those for average linkage and we present them in Appendix~\ref{sec:extra_sim}. The $p$-values for HAC with complete linkage were computed as their Monte Carlo approximation \eqref{monte_carlo_pv} with $N=2000$ iterations. In all cases, the $p$-values follow a uniform distribution when the null hypothesis \eqref{h0} holds.

\begin{figure}
    \centering
    \includegraphics[width=0.95\textwidth]{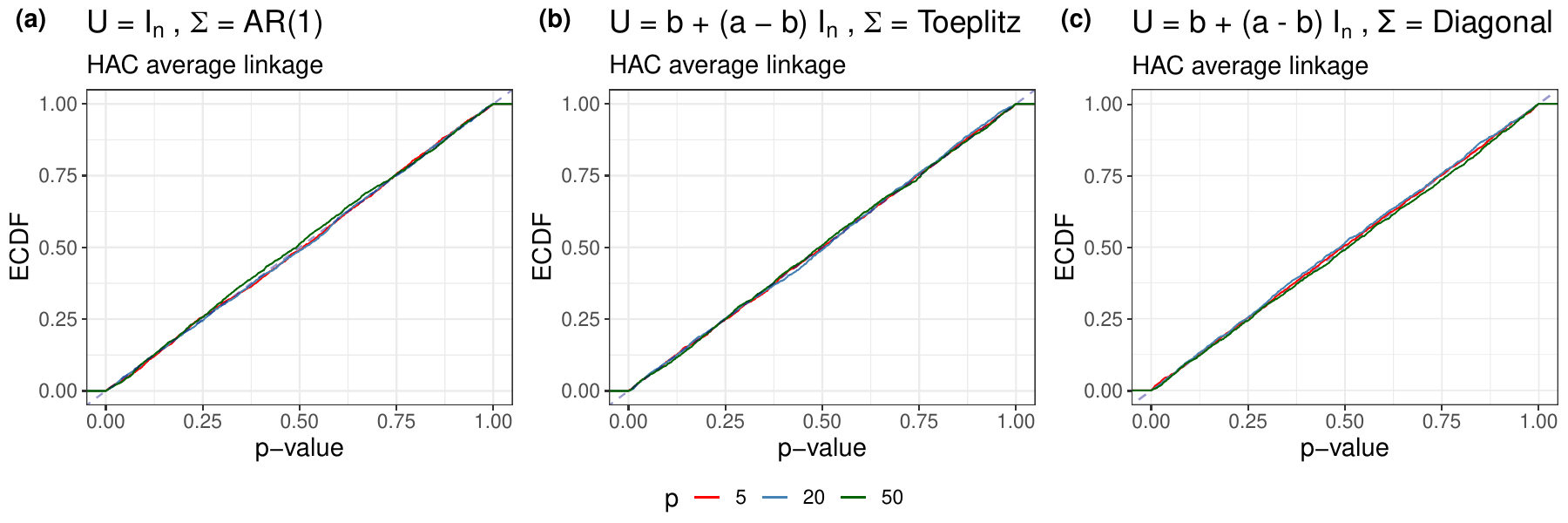}
    \caption{Empirical cumulative distribution functions (ECDF) of $p$-values \eqref{pvalue_V} with $\mathcal{C}$ being a hierarchical clustering algorithm with average linkage. The ECDF were computed from $M=2000$ realizations of \eqref{model} under the three dependence settings $(D1)$, $(D2)$ and $(D3)$ with $\boldsymbol{\mu}=\mathbf{0}_{n\times p}$, $n=100$ and $p\in\lbrace 5,20,50\rbrace$.}
    \label{fig:average_global_h0}
\end{figure}

\subsection[Super-uniform p-values for unknown Sigma]{Super-uniform $p$-values for unknown $\mathbf{\Sigma}$}\label{sec:sigma_est}

In this section, we illustrate that $p$-values \eqref{hat_pvalue} are asymptotically super-uniform under~\eqref{h0} when $\mathbf{\Sigma}$ is asymptotically over-estimated in the sense of Loewner partial order, as proved in Theorem~\ref{th:over_estimate}. We use the estimator \eqref{estSigma} that asymptotically over-estimates $\mathbf{\Sigma}$ for $\mathbf{U}\in\mathcal{CS}(n)$ if Assumptions~\ref{as_1} and \ref{as_3} hold. The estimate is computed using an independent and identically distributed copy of the sample where the clustering was performed, following Proposition~\ref{prop_indcopy}.

We follow the same steps as in \cite[Section D.1]{Gao}. For $n=100$ and $p=5$, we simulate $M=5000$ samples drawn from \eqref{model} in settings $(D1)$, $(D2)$ and $(D3)$ with $\boldsymbol{\mu}$ being divided into two clusters:

\begin{equation}\label{mu_est}
    \mu_{ij}=
    \begin{cases}
        \frac{\delta}{j} & \mathrm{if }\,\,i \leq \frac{n}{2},\\
        -\frac{\delta}{j} & \mathrm{otherwise},
    \end{cases}
    \quad\forall\,i\in[n],\,\forall\,j\in\lbrace 1,\ldots,p\rbrace,
\end{equation}
with $\delta\in\lbrace 6,8\rbrace$. For $k$-means and HAC with average, centroid, single and complete linkage we set $\mathcal{C}$ to chose three clusters. The samples for which \eqref{h0} held when comparing two randomly selected clusters are kept. Results for HAC with average linkage are presented in Figure~\ref{fig:sim_overest}. The results for $k$-means and HAC with centroid, single and complete linkage are analogous and we present them in Appendix~\ref{sec:extra_sim}. All simulations illustrate the asymptotic super-uniformity of $p$-values \eqref{pvalue_V} under the null hypothesis, when $\mathbf{\Sigma}$ is asymptotically over-estimated using \eqref{estSigma}. Moreover, as the distance between clusters $\delta$ decreases, the over-estimation is less severe and the null distribution of $p$-values approaches the one of a uniform random variable. 

It is important to remark that Figure~\ref{fig:sim_overest} serves only to illustrate the validity of Theorem~\ref{th:over_estimate}, but in no way to interpret the conservativeness of $p$-values when $\boldsymbol{\Sigma}$ is over-estimated. The deviation from uniformity of the null distribution of \eqref{hat_pvalue} or, equivalently, the power of the corresponding test, depends on the measure of the conditioning set, which in Figure~\ref{fig:sim_overest} is determined by the frequency of iterations satisfying \eqref{h0}.

\begin{figure}[t]
    \centering
    \includegraphics[width = \textwidth]{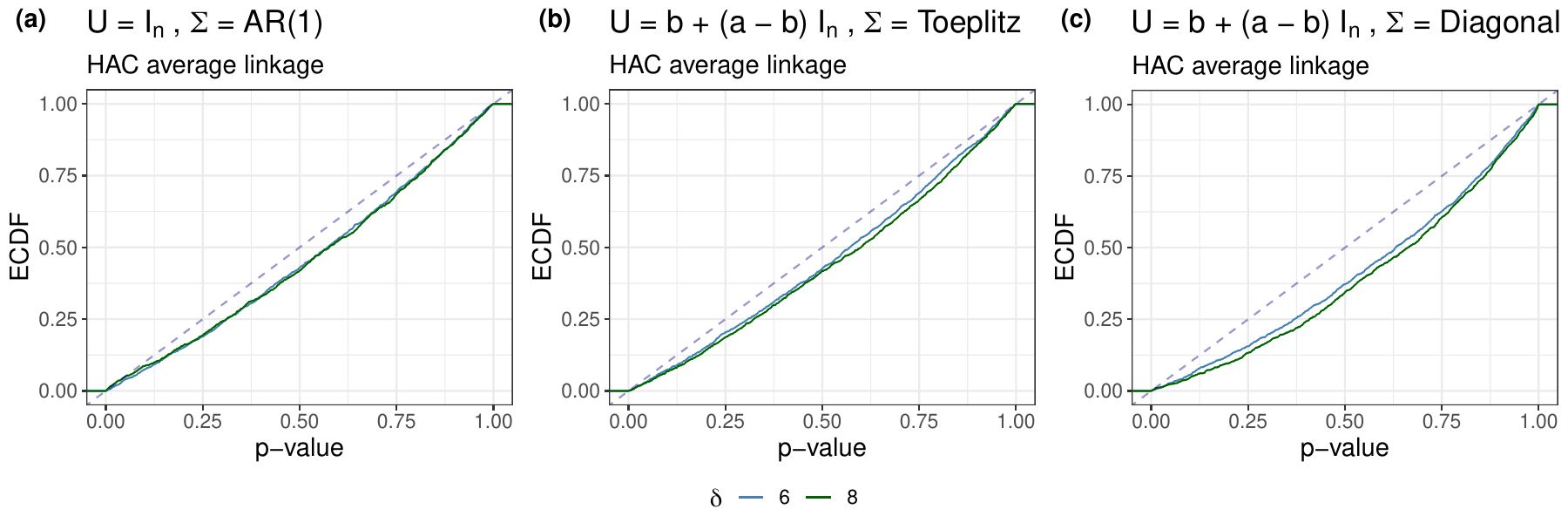}
    \caption{Empirical cumulative distribution functions (ECDF) of $p$-values \eqref{hat_pvalue} with $\mathcal{C}$ being a hierarchical clustering algorithm with average linkage. The ECDF are computed from $M=5000$ realizations of \eqref{model} under the three dependence settings $(D1)$, $(D2)$ and $(D3)$ with $n=100$, $p=5$ and $\boldsymbol\mu$ given by \eqref{mu_est} with $\delta\in\lbrace 6,8\rbrace$. Only samples for which the null hypothesis held were kept, as described in Section~\ref{sec:sigma_est}.}
    \label{fig:sim_overest}
\end{figure}

\subsection{Power analysis}\label{sec:power}

We now assess the relative efficiency of the five clustering algorithms considered in terms of power, as well as their power loss when one of the scale matrices is estimated using~\eqref{estSigma}. As in \cite[Section 5.2]{Gao}, we consider the \textit{conditional} power of the $p$-value \eqref{pvalue_V}, which is the probability of rejecting the null \eqref{h0} for a randomly selected pair of clusters given that they are different. To estimate the conditional power, we simulate $M=5000$ samples drawn from \eqref{model} under the three settings $(D1)$, $(D2)$ and $(D3)$ with $\boldsymbol{\mu}$ dividing the $n=200$ observations into three true clusters:
\begin{equation}\label{mean_power}
    \mu_{i}=
      \begin{cases} 
      \left(-\frac{\delta}{2},0,\ldots,0\right)& \mathrm{if }\,\,i \leq \lfloor \frac{n}{3}\rfloor,\\
      \left(0,\ldots,0,\frac{\sqrt{3}\delta}{2}\right) & \mathrm{if }\,\,\lfloor \frac{n}{3}\rfloor<i\leq \lfloor \frac{2n}{3}\rfloor,\\
      \left(\frac{\delta}{2},0,\ldots,0\right) & \mathrm{otherwise},
   \end{cases}
     \quad\forall\,i\in[n],
\end{equation}
for $p=5$ and for $13$ evenly-spaced values of $\delta\in[4,10]$. Then, we estimate the conditional power as the proportion of rejections at level $\alpha=0.05$ among the samples for which the null hypothesis \eqref{h0} did not hold (which were above the 90\% of $n$ in all settings). The conditional power as a function of $\delta$ is shown in Figure~\ref{fig:power}(a-c) for the three scenarios $(D1)$, $(D2)$ and $(D3)$ and the five considered clustering algorithms. The $p$-values for HAC with complete linkage are estimated using the approximation \eqref{monte_carlo_pv} with $N=2000$ iterations.

Figure~\ref{fig:power}(a-c) shows that, in all cases, conditional power increases with the distance between true clusters. Regarding HAC, we observe that average linkage presents the best relative efficiency among the four considered linkages in all the dependence settings, followed closely by complete linkage, which seems to weaken in $(D2)$. This might suggest that conditional power depends on the scale matrices and some scenarios might strongly differ from the overall observed behavior. Indeed, the qualitative difference between average or complete linkage and centroid or single linkage that is observed in $(D1)$ and $(D3)$ considerably lessens in $(D2)$. In $(D1)$ and $(D3)$, the performance of single linkage is undoubtedly the lowest, and large differences between clusters are required to attain satisfactory levels of conditional power. However, single linkage achieves one of the best performances in $(D2)$.

The relative efficiency of the $k$-means algorithm in terms of conditional power is the best in $(D2)$, but one of the worst among all the considered algorithms in $(D1)$ and $(D3)$ settings. These unsatisfactory performances might be explained by the behavior already pointed out by the authors in \cite{chen2022selective}, who referred to the fact that conditioning on too much information entails a loss of power \cite{Jewell2022, liu2018powerful, chen2022powerful, fithian2017optimal}.
Recall that the truncation set for $k$-means post-clustering inference defined in \cite{chen2022powerful} is non-maximal to allow its efficient computation (see Appendix~\ref{sec:finer_cond} and \cite[Equation (9)]{chen2022selective}). This approach, although respecting the selective type I error as shown in Theorem~\ref{th:pvalue_V_E}, might sacrifice the efficiency in terms of power of the corresponding test, as illustrated in Figure~\ref{fig:power}(a,c).

\begin{figure}[th!]
    \centering
    \includegraphics[width=0.95\textwidth]{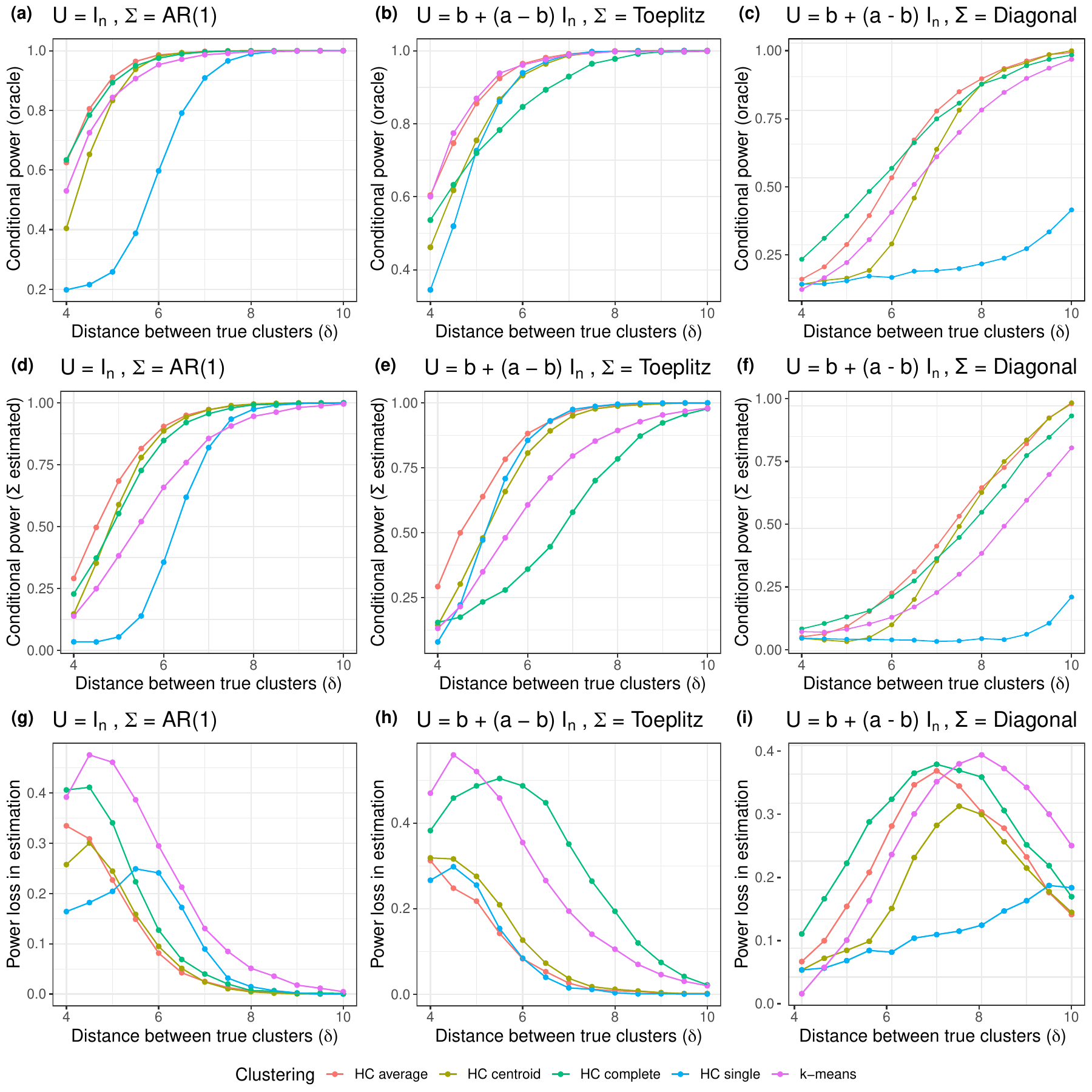}
    \caption{(a-f): conditional power for the test proposed in Section~\ref{sec:sel_clustering} under model \eqref{model} with the three dependence settings $(D1)$, $(D2)$ and $(D3)$ and the mean matrix defined in \eqref{mean_power}. The conditional power is estimated as the proportion of rejection at level $\alpha=0.05$ among the subset of  the $M=5000$ realizations of \eqref{model} for which the null hypothesis \eqref{h0} holds. In (a-c), $\mathbf{\Sigma}$ is known and in (d-f) it is over-estimated using \eqref{estSigma}. (g-i): power loss in estimation defined as the absolute difference of the curves in (a-c) and (d-f).}
    \label{fig:power}
\end{figure}

Next, we evaluate the loss of power entailed by estimating one of the scale matrices using \eqref{estSigma}. Recall that, following Theorem~\ref{th:over_estimate}, the $p$-values \eqref{hat_pvalue} are asymptotically super-uniform under the null, so conditional power is expected to decrease due to both the estimation of unknown parameters and the conservativeness of the testing approach. We repeat the previously described analysis but replacing $\boldsymbol{\Sigma}$ by its estimate \eqref{estSigma}, and calculate the counterparts of the curves in Figure~\ref{fig:power}(a-c) for $p$-values \eqref{hat_pvalue}. They are shown in Figure~\ref{fig:power}(d-f). In Figure~\ref{fig:power}(g-i), we depict the loss of power in estimation, defined as the absolute difference of the conditional power computed with known and over-estimated $\boldsymbol{\Sigma}$, for every fixed clustering algorithm and value of $\delta$.

Figure~\ref{fig:power}(g-i) illustrates how power loss varies substantially across settings $(D1)$, $(D2)$ and $(D3)$. Overall, average and centroid linkages exhibit the slightest loss, falling below 10\% for $\delta> 6$ in $(D1)$ and $(D2)$. A greater separation between clusters is required to achieve a reasonable power loss under $(D3)$. The power loss curve of complete linkage closely resembles that of average and centroid linkages in $(D1)$ and $(D3)$, but takes substantially higher values in $(D2)$. Conversely, single linkage shows a similar behavior to centroid and average linkages in $(D2)$ but differs notably in $(D1)$ and $(D3)$. Once again, we find that the $k$-means algorithm exhibits the worst relative efficiency in terms of power loss, especially in $(D1)$ and $(D3)$. A similar behavior was observed in~\cite{chen2022selective} for $k$-means clustering when over-estimating $\sigma$ under~\eqref{model_gao} using the estimator proposed in~\cite{Gao}. This suggests that the unsatisfactory efficiency of post-$k$-means inference is intrinsic to the $p$-value defined in \cite{chen2022selective}, and that the extension proposed here inherits that drawback. An alternative approach would be to explore the use of consistent estimators of $\boldsymbol{\Sigma}$ under~\eqref{model}, which would reduce power loss as demonstrated in \cite{chen2022selective} for the simpler model~\eqref{model_gao}. Following all panels in Figure~\ref{fig:power}, we can conclude that HAC with average linkage exhibits the highest relative efficiency and lower power loss when $\boldsymbol{\Sigma}$ is estimated, making it the most suitable algorithm in practice. Note that substantial power loss in the estimation of unknown parameters was similarly observed in the methods proposed in \cite{Gao, Yun2023Jan}, as demonstrated in \cite{Yun2023Jan} for HAC algorithms under~\eqref{model_gao}.

\subsection{Robustness to model misspeciﬁcation}\label{sec:robustness}

We conclude the numerical simulations on synthetic data by studying the robustness of the proposed approach to model misspecification. We particularly evaluate settings where the theoretical constraints on the dependence between observations given by $\mathbf{U}$ are not satisfied or known. First, in Section~\ref{sec:no_CS_U}, we analyze how $p$-values \eqref{pvalue_V} behave when the covariance matrix $\mathbf{U}$ is not compound symmetry, but is compatible with the over-estimation of $\mathbf{\Sigma}$. Then, in Section~\ref{sec:no_ad_U}, we explore the setting where $\mathbf{U}$ does not fit into $\mathcal{CS}(n)$ nor belongs to any of the models stated in Remarks~\ref{remark_diagonal} or \ref{remark_ar}. Finally, in Section~\ref{sec:ignoredep_U}, we evaluate the validity of the method when $\mathbf{U}\neq\mathbf{I}_n$ is unknown and observations are assumed to be independent.

\subsubsection[Non-compound-symmetry U structures]{Non-compound-symmetry $\mathbf{U}$ structures}\label{sec:no_CS_U}

In this section we evaluate the robustness of $p$-values~\eqref{pvalue_V} and \eqref{hat_pvalue} to $\mathbf{U}\notin\mathcal{CS}(n)$. We choose three dependence settings that satisfy Assumption~\ref{as_2}, so that \eqref{estSigma} satisfies~\eqref{sigma_overestimate_condition}. In all cases, $\mathbf{\Sigma}$ is a diagonal matrix with entries $\Sigma_{ii}=1+1/i$. The dependence structure between observations is given by the three following settings: 
\begin{enumerate}
    \item[(D4)] $\mathbf{U}$ is a diagonal matrix with entries $U_{ii}=1+1/i$.
    \item[(D5)] $\mathbf{U}$ is the covariance matrix of an AR(1) model with $\sigma=1$ and $\rho = 0.1$.
    \item[(D6)] $\mathbf{U}$ is the covariance matrix of an AR(2) model with $\sigma=1$, $\beta_1=0.4$ and $\beta_2=0.1$.
\end{enumerate}

We start by simulating the distribution of $p$-values~\eqref{pvalue_V} under the global null hypothesis, repeating the numerical experience described in Section~\ref{sec:global_h0}. The counterpart of Figure~\ref{fig:average_global_h0} for $(D4)$, $(D5)$ and $(D6)$ is presented in Figure~\ref{fig:non_CS_U_average}. The empirical distribution of $p$-values does not markedly deviate from uniformity in settings $(D4)$ and $(D5)$, especially for $p\in\lbrace 5,10\rbrace$. This was expected since the $\mathbf{U}$ matrices in both cases do not deviate substantially from the compound symmetry structure. In $(D6)$, the entries of $\mathbf{U}$ decay more slowly to zero along the columns, which makes this structure to deviate more from $\mathcal{CS}(n)$. This results in a greater departure from uniformity of the $p$-value distribution, as seen in Figure~\ref{fig:non_CS_U_average}(c). However, this deviation occurs within the super-uniformity regime, meaning that the $p$-values still maintain statistical guarantees, despite the power loss. The corresponding results for $k$-means and HAC with centroid, single and complete linkages are analogous. We include them in Appendix~\ref{sec:extra_sim}.

\begin{figure}[t]
    \centering
    \includegraphics[width=0.95\textwidth]{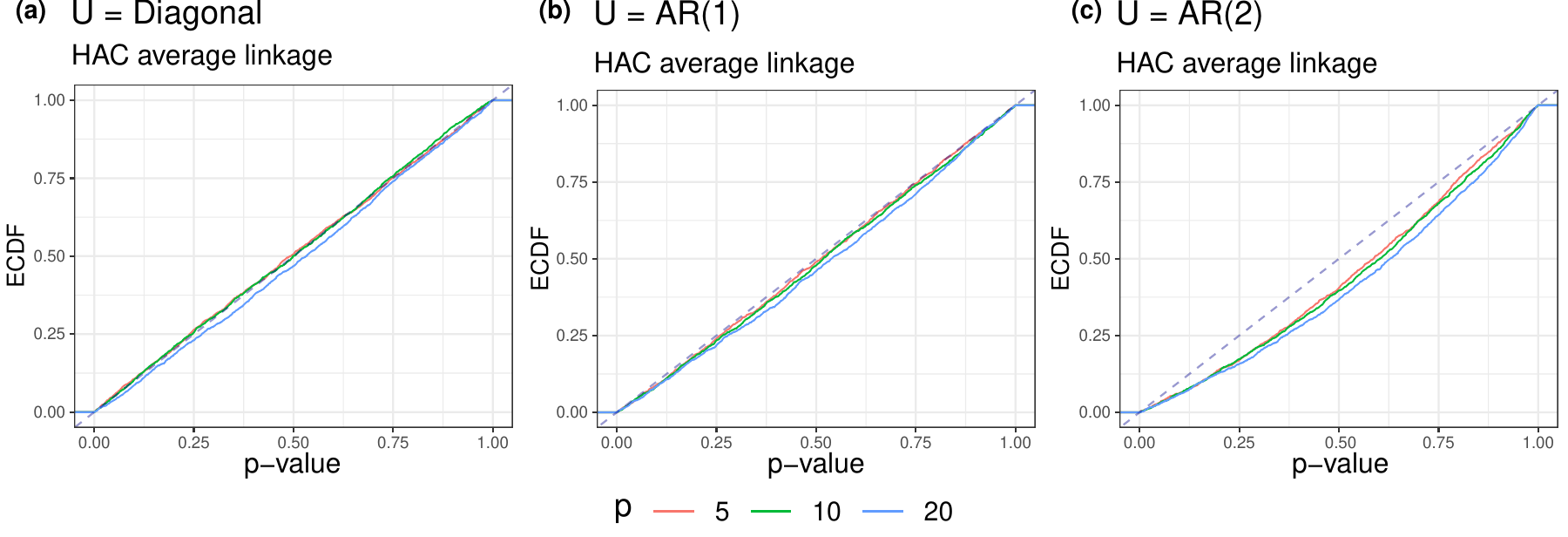}
    \caption{Empirical cumulative distribution functions (ECDF) of $p$-values \eqref{pvalue_V} with $\mathcal{C}$ being a hierarchical clustering algorithm with average linkage. The ECDF were computed from $M=2000$ realizations of \eqref{model} under the three dependence settings $(D4)$, $(D5)$ and $(D6)$ with $\boldsymbol{\mu}=\mathbf{0}_{n\times p}$, $n=100$ and $p\in\lbrace 5,20,50\rbrace$.}
    \label{fig:non_CS_U_average}
\end{figure}

The previous analysis suggests that $p$-values~\eqref{pvalue_V} are robust to small deviations from $\mathbf{U}\in\mathcal{CS}(n)$. As discussed in Section~\ref{sec:estimation_not_CS}, if the over-estimate condition~\eqref{sigma_overestimate_condition} of Theorem~\ref{th:over_estimate} is satisfied, this would mean that $p$-values~\eqref{hat_pvalue} are equally robust in that setting. Following from Remarks~\ref{remark_diagonal} and \ref{remark_ar}, settings $(D4)$, $(D5)$ and $(D6)$ are compatible with the asymptotic over-estimation of $\mathbf{\Sigma}$ using~\eqref{estSigma}. Consequently, we reproduce the analyses of Section~\ref{sec:sigma_est} for such dependence structures to assess whether the previously illustrated robustness is maintained with estimation. Results are presented in Figure~\ref{fig:non_CS_U_est_average} for HAC with average linkage and in Appendix~\ref{sec:extra_sim} for the remaining clustering algorithms. In all cases, the empirical null distribution of $p$-values is super-uniform, confirming the robustness of \eqref{hat_pvalue} to small deviations from $\mathbf{U}\in\mathcal{CS}(n)$.

\begin{figure}[t]
    \centering
    \includegraphics[width=0.95\textwidth]{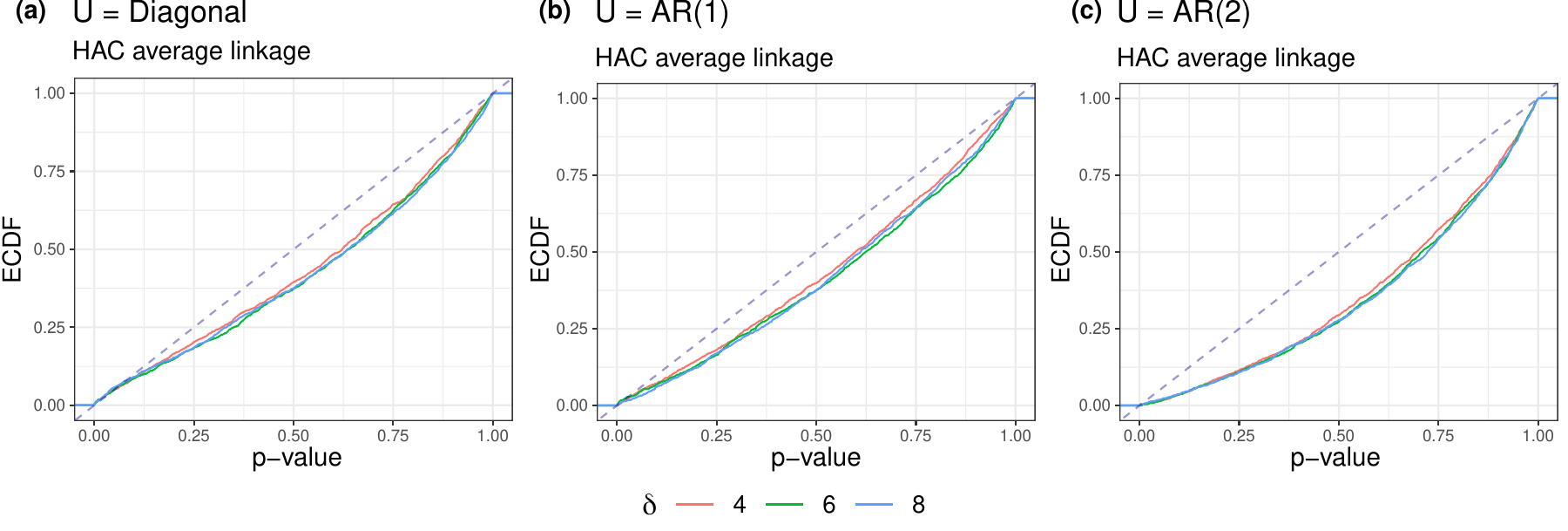}
    \caption{Empirical cumulative distribution functions (ECDF) of $p$-values \eqref{hat_pvalue} with $\mathcal{C}$ being a hierarchical clustering algorithm with average linkage. The ECDF are computed from $M=5000$ realizations of \eqref{model} under the three dependence settings $(D4)$, $(D5)$ and $(D6)$ with $n=50$, $p=5$ and $\boldsymbol\mu$ given by \eqref{mu_est} with $\delta\in\lbrace 4,6,8\rbrace$. Only samples for which the null hypothesis held were kept, as described in Section~\ref{sec:sigma_est}.}
    \label{fig:non_CS_U_est_average}
\end{figure}

\subsubsection[Non-admissible U for the over-estimation of Sigma]{Non-admissible $\mathbf{U}$ for the over-estimation of $\mathbf{\Sigma}$}\label{sec:no_ad_U}

Let us recall that Assumption~\ref{as_2} is a sufficient condition for the sequence $\lbrace \mathbf{U}^{(n)}\rbrace_{n\in\mathbb{N}}$ to ensure that \eqref{estSigma} satisfies~\eqref{sigma_overestimate_condition}. As discussed in Section~\ref{sec:unknown_sigma}, proving that a given dependence model satisfies this Assumption is non-trivial in most cases. In Remarks~\ref{remark_diagonal}, \ref{remark_compound} and \ref{remark_ar}, we showed that Assumption~\ref{as_2} is satisfied by three common dependence structures, but other sequences $\lbrace \mathbf{U}^{(n)}\rbrace_{n\in\mathbb{N}}$ might also satisfy the same sufficient condition or other unknown hypotheses that ensure that \eqref{estSigma} asymptotically over-estimates $\boldsymbol{\Sigma}$. In this section, we repeat the simulations of Section~\ref{sec:sigma_est} under three settings that do not fit Remarks~\ref{remark_compound}, \ref{remark_diagonal} or \ref{remark_ar}:

\begin{itemize}
    \item[$(D7)$]  $\mathbf{U}$ is a Toeplitz matrix with $U_{ij}=1+1/(1+\abs{i-j})$.
    \item[$(D8)$]  $\mathbf{U}$ is the covariance matrix of an AR(3) model with $\sigma=1$, $\beta_1=0.4$, $\beta_2=-0.2$ and $\beta_3 = 0.1$.
    \item[$(D9)$]  $\mathbf{U}$ is a banded matrix with $U_{ii}=1$, $U_{ii+1}=0.6$, $U_{ii+2}=0.5$, $U_{ii+3}=0.2$ and $U_{ii+r}$ for all $r>3$. 
\end{itemize}
In all cases, we chose $\boldsymbol{\Sigma}$ as a diagonal matrix with entries $\Sigma_{ii}=1+1/i$. We also set $n=50$, $p=5$ and $\delta\in\lbrace 4,6,8\rbrace$. Results are presented in Figure~\ref{fig:non_ad_U_average} for HAC with average linkage, and in Appendix~\ref{sec:extra_sim} for the rest of clustering algorithms. The simulated $p$-values are super-uniform in all settings, suggesting that \eqref{estSigma} might asymptotically over-estimate $\boldsymbol{\Sigma}$ for further models of dependence between observations. Note that, in particular, results corresponding to $(D8)$ suggest that the requirement $\beta_k\beta_{k'}\geq0$ for $P>2$ in Remark~\ref{remark_ar} is not very restrictive.

These results might also motivate further theoretical inspection on Toeplitz and banded structures to verify whether they satisfy Assumption~\ref{as_2}. Extensive work has been done on the asymptotic behavior of continuous functions of Toeplitz matrices \cite{Gray}. However, it mainly concerns their average behavior rather than their element-wise one. Notably, in \cite{Gray}, it is proved that the mean of the eigenvalues of $\left(\mathbf{U}^{(n)}\right)^{-1}$ converges when $n$ tends to infinity, if the sequence $\lbrace U_{1n}\rbrace_{n\in\mathbb{N}}$ is absolutely summable. This implies that the mean of the sequence $\lbrace(\mathbf{U}^{(n)})^{-1}_{ii}\rbrace_{i=1,\ldots,n}$ also converges with $n$. However, this is insufficient to state convergence of \eqref{convergence_U} and the asymptotic behavior of the individual entries need to be studied. If we impose $\mathbf{U}^{(n)}$ to be banded, the entry-wise convergence of the elements $\left(\mathbf{U}^{(n)}\right)^{-1}_{i\,i+r}$ has been demonstrated in \cite{elementwise_toep} for the tridiagonal case. This, together with the exponential decay of the entries of banded matrices~\cite{expdecay}, is enough to prove the first part of Assumption~\ref{as_2} for tridiagonal Toeplitz matrices. Unfortunately, the existing results do not ensure that any of the conditions $(i)$ or $(ii)$ in Assumption~\ref{as_2} hold. Assessing that remaining step is mathematically very challenging and it is left for future work.

\begin{figure}[t]
    \centering
    \includegraphics[width=0.95\textwidth]{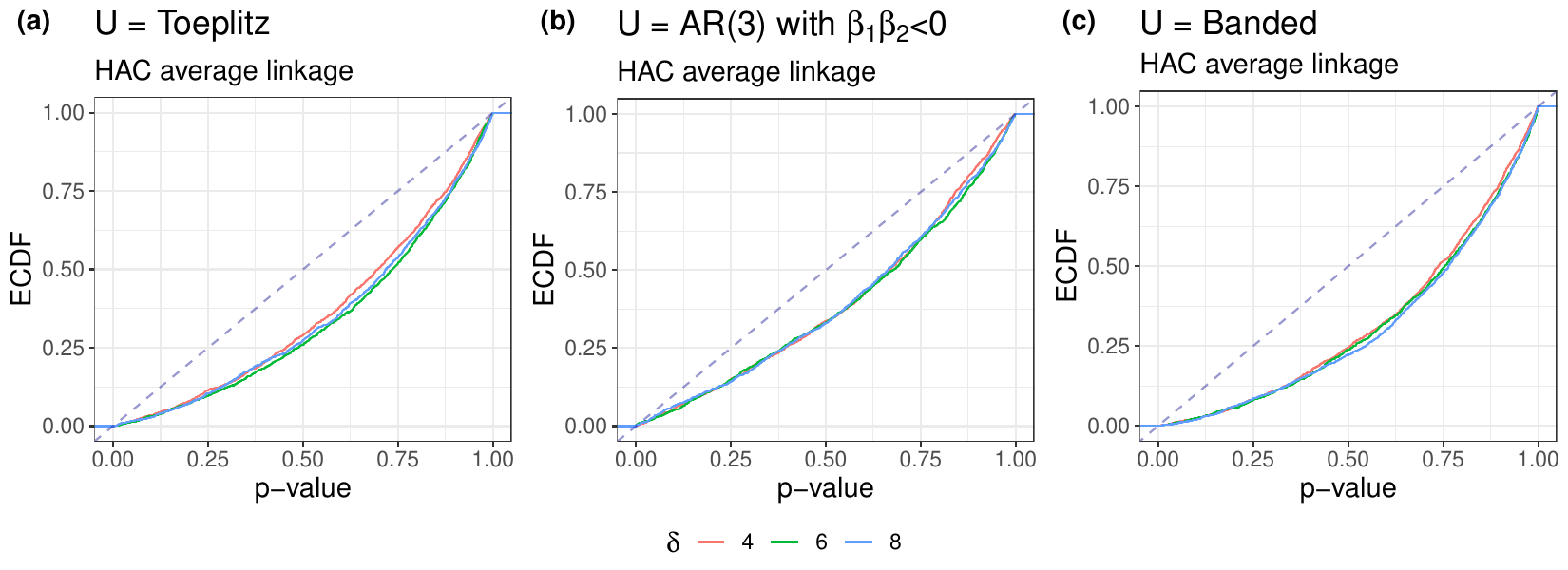}
    \caption{Empirical cumulative distribution functions (ECDF) of $p$-values \eqref{hat_pvalue} with $\mathcal{C}$ being a hierarchical clustering algorithm with average linkage. The ECDF are computed from $M=5000$ realizations of \eqref{model} under the three dependence settings $(D7)$, $(D8)$ and $(D9)$ with $n=50$, $p=5$ and $\boldsymbol\mu$ given by \eqref{mu_est} with $\delta\in\lbrace 4,6,8\rbrace$. Only samples for which the null hypothesis held were kept, as described in Section~\ref{sec:sigma_est}.}
    \label{fig:non_ad_U_average}
\end{figure}

\subsubsection{Ignoring weak dependence between observations}\label{sec:ignoredep_U}

In real applications, it might be common that the practitioner lacks knowledge of both dependence structures between observations and variables. As discussed in Section~\ref{sec:unknown_sigma}, simultaneous estimation of both matrices $\mathbf{U}$ and $\boldsymbol{\Sigma}$ is unfeasible under the matrix normal model \eqref{model} when only one or few copies of $\mathbf{X}$ are available. Consequently, even ignoring the control of statistical guarantees, we are unable to simultaneously consider a pair of estimators $\hat{\mathbf{U}}$, $\hat{\boldsymbol{\Sigma}}$ (or one of the Kronecker product $\mathbf{U}\otimes\boldsymbol{\Sigma}$) in the context of this work. In practice, a common alternative strategy is to assume weak dependence between observations, and ignore this dependence by considering $\mathbf{U}=\mathbf{I}_n$ in the method. In this section, we study the robustness of the proposed approach when observations are supposed independent but it is known that $\mathbf{U}\neq\mathbf{I}_n$.

We consider $\mathbf{X}$ drawn from \eqref{model} with $\mathbf{\Sigma}$ a diagonal matrix having as entries $\Sigma_{ii}=1+1/i$ as in the previous section. The dependence between observations is encoded by the covariance matrix of an AR(1) model, that is, $U_{ij}=\sigma^2\rho^{\abs{i-j}}$, with $\sigma=1$ and $\rho\in\lbrace 0.1,0.2,0.3,0.4,.0.5\rbrace$. Once again, we repeated the simulations described in Section~\ref{sec:sigma_est} and computed the $p$-values \eqref{hat_pvalue} using \eqref{estSigma} to estimate $\boldsymbol{\Sigma}$ and assuming $\mathbf{U}=\mathbf{I}_n$. Results for HAC with average linkage are presented in Figure~\ref{fig:ignoredep_average}, and in Appendix~\ref{sec:extra_sim} for the rest of clustering algorithms. In all cases, the simulated $p$-values do not substantially deviate from the super-uniform regime. Besides, if we take a closer look at $[0,0.1]$, we see that the simulated ECDF strictly lie below the diagonal for small values of $\rho$. In other words, when the dependence between observations is weak, the proposed test is robust to departures from the assumption $\mathbf{U}=\mathbf{I}_n$, and the estimation of $\mathbf{\Sigma}$ using \eqref{estSigma} yields $p$-values that asymptotically control the selective type I error.

\begin{figure}[t]
    \centering
    \includegraphics[width=0.95\textwidth]{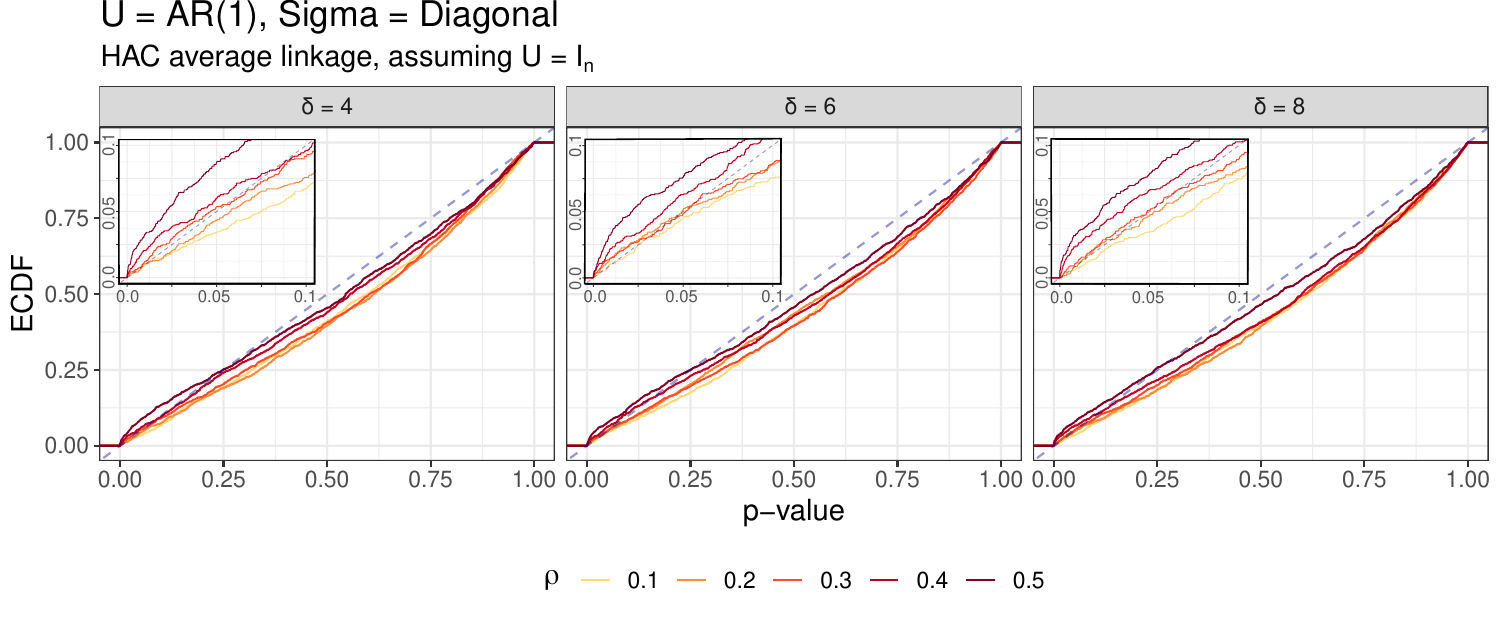}
    \caption{Empirical cumulative distribution functions (ECDF) of $p$-values \eqref{hat_pvalue} with $\mathcal{C}$ being a hierarchical clustering algorithm with average linkage. The ECDF are computed from $M=5000$ realizations of \eqref{model} as described in Section~\ref{sec:ignoredep_U} with $n=50$, $p=5$ and $\boldsymbol\mu$ given by \eqref{mu_est} with $\delta\in\lbrace 4,6,8\rbrace$. Only samples for which the null hypothesis held were kept, as described in Section~\ref{sec:sigma_est}.}
    \label{fig:ignoredep_average}
\end{figure}

\section{Application to clustering of protein structures}\label{sec:proteins}

Proteins are essential molecules in all living organisms. Many of their numerous functions are closely related to their non-static structure, which exhibits high variability within numerous protein families~\cite{Liljas:2009, Dyson:2005, Oldfield:2014}. The characterization of such intrinsic structural complexity represents a highly active area of research in the field of Structural Biology. In this pursuit, clustering methods applied to protein conformations have provided valuable insights into this challenging problem~\cite{engens, Appadurai2022}. One of the most commonly-chosen descriptors to characterize a protein conformation is the set of pairwise Euclidean distances between every pair of amino acids along the  sequence \cite{phillips1970, Nishikawa, Lazar}, usually referred to as distance maps. As these distances are strongly correlated, assuming a constant diagonal covariance matrix as in \cite{Gao} seems very unrealistic. Instead, we opt for the more convenient model
\begin{equation}\label{prot_model}
    \mathbf{X}\sim\mathcal{MN}_{n\times p}(\boldsymbol\mu, \mathbf{I}_n, \mathbf{\Sigma}),
\end{equation}
where $\mathbf{\Sigma}$ can be estimated using \eqref{estSigma}. Each row of $\mathbf{X}$ corresponds to a protein conformation, featured by a vector of Euclidean distances between every pair of amino acids, which constitute the columns of $\mathbf{X}$. We perform hierarchical agglomerative clustering with average linkage (as it showed the best relative efficiency in Section~\ref{sec:power}) to estimate $k=6$ clusters among $n=2000$ conformations of a disordered protein called Histatin-5 (Hst5). The number of clusters was chosen arbitrarily. The corresponding sequence is 24 amino acids long, so $p=23\cdot 24/2=276$. The conformations were generated using Flexible-Meccano \cite{Ozenne:2012, Bernado:2005} and refined using previously reported small-angle X-ray scattering (SAXS) data~\cite{sagar2021}. Note that Flexible-Meccano is a sampling algorithm that generates an independent conformation at each iteration, contrary to Molecular Dynamics simulation techniques that present temporal dependence between samples. This justifies our choice of $\mathbf{U}=\mathbf{I}_n$. Moreover, we had access to an independent replica of the simulated ensemble that we used to estimate $\mathbf{\Sigma}$, as it is usual for generated protein ensembles. The obtained estimate $\hat{\bSigma}$ substantially deviated from the spherical structure. Figure~\ref{fig:histatin} shows the average distance map across all conformations in a given cluster or, in other words, the empirical cluster means as defined in \eqref{empirical_mean}. Table~\ref{table} presents the $p$-values corresponding to every pair of clusters, corrected for multiple testing using the Bonferroni-Holm adjustment~\cite{holm}.

\begin{figure}[t]
    \centering
    \includegraphics[width=0.7\textwidth]{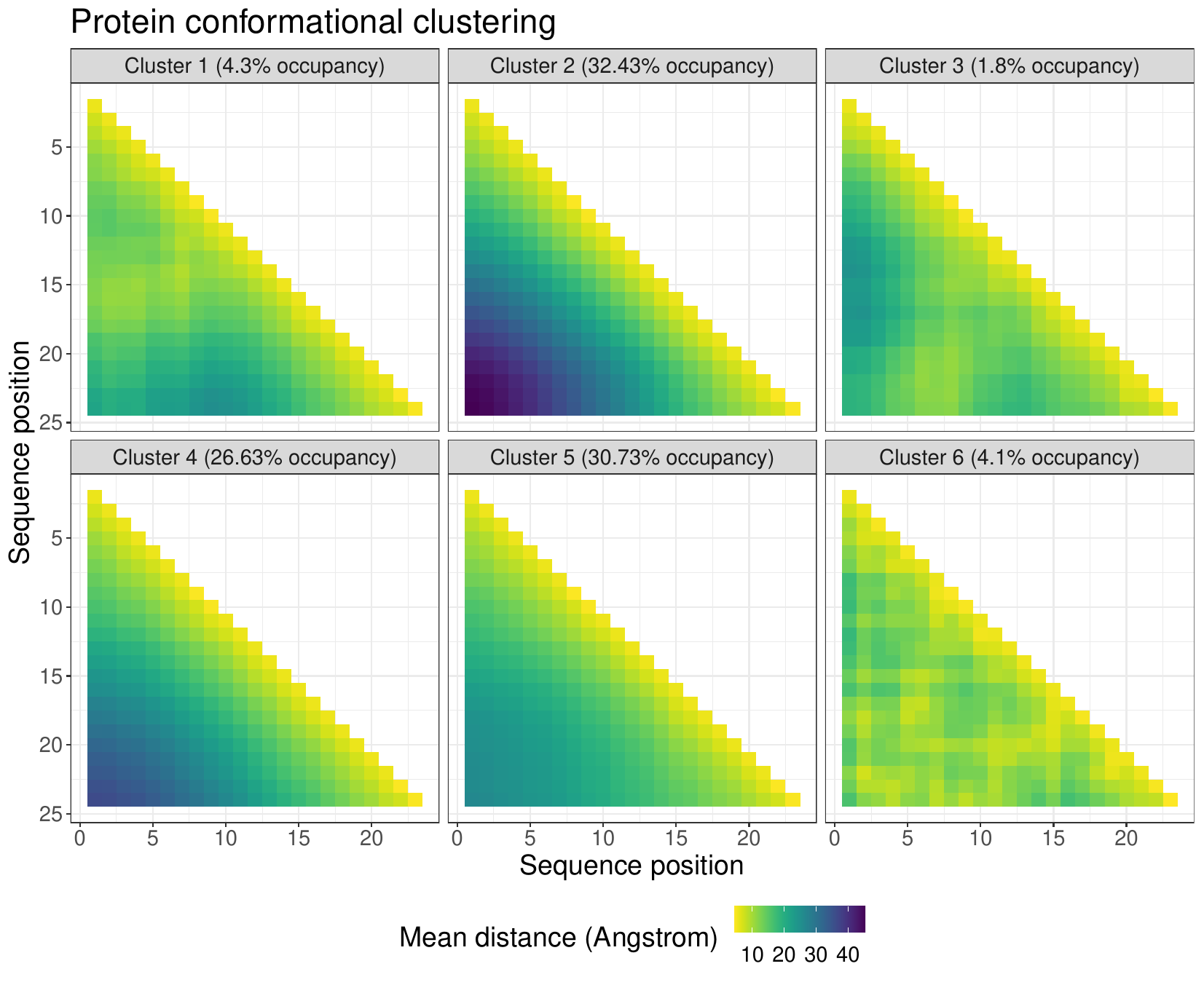}
    \caption{Average pairwise distances between every pair of amino acids across the conformations of each cluster. The clusters were found after performing hierarchical clustering with average linkage on the protein data presented in Section~\ref{sec:proteins}.}
    \label{fig:histatin}
\end{figure}

\begin{table}[ht!]
\centering
\vspace{10pt}
\scriptsize{
\begin{tabular}{cccccc}
\toprule
\textbf{Cluster} & 1 & 2 & 3 & 4 & 5 \\ \midrule
2 & 2.187589$\cdot 10^{-4}$ &  \\
3 &  3.039844$\cdot 10^{-11}$ & 1.41$\cdot 10^{-3}$  &  \\
4 &  1.070993$\cdot 10^{-10}$ & \textcolor{blue}{0.300540} & 2.98464$\cdot 10^{-4}$  &\\
5 & 3.038979$\cdot 10^{-16}$ &  \textcolor{blue}{0.093018} & 6.015797$\cdot 10^{-5}$ & \textcolor{blue}{0.105446} & \\
6 &  1.729616$\cdot 10^{-6}$ & 0.010612  & 9.290826$\cdot 10^{-9}$ &  2.105$\cdot 10^{-3}$ & 5.624624$\cdot 10^{-5}$  \\\bottomrule
\end{tabular}}
\caption{$p$-values \eqref{pvalue_V} computed under model \eqref{prot_model} retrieved after testing \eqref{h0} on the protein data presented in Section~\ref{sec:proteins}. The hierarchical clustering algorithm was set to find six clusters using average linkage. In blue, adjusted $p$-values for which the null is not rejected at level $\alpha=0.05$.}
\label{table}
\end{table}

The $p$-values presented in Table~\ref{table} show significant differences between the most part of the average distance maps depicted in Figure~\ref{fig:histatin}. The non-rejecting pairs of clusters at level $\alpha=0.05$, marked in blue in Table~\ref{table}, suggest that clusters 2, 4 and 5 could be merged into a single group. Indeed, when looking at the corresponding empirical means in Figure~\ref{fig:histatin}, we appreciate that these three clusters are characterized by large distances between pairs of amino acids that are far apart in the sequence, which indicates a lack of interactions between the sequence termini and a more extended structure of the corresponding conformations. This feature appears as an exclusive and prominent characteristic of clusters 2, 4 and 5, which might explain the non-rejection of the corresponding nulls. For the rest of rejecting pairs of clusters, clear differences in distance patterns are retrieved in Figure~\ref{fig:histatin}, accounting for significant changes on Hst5 structure between the corresponding groups. The results presented in Table~\ref{table} are coherent with the HAC dendrogram, presented in Figure~\ref{fig:dendrogram}, showing that clusters 2, 4, and 5 form a subgroup that is promptly separated from the rest.

\section{Discussion}\label{sec:discussion}

The seminal work by Gao \textit{et al.} \cite{Gao} has laid the foundation for selective inference after clustering by introducing a theoretical framework allowing to test differences between cluster means, conditioning on having estimated those clusters. Furthermore, the authors have tackled the problem of estimating unknown parameters while controlling the selective type I error, which had been overlooked in previous works \cite{datafission,rasines}, but which is crucial for the practical application of this theory. Their contribution motivates extensions of post-clustering inference to more general frameworks that arise in complex real applications, where observations or features present non-negligible dependence structures. To generalize the model considered in~\cite{Gao} to the more general~\eqref{model}, we consider a $p$-value of the form~\eqref{pvalue_gao}, choosing a test statistic based on $\mathbf{X}^T\nu$ and conditioning on both its direction and the projection $\mathbf{\pi}_\nu^{\perp}\mathbf{X}$, as done in~\cite{Gao}. In that setting, we prove that the strategy of~\cite{Gao} can be extended to~\eqref{model} if and only if the dependence structure between observations $\mathbf{U}$ is compound symmetry. Otherwise, we show that the natural generalization of~\eqref{pvalue_gao} to arbitrary $\bU$ yields a quantity that can be efficiently characterized, but whose statistical guarantees are difficult to assess. Numerically, we illustrate that the control of the selective type I error is not ensured in that setting. We also generalize the estimation of one covariance matrix compatible with the selective type I error control when $\bU\in\mathcal{CS}(n)$. These extensions, presented in Sections~\ref{sec:sel_clustering} and \ref{sec:unknown_sigma} respectively, and numerically illustrated in Sections~\ref{sec:numerical_experiments} and \ref{sec:proteins}, represent the main contributions of this work. 

The theoretical framework presented in Section~\ref{sec:sel_clustering} limits the use of $p$-values of the form~\eqref{pvalue_gao} to structures $\mathbf{U}\in\mathcal{CS}(n)$. Following from the analyses that we present in Section~\ref{sec:general_known_U}, generalizing the family of admissible $\mathbf{U}$ is a complex problem in this context and would require exploring $p$-values with alternative conditioning sets. As we have suggested, such a strategy would require the definition of extra conditioning events that are \textit{independent} of the test statistic. According to Proposition~\ref{prop_conditions}, this would mean to replace the projection $\boldsymbol{\pi}_\nu^\perp$ by one that is independent of $\mathbf{X}^T\nu$ for any $\mathbf{U}$. If the projection is taken with respect to the scalar product defined by $\mathbf{U}$, that is,
\begin{equation}\label{proj_U}
    \boldsymbol{\pi}_{\mathbf{U};\nu}^\perp\mathbf{X}=\mathbf{X}-\frac{\mathbf{X}^T \mathbf{U} \nu}{\nu^T\mathbf{U}\nu}\nu,
\end{equation}
the independence $\boldsymbol{\pi}_{\mathbf{U};\nu}^\perp\mathbf{X}\independent \mathbf{X}^T\nu$ follows from the Cochran Theorem~\cite{Cochran1934Apr}. However, replacing \eqref{proj_U} in~\eqref{decomposition_gao} and proceeding with the same reasoning would mean to consider a test statistic based on $\mathbf{X}^T\mathbf{U}\nu$, that would account for a less interpretable null hypothesis of the form $\boldsymbol{\mu}^T\mathbf{U}\nu=0$. Besides, maintaining both the projection~\eqref{proj_U} and the null hypothesis~\eqref{h0} would substantially complicate the derivation of a tractable $p$-value. For this reason, we believe that extending the conditional post-clustering inference approaches to arbitrary structures $\mathbf{U}$ would require a substantial shift in framework and should follow alternative paths to the strategy initiated in~\cite{Gao}. Recall, however, that the method proposed here has been shown to be robust to $\mathbf{U}\notin\mathcal{CS}(n)$ in several scenarios.

The estimation of unknown parameters, which is essential for practical applications, inherently leads to a loss of power in any hypothesis test. This has been illustrated in Section~\ref{sec:power} for the method proposed here. A relevant avenue for future work would be the exploration of alternative scenarios where the power loss in estimation could be mitigated. One possibility would be to develop a framework inspired by the work of Yun and Foygel Barber~\cite{Yun2023Jan}, in which they consider, under model~\eqref{model_gao}, a test statistic that does not depend on the unknown parameter $\sigma$. This results in a method that is relatively more efficient than the one proposed in~\cite{Gao} in some settings. Adapting this idea to the general model~\eqref{model} would require an appropriate test statistic that does not depend on the unknown parameters $\mathbf{U}$ and $\mathbf{\Sigma}$. However, the direct adaptation of~\cite{Yun2023Jan} to~\eqref{model} presents a non-trivial theoretical challenge while offering limited practical advantages compared to the extension presented here. Indeed, an efficient computation is proposed only for binary partitions of the data. An alternative approach would be the definition of consistent estimators of $\mathbf{U}$ or $\mathbf{\Sigma}$ that are compatible with the selective type I error control. This was studied in~\cite{chen2022selective} in the context of $k$-means clustering under \eqref{model_gao}, where the authors showed that considering a median-based consistent estimator of $\sigma$ yields better performances than the over-estimation strategy proposed in~\cite{Gao}.

Clustering is a multidimensional method that incorporates information from $p$ descriptors to classify $n$ observations. However, the estimated groups are often distinguished by a subset of variables, whose determination is essential in various fields of application~\cite{Ntranos2019,Vandenbon2020}. The framework presented in~\cite{Gao} has also been adapted to feature-level post-clustering inference~\cite{hivert2022post-clustering, Chen2023Nov}, testing for the difference of the $g$-th coordinate of cluster means, for a fixed $g\in\lbrace 1,\ldots,p\rbrace$. In that case, clustering is performed on the complete data set $\mathbf{X}$ but inference is carried out on the $g$-th column, modeled by a $n$-dimensional Gaussian. In a recent contribution~\cite{Chen2023Nov}, the covariance matrix is let arbitrary and $p$-values can be efficiently computed following a similar reasoning as in~\cite{Gao}. Nevertheless, none of these works deal with the estimation of unknown parameters. The extension of the over-estimation strategy presented in Section~\ref{sec:unknown_sigma} to this framework is non-trivial, and would represent a very relevant line for future research.


As discussed in Appendix~\ref{sec:finer_cond}, performing analytically tractable post-clustering inference requires the addition of technical events to the conditioning set, which implies a reduction in power. Investigating whether these conditions might be relaxed is an interesting path for future research. The problem of power loss due to extra conditioning is not exclusive to this method. Techniques like data fission~\cite{datafission} need to calibrate the conditioning information and consequences in terms of power are analogous.  However, it is still unknown whether power loss is more drastic in one method or the other. An interesting contribution would be to establish a framework allowing for a proper comparison of this effect when performing post-clustering inference using data fission and the approach proposed in \cite{Gao}. Nevertheless, extending this comparison to practical applications would be unfeasible as long as the estimation of the covariance structure with statistical guarantees cannot be carried out in both methods.


\section*{Code availability}

The methods introduced in the present work were implemented in the \textsf{R} package \texttt{PCIdep}, available at \href{https://github.com/gonzalez-delgado/PCIdep}{https://github.com/gonzalez-delgado/PCIdep}. All the numerical experiments on synthetic and real data can be reproduced with the code available at \href{https://github.com/gonzalez-delgado/PCIdep-experiments}{https://github.com/gonzalez-delgado/PCIdep-experiments}. The dataset of protein structures used in Section~\ref{sec:proteins} can be downloaded at \href{ https://doi.org/10.5281/zenodo.10021202}{https://doi.org/10.5281/zenodo.10021202}.

\section*{Acknowledgments}

We thank Amin Sagar and Pau Bernadó for providing protein structure data. 

This work was supported by the French National Research Agency (ANR) under grants: ANR-11-LABX-0040 (LabEx CIMI) within the French State Program ``Investissements d’Avenir'', ANR-22-CE45-0003 (CORNFLEX project) and ANR-21-CE40-0007 (GAP project).

\appendix
\counterwithin{figure}{section}

\section{Proofs}

\subsection{Proofs of Section \ref{sec:sel_clustering}}\label{proofs_1}

\subsubsection{Proof of Proposition~\ref{prop_conditions}}\label{proofs_1_lemma}

We begin by recalling a useful established result. We then state and prove Lemma~\ref{lem:caractSigma}, which is essential for the proof of Proposition~\ref{prop_conditions}, presented at the end of the section.

\begin{lemma}[Proposition 3.4 in~\cite{Eaton2007}]\label{thm:independence:Proj:GV}
Let $y \sim \normal(0, \mathbf{S})$ be a $p$-dimensional non-degenerated Gaussian vector and $F\subset \R^p$ a vector subspace. We denote by $\bP_F$ the orthogonal projection on $F$ and by $\bP^\perp_F$ the orthogonal projection on $F^\perp$. Then, $\mathbf{S} \bP_F = \bP_F \mathbf{S}$ if and only if the Gaussian vectors $\bP_F y$ and $\bP^\perp_F y$ are independent.
\end{lemma}

\begin{lemma}\label{lem:caractSigma}
Let $\mathbf{T}$ be a $n\times n$ positive definite symmetric matrix. Then, $\mathbf{T}\in\mathcal{CS}(n)$ if and only if $\nucl{\mG_1}{\mG_2}$ is an eigenvector of $\mathbf{T}$ for all $(\mG_1,\mG_2)\in\setC{[n]}$.
\end{lemma}

\begin{proof}[Proof of Lemma~\ref{lem:caractSigma}] Let $\mathbf{T}=(a-b)\mathbf{I}_n+b\mathbf{1}_{n\times n}\in\mathcal{CS}(n)$. Then, $\mathbf{T}\nucl{\mG_1}{\mG_2}=(a-b)\nucl{\mG_1}{\mG_2}$ as $\mathbf{1}_{n\times n}\nucl{\mG_1}{\mG_2}=\mathbf{0}_n$ for any $(\mG_1,\mG_2)\in \setC{[n]}$. To prove the reciprocal implication, we first define the set
\begin{equation*}
    \mathcal{C}_\mathcal{P}=\lbrace (\mG_1,\mG_2) \,|\, \mG_1,\mG_2 \subset \mP,\, \mG_1\cap\mG_2=\emptyset\rbrace,
\end{equation*}
for any $\mathcal{P}\subset[n]$. Then, we prove the following proposition by induction over $k\geq 2$:
\begin{flalign}\nonumber
\text{For any }\mP\subset[n]\text{ with }2\leq |\mP|\leq k\text{, if }\nucl{\mG_1}{\mG_2}\text{ is an eigenvector of }\mathbf{T} \text{ for all }(\mG_1,\mG_2)\in\mathcal{C}_{\mathcal{P}},\\\text{then the restriction of }\mathbf{T}\text{ on }F_{\mP}:=\spanned \{ \bnu_{\mG_1,\mG_2}\,: \,(\mG_1,\mG_2)\in \setC{\mP}\}\text{ is a uniform scaling}.\tag{$H_k$}
\end{flalign}
\begin{itemize}
    \item \textit{Initialization} ($k=2,3)$. If $\mP= \{p_1,p_2\}$, ($H_2$) holds as $F_{\mP}= \spanned \lbrace \nu_{\{p_1\},\{p_2\}}\rbrace$ and  $\nu_{\{p_1\},\{p_2\}}$ is an eigenvector of $\boldsymbol{T}$. The same strategy yields ($H_3$).
    \item \textit{Induction}. Let ($H_k$) be true for $3<k<n$. Let $\mP\subset[n]$ with $\abs{\mP}=k+1$ and assume that $\nuclst$ is an eigenvector of $\boldsymbol{T}$ for any $(\mG_1,\mG_2)\in \setC{\mP}$. Consider also $\mP_1,\mP_2\subset\mP$ with $\abs{\mP_1}=\abs{\mP_2}=k$ and $\mP_1 \neq \mP_2$. Note that from previous assumptions we have $\mP_1 \cup \mP_2= \mP$. Now, since $\setC{\mP_1}$ and $\setC{\mP_2}$ are subsets of $\setC{\mP}$, property $(H_k)$ ensures that the restrictions of $\boldsymbol{T}$ on $F_{\mP_1}$ and $F_{\mP_2}$ are uniform scalings, that is,
\begin{equation*}
\mathbf{T}_{|F_{\mP_1}} = \lambda_{\mP_1} \bI_{n}
\hspace{0.5cm}\text{ and }\hspace{0.5cm} 
\mathbf{T}_{|F_{\mP_2}} = \lambda_{\mP_2} \bI_{n} \hspace{0.5cm} \text{ for some }  \hspace{0.5cm} \lambda_{\mP_1},\lambda_{\mP_2} \in \R.
\end{equation*}
Moreover, as  $|\mP_1\cap\mP_2|=k-1\geq2$, there exist two distinct elements $i_1$ and $i_2$ in the intersection $\mP_1\cap\mP_2$. Then, $\nu_{\{i_1\},\{i_2\}}\in F_{\mP_1}\cap F_{\mP_2}$. Since $F_{\mP_1}$ and $F_{\mP_2}$ share a non-zero element, we have $\lambda_{\mP_1}=\lambda_{\mP_2}$. We conclude the induction step noticing that $F_{\mP}=F_{\mP_1}+F_{\mP_2}$ (by inclusion and dimensional argument).
\item \textit{Conclusion}. The property ($H_k$) is initialized and inductive, then true for any $2\leq k \leq n$.
\end{itemize}
Following from the previous reasoning, ($H_n$) is true. Then, if $\nucl{\mG_1}{\mG_2}$ is an eigenvector of $\mathbf{T}$ for all $(\mG_1,\mG_2)\in\mathcal{C}_{[n]}$, the restriction $\boldsymbol{T}_{|F_{[n]}}$ is a uniform scaling of parameter $\lambda$. Moreover, as $\boldsymbol{T}$ is symmetric, both $F_{[n]}$ and its orthogonal $F_{[n]}^\perp$ are stable under $\boldsymbol{T}$. It can be easily shown that $F_{[n]}^\perp= \spanned\lbrace\bOne_n\rbrace$. Then, $\bOne_n$ is an eigenvector of $\boldsymbol{T}$, whose associated eigenvalue will be denoted by $\beta$.
Noting that $n^{-1}\bJ=n^{-1}\bOne_n\cdot\bOne_n^T$ is the orthogonal projection over $\spanned\lbrace\bOne_n\rbrace$ and $\bI_n-n^{-1}\bJ$ is the orthogonal projection over $F_{[n]}$, we can write:
\begin{align*}
\mathbf{T} &= \mathbf{T} (n^{-1}\bJ + \bI_n  - n^{-1}\bJ) = \mathbf{T}_{|F_{[n]}^\perp} n^{-1}\bJ + \mathbf{T}_{|F_{[n]}}(\bI_n- n^{-1}\bJ)\\
&= \beta n^{-1}\bJ  + 
\lambda (\bI_n - n^{-1}\bJ) = (\lambda \bI_n + n^{-1}(\beta-\lambda) \bJ)\in\mathcal{CS}(n),
\end{align*}
concluding the proof.
\end{proof}

\begin{proof}[Proof of Proposition~\ref{prop_conditions}]

We start showing the first equivalence in Proposition~\ref{prop_conditions}. Let us denote by $\knu \subset \R^{n\times p}$ the kernel of the linear mapping $\nu_{\mG_1,\mG_2}^T: \bM \in \R^{n\times p} \mapsto \nu_{\mG_1,\mG_2}^T \bM$ and by $\knu^\perp$ its orthogonal complement. We omit their dependence on $\mG_1,\mG_2$ for the sake of a simpler notation. Next, we denote by $\Piknu:=\boldsymbol{\pi}^\perp_{\nu_{\mG_1,\mG_2}}$ and $\Piknuperp:= \bI - \Piknu = \nu_{\mG_1,\mG_2}^T\nuclst/\norm{\nuclst}^2$ the orthogonal projections on $\Lambda$ and $\Lambda^\perp$, respectively. Then, for all $(\mG_1,\mG_2)\in \setC{[n]}$, we have:
\begin{align*}
\mathbf{X}^T\nu_{\mG_1,\mG_2}\independent\boldsymbol{\pi}^\perp_{\nu_{\mG_1,\mG_2}}\mathbf{X}
&
\iff 
\Piknuperp \bX \independent \Piknu\bX
\\
\text{(By Lemma~\ref{thm:independence:Proj:GV})} &
\iff 
(\bSigma\otimes\bU) (\bI_p\otimes\Piknuperp)= (\bI_p\otimes\Piknuperp)(\bSigma\otimes\bU)
\\
&
\iff
\bSigma\otimes\bU\Piknuperp = \bSigma\otimes\Piknuperp\bU\\
\text{(By injectivity of the mapping }\bM \mapsto \bSigma\otimes \bM)
&
\iff \bU\Piknuperp = \Piknuperp\bU\\
& \iff
\text{The eigenspaces of $\Piknuperp$ are stable under $\bU$}\\
& \iff
\operatorname{span} \nuclst \text{ and } \nuclst^\perp \text{ are stable under $\bU$}\\
& \iff
\text{$\nuclst$ is an eigenvector of $\bU$.}
\end{align*}
In the last equivalence, we consider the matrix $\Piknuperp$ as a linear operator on $\R^{n\times 1}$. As this holds for every $(\mG_1, \mG_2)$ in $\setC{[n]}$, equivalence $(i)$ in Proposition~\ref{prop_conditions} follows directly from Lemma~\ref{lem:caractSigma}.

The second equivalence in Proposition~\ref{prop_conditions} is a consequence of the following well-known result. For any $p$-dimensional Gaussian vector $z\sim \normal(\mu,\bA)$,
\begin{equation}\label{direction_equivalence}
\norm{z}_2 \independent \mathrm{dir}(z) \iff
\mu=0 \text{ and } \bA=\lambda \bI_p \text{ for some } \lambda >0. 
\end{equation}
Let $y\sim\mathcal{N}(0,\mathbf{S})$ be a $p$-dimensional Gaussian vector and consider $z=\sqrt{\mathbf{A}^{-1}}y$, for any $p\times p$ positive definite matrix $\mathbf{A}$. Then, $\norm{y}_{\mathbf{A}}= \norm{z}_2$ and $z \sim \normal (0,\sqrt{\mathbf{A}^{-1}}\mathbf{S}\sqrt{\mathbf{A}^{-1}})$. Consequently, we have:
\begin{eqnarray*}
\|y\|_{\bA}\independent \mathrm{dir}_{\bA}(y) 
&\iff&
\norm{z}_2 \independent \frac{y}{\norm{z}_{2}} \\
\text{(}\bM \mapsto \sqrt{\bA^{-1}} \bM \text{ is a one-to-one mapping)} &\iff&
\norm{z}_2 \independent \mathrm{dir}(z)\\
\text{(Equivalence~\eqref{direction_equivalence})}&\iff &
\sqrt{\bA^{-1}}\mathbf{S}\sqrt{\bA^{-1}} = \lambda \bI, \text{ for some }\lambda >0.\\
\text{(}\bM \mapsto \sqrt{\bM^{-1}}\mathbf{S}\sqrt{\bM^{-1}} \text{ is one-to-one on positive matrices}) &\iff &
\bA= \lambda\mathbf{S},\text{ for some }\lambda >0.
\end{eqnarray*}
Setting $\mathbf{S}= \mathbf{V}_{\mG_1,\mG_2}$ yields the result.
\end{proof}

\subsubsection{Proofs of Section~\ref{sec:CS_known_U}}\label{proofs_1_CS}

\begin{proof}[Proof of Theorem~\ref{th:pvalue_V}]
We follow the steps of the proof of Theorem 1 in \cite{Gao}. We begin by deriving the null distribution of the test statistic $\norm{\mathbf{X}^T\nu}_{\mathbf{V}_{\mathcal{G}_1,\mathcal{G}_2}}$ under~\eqref{h0}. First, from~\cite[Theorem 2.3.10]{Gupta2018}, we have:
\begin{equation}
    \mathbf{X}^T\nu\overset{H_0^{\lbrace \mathcal{G}_1,\mathcal{G}_2\rbrace}}{\sim}\mathcal{N}_{p}(0,\mathbf{V}_{\mathcal{G}_1,\mathcal{G}_2}),
\end{equation}
which yields
\begin{equation}\label{chi_stat}
   \norm{\mathbf{X}^T\nu}_{\mathbf{V}_{\mathcal{G}_1,\mathcal{G}_2}}\overset{H_0^{\lbrace \mathcal{G}_1,\mathcal{G}_2 \rbrace}}{\sim}\chi_p,
\end{equation}
where the norm $\norm{\cdot}_{\mathbf{V}_{\mathcal{G}_1,\mathcal{G}_2}}$ is defined in~\eqref{norm_V}.
Let us now build the $p$-value for $H_0^{\lbrace \mathcal{G}_1,\mathcal{G}_2\rbrace}$, by slightly adapting the reasoning in \cite{Gao}. On one hand, for any $\nu\in\mathbb{R}^n$, we have
\begin{equation}\label{X_decomposed}
    \mathbf{X}=\boldsymbol\pi_\nu^\perp\mathbf{X}+(\mathbf{I}_n-\boldsymbol\pi_\nu^\perp\mathbf{X})=\boldsymbol\pi_\nu^\perp\mathbf{X}+\left(\frac{\norm{\mathbf{X}^T\nu}_{\mathbf{V}_{\mathcal{G}_1,\mathcal{G}_2}}}{\norm{\nu}_2^2}\right)\nu\,\mathrm{dir}_{\mathbf{V}_{\mathcal{G}_1,\mathcal{G}_2}}(\mathbf{X}^T\nu)^T.
\end{equation}
On the other hand, from Proposition~\ref{prop_conditions} we have $\boldsymbol\pi_\nu^\perp\mathbf{X}\independent \mathbf{X}^T\nu$, which implies $\norm{\mathbf{X}^T\nu}_{\mathbf{V}_{\mathcal{G}_1,\mathcal{G}_2}}\independent\boldsymbol\pi_\nu^\perp\mathbf{X}$, and $\norm{\mathbf{X}^T\nu}_{\mathbf{V}_{\mathcal{G}_1,\mathcal{G}_2}}\independent \mathrm{dir}_{\mathbf{V}_{\mathcal{G}_1,\mathcal{G}_2}}(\mathbf{X}^T\nu)$.
We can now plug \eqref{X_decomposed} in the definition of~\eqref{pvalue_V} and, taking into account the previous independence relationships, we can write:
\begin{align}\label{pv_proof}
    p_{\mathbf{V}_{\mathcal{G}_1,\mathcal{G}_2}}(\mathbf{x};\lbrace \mathcal{G}_1,\mathcal{G}_2\rbrace)=\mathbb{P}_{H_0^{\lbrace \mathcal{G}_1, \mathcal{G}_2\rbrace}}\left(\norm{\mathbf{X}^T\nu}_{\mathbf{V}_{\mathcal{G}_1,\mathcal{G}_2}}\geq \norm{\mathbf{x}^T\nu}_{\mathbf{V}_{\mathcal{G}_1,\mathcal{G}_2}}\,\, \right.\bigr\rvert\nonumber\\
    \left. \norm{\mathbf{X}^T\nu}_{\mathbf{V}_{\mathcal{G}_1,\mathcal{G}_2}}\in \,\mathcal{S}_{\mathbf{V}_{\mathcal{G}_1,\mathcal{G}_2}}(\mathbf{x};\lbrace \mathcal{G}_1,\mathcal{G}_2\rbrace)  \right),
\end{align}
where the set $\mathcal{S}_{\mathbf{V}_{\mathcal{G}_1,\mathcal{G}_2}}(\mathbf{x};\lbrace \mathcal{G}_1,\mathcal{G}_2\rbrace)$ is defined in \eqref{set_S}. Consequently, if we denote by $\mathbb{F}_p(t,\mathcal{S})$ the cumulative distribution function of a $\chi_p$ random variable truncated to the set $\mathcal{S}$, from \eqref{pv_proof} and \eqref{chi_stat} we have
\begin{equation}
    p_{\mathbf{V}_{\mathcal{G}_1,\mathcal{G}_2}}(\mathbf{x};\lbrace \mathcal{G}_1,\mathcal{G}_2\rbrace) = 1 - \mathbb{F}_p\left(\norm{\mathbf{x}^T\nu}_{\mathbf{V}_{\mathcal{G}_1,\mathcal{G}_2}}\,,\,\mathcal{S}_{\mathbf{V}_{\mathcal{G}_1,\mathcal{G}_2}}(\mathbf{x};\lbrace \mathcal{G}_1,\mathcal{G}_2\rbrace)\right),
\end{equation}
which proves the first statement \eqref{pvalue_F}. The control of selective type I error is proved identically to the reasoning in the proof of \cite[Theorem 1]{Gao}. 
\end{proof}

\begin{proof}[Proof of Lemma \ref{equivalence_sets}] Let us first show that the perturbed data sets $\mathbf{x}'(\phi)$, defined in \cite[Equation (13)]{Gao} and $\mathbf{x}'_{\mathbf{V}_{\mathcal{G}_1,\mathcal{G}_2}}(\phi)$, defined in \eqref{perturbed_x} are the same up to a scale transformation, i.e. that
\begin{equation}\label{scale_transformation}
    \mathbf{x}'_{\mathbf{V}_{\mathcal{G}_1,\mathcal{G}_2}}(\phi)=\mathbf{x}'\left(\frac{\norm{\mathbf{x}^T\nu}_2}{\norm{\mathbf{x}^T\nu}_{\mathbf{V}_{\mathcal{G}_1,\mathcal{G}_2}}}\,\phi\right)\quad\forall\,\phi\geq 0.    
\end{equation}
Note first that we can write
\begin{equation}\label{eq_proof_remark_1}
    \left(\frac{\norm{\mathbf{x}^T\nu}_2}{\norm{\mathbf{x}^T\nu}_{\mathbf{V}_{\mathcal{G}_1,\mathcal{G}_2}}}\,\phi-\norm{\mathbf{x}^T\nu}_{2}\right)\mathrm{dir}(\mathbf{x}^T\nu)=\left(\phi -\norm{\mathbf{x}^T\nu}_{\mathbf{V}_{\mathcal{G}_1,\mathcal{G}_2}}\right)\mathrm{dir}_{\mathbf{V}_{\mathcal{G}_1,\mathcal{G}_2}}(\mathbf{x}^T\nu),
\end{equation}
where $\mathrm{dir}(u)=u/\norm{u}_2\mathds{1}\lbrace u\neq 0\rbrace$. Replacing \eqref{eq_proof_remark_1} in \eqref{perturbed_2}, we have \eqref{scale_transformation}. Finally, it suffices to remark that
\begin{align*}
    \hat{\mathcal{S}}_{\mathbf{V}_{\mathcal{G}_1,\mathcal{G}_2}}=\left\lbrace \phi\geq 0\,:\,\mG_1,\mG_2\in\mathcal{C}\left(\mathbf{x}'_{\mathbf{V}_{\mathcal{G}_1,\mathcal{G}_2}}(\phi)\right)\right\rbrace=\left\lbrace \phi\geq 0\,:\,\mG_1,\mG_2\in\mathcal{C}\left(\mathbf{x}'\left(\frac{\norm{\mathbf{x}^T\nu}_2}{\norm{\mathbf{x}^T\nu}_{\mathbf{V}_{\mathcal{G}_1,\mathcal{G}_2}}}\,\phi\right)\right)\right\rbrace\nonumber\\
    =\left\lbrace \frac{\norm{\mathbf{x}^T\nu}_{\mathbf{V}_{\mathcal{G}_1,\mathcal{G}_2}}}{\norm{\mathbf{x}^T\nu}_2}\,\phi\,:\,\mG_1,\mG_2\in\mathcal{C}(\mathbf{x}'(\phi))\right\rbrace=\frac{\norm{\mathbf{x}^T\nu}_{\mathbf{V}_{\mathcal{G}_1,\mathcal{G}_2}}}{\norm{\mathbf{x}^T\nu}_2}\,\hat{\mathcal{S}},
\end{align*}
which concludes the proof.
\end{proof}

\subsubsection{Proofs of Section~\ref{sec:general_known_U}}\label{proofs_1_general}

We start by stating some technical results that are needed for the proof of Theorem~\ref{th:null_dist_XC}. In what follows, we will use the notation $\mathbf{S}$ to denote both a $d\times d$ real matrix and its associated linear mapping, that is, the map $\mathbf{S}\,:\,\mathbb{R}^d\rightarrow \mathbb{R}^d$ such that $\mathbf{S}(y)=\mathbf{S}y$ for all $y\in\mathbb{R}^d$. For any vector subspace $F\subset\mathbb{R}^d$, we will denote by $\bPi_F$ the orthogonal projection onto $F$.

\begin{theorem} [Proposition 3.13 in~\cite{Eaton2007}]\label{thm:cond:gaussian}
Let $\bX \sim \normal(\bmu,\bSigma)$ be a Gaussian vector in $\R^n$, let $\bA$ be a matrix in $M_{p,n}(\R)$. Then, the conditional vector $(\bX|\bA\bX=\by)$ is a Gaussian vector satisfying  
\begin{equation}
    (\bX | \bA\bX = \by)\sim \normal(\bmu + \bSigma \bA^T(
    \bA\bSigma \bA^T)^{\dagger}(\by-\bA\bmu),\bSigma-\bSigma \bA^T(\bA\bSigma \bA^T)^{\dagger}\bA\bSigma),
\end{equation}
where $(\bA\bSigma \bA^T)^{\dagger}$ is the Moore-Penrose pseudoinverse of the matrix $\bA\bSigma \bA^T$.
\end{theorem}

\begin{lemma}
\label{lemma:proj:prod:spaces:general}
Let $F$, $G$ be two orthogonal subspaces of $\R^d$. For any full-rank symmetric matrix $\mathbf{S}\in\mathbb{R}^{d\times d}$, let $\mathbf{S}_{F,G}$ be the $2d\times 2d$ matrix: 
\begin{equation}\label{SFG}
\mathbf{S}_{F, G} := 
(\bPi_F,\bPi_G)^T \mathbf{S}\, (\bPi_F,\bPi_G) =
\begin{bNiceArray}{cc}
\bPi_F \mathbf{S} \bPi_F
&
\bPi_F \mathbf{S} \bPi_G\\
\bPi_G \mathbf{S} \bPi_F
&
\bPi_G \mathbf{S} \bPi_G
\end{bNiceArray}.
\end{equation}
Then, the range of the linear mapping associated to $\mathbf{S}_{F,G}$ is given by:
\begin{equation}
\label{eq:range:SigmaFG}
\range(\mathbf{S}_{F,G}) = \tilde{F}\times\tilde{G},\quad\textrm{with}\quad\tilde{F}=F \cap \mathbf{S}(F\oplus G)\quad\textrm{and}\quad\tilde{G}=G \cap \mathbf{S}(F\oplus G).
\end{equation}
Moreover, the restriction of $\mathbf{S}_{F,G}$ to its range is a one-to-one mapping whose inverse is given by:
\begin{equation}
\label{eq:inverse:SigmaFG}
\mathbf{S}_{F,G}^{-1} (u,v)^T 
= 
\left(\bPi_{\tilde{F}} \circ\mathbf{S}^{-1} (u+v),\bPi_{\tilde{G}} \circ \mathbf{S}^{-1} (u+v)\right)^T,\quad\forall\,(u,v)\in\range(\mathbf{S}_{F,G}).
\end{equation}
\end{lemma}

\begin{proof}[Proof of Lemma~\ref{lemma:proj:prod:spaces:general}]
We show~\eqref{eq:range:SigmaFG} using double inclusion.  The following reasoning shows that $\range(\mathbf{S}_{F,G})\subset \tilde{F}\times \tilde{G}$. 
\begin{eqnarray*}
\mathbf{S}_{F,G}(\R^d\times \R^d)
&\subset&
(\bPi_{F} \mathbf{S} \bPi_{F})(\R^d) +(\bPi_{F} \mathbf{S} \bPi_{G})(\R^d)\times (\bPi_{G} \mathbf{S} \bPi_{F})(\R^d) +(\bPi_{G} \mathbf{S} \bPi_{G})(\R^d)\\
&\subset&
\bPi_{F} \mathbf{S} (F\oplus G)\times \bPi_{G} \mathbf{S} (F\oplus G)\\
&\subset&
F \cap \mathbf{S}(F\oplus G) \times G \cap \mathbf{S}(F\oplus G)= \tilde{F}\times \tilde{G}.
\end{eqnarray*}
We show now the reciprocal inclusion. Letting $(u,v)\in \tilde{F}\times \tilde{G}$, we have:
\begin{eqnarray*}
\mathbf{S}_{F,G}\left(\mathbf{S}^{-1} (u+v),\mathbf{S}^{-1} (u+v)\right)^T
&=& 
\Big(\bPi_{F} \mathbf{S} \bPi_{F} \mathbf{S}^{-1} (u+v) +
\bPi_{F} \mathbf{S} \bPi_{G} \mathbf{S}^{-1} (u+v), \\
& &
\bPi_{F} \mathbf{S} \bPi_{F} \mathbf{S}^{-1} (u+v) +
\bPi_{F} \mathbf{S} \bPi_{G} \mathbf{S}^{-1} (u+v)\Big)\\
\big(\text{As } F\perp G: \, \bPi_{F}+\bPi_{G} = \bPi_{F\oplus G} \big)
&=&
\big(\bPi_{F} \mathbf{S} \bPi_{F\oplus G} \mathbf{S}^{-1} (u+v), \,
\bPi_{G} \mathbf{S} \bPi_{F\oplus G} \mathbf{S}^{-1} (u+v)\big)\\
\big(\text{As } u+v \in \mathbf{S} (F\oplus G) \big)
&=& 
\big(\bPi_{F} \mathbf{S} (\mathbf{S}^{-1} (u+v)),\, 
\bPi_{G} \mathbf{S} (\mathbf{S}^{-1} (u+v))\big)\\
&=& 
\big(\bPi_{F} (u+v), \,
\bPi_{G} (u+v) \big)\\
\big(\text{As } u\in F\text{ and } v\in G \big)
&=& 
(u,v).
\end{eqnarray*}
Consequently, we have $\tilde{F}\times \tilde{G}\subset \range(\mathbf{S}_{F,G})$. We conclude by showing~\eqref{eq:inverse:SigmaFG}. Following from the fact that the range of the linear mapping associated to any symmetric matrix is orthogonal to its kernel, we have that $\mathbf{S}_{F,G}=\mathbf{S}_{F,G}\circ\bPi_{\range{(\mathbf{S}_{F,G}})}=\mathbf{S}_{F,G}\circ\bPi_{\tilde{F}\times \tilde{G}}$. This, together with the fact that $\bPi_{\tilde{F}\times \tilde{G}} = (\bPi_{\tilde{F}},\bPi_{\tilde{G}})$, yields:
\[
\mathbf{S}_{F,G} \left(\mathbf{S}^{-1} (u+v),\mathbf{S}^{-1} (u+v)\right)^T= \mathbf{S}_{F,G} \left(\bPi_{\tilde{F}} \circ\mathbf{S}^{-1} (u+v),\bPi_{\tilde{G}} \circ \mathbf{S}^{-1} (u+v)\right)^T,
\]
for all $(u,v)\in\range{(\mathbf{S}_{F,G}})$, which concludes the proof.
\end{proof}

\begin{lemma}
\label{cor:inclusion:prod:space}
Let $F,F',G,G'$ be subspaces of $\R^d$ such that $F'\subset F$, $G'\subset G$ and $F \perp G$. For any symmetric matrix $\mathbf{S} \in \R^{d\times d}$, let $\mathbf{S}_{F,G}$ be the one defined in~\eqref{SFG} and let $\mathbf{S}_{F',G'}$ be defined analogously. Then, the following inclusions hold: 
\begin{itemize}
\item[$(i)$] $\range(\mathbf{S}_{F',G'}) \subset \range(\mathbf{S}_{F,G})$,
\item[$(ii)$] $\mathbf{S}_{F',G'}^\dagger(0_d\times \R^d)
\subset 
\mathbf{S}_{F,G}^\dagger(0_d\times \R^d)$,
\item[$(iii)$] $\mathbf{S}_{F',G'}^\dagger(\R^d \times 0_d)
\subset \mathbf{S}_{F,G}^\dagger(\R^d\times 0_d)$,
\end{itemize}
where $\mathbf{A}^\dagger(\cdot)$ is the linear mapping associated to the Moore-Penrose pseudo-inverse of a matrix $\mathbf{A}$.
\end{lemma}

\begin{proof}[Proof of Lemma~\ref{cor:inclusion:prod:space}]
We start by showing $(i)$. As, by hypothesis, we have:
\[
F' \cap \mathbf{S}(F'\oplus G') 
\subset
F \cap \mathbf{S}(F\oplus G) 
\quad \text{and} \quad 
G' \cap \mathbf{S}(F'\oplus G') 
\subset 
G \cap \mathbf{S}(F\oplus G), 
\]
Equation \eqref{eq:range:SigmaFG} yields $\range(\mathbf{S}_{F',G'}) \subset \range(\mathbf{S}_{F,G})$. Let us show $(ii)$. In the following, we will write $\range(\mathbf{S}_{F',G'})= \tilde{F}'\times \tilde{G}'$ and $\range(\mathbf{S}_{F,G})= \tilde{F}\times \tilde{G}$, as in~\eqref{eq:range:SigmaFG}. Inclusion $(i)$ implies:
\begin{equation}\label{inclusions_lemma}
\tilde{F}'\subset \tilde{F} \quad \text{and} \quad \tilde{G}'\subset \tilde{G}.
\end{equation}
As the pseudo-inverse of a symmetric matrix can be written as the composition of the orthogonal projection onto its range with its inverse on its range, Equation~\eqref{eq:inverse:SigmaFG} yields:
\begin{eqnarray*}
\mathbf{S}_{F',G'}^\dagger(0_d\times \R^d) 
&=&
\mathbf{S}_{F',G'}^{-1}\circ \bPi_{\tilde{F}'\times \tilde{G}'}\,(0_d \times \R^d)\\
&=&
\mathbf{S}_{F,G'}^{-1}(0_d \times \tilde{G}')\\
&=&
(\bPi_{\tilde{F}'},\bPi_{\tilde{G}'})(\mathbf{S}^{-1} \tilde{G}').
\end{eqnarray*}
Following the same reasoning we can show that $\mathbf{S}_{F,G}^\dagger(0_d\times \R^d) = (\bPi_{\tilde{F}},\bPi_{\tilde{G}})(\mathbf{S}^{-1} \tilde{G})$, so if we prove that
\begin{equation}\label{last_inclusion_jpp}
(\bPi_{\tilde{F}'},\bPi_{\tilde{G}'})(\mathbf{S}^{-1} \tilde{G}') \subset (\bPi_{\tilde{F}},\bPi_{\tilde{G}})(\mathbf{S}^{-1} \tilde{G}),
\end{equation}
inclusion $(ii)$ will follow. Let $(h_{\tilde{F}'},h_{\tilde{G}'}) \in (\bPi_{\tilde{F}'},\bPi_{\tilde{G}'})(\mathbf{S}^{-1} \tilde{G}')$ and let $h = h_{\tilde{F}'} + h_{\tilde{G}'}$. Thus, $h\in(\tilde{F}'\oplus \tilde{G}')\cap (\mathbf{S}^{-1} \tilde{G}')\subset(\tilde{F}\oplus \tilde{G})\cap (\mathbf{S}^{-1} \tilde{G})$. From the unicity of the decomposition in $\tilde{F}\oplus\tilde{G}$ and~\eqref{inclusions_lemma}, we have 
$(\bPi_{\tilde{F}},\bPi_{\tilde{G}})(h)=(h_{\tilde{F}'},h_{\tilde{G}'})$, which yields~\eqref{last_inclusion_jpp}. The reasoning to prove $(iii)$ is identical.
\end{proof}

We are now ready to prove Theorem~\ref{th:null_dist_XC} and Proposition~\ref{prop:gamma_chi}. Throughout the following proofs we will manipulate two intrinsically similar vector spaces, $\R^{n\times p }$ and $\R^{np\times 1}$, that are identified through the isometry $\text{vec}: \R^{n\times p} \to \R^{np\times 1}$. For the sake of a simpler notation, we will write $\widetilde{F} = \vect{F}$ for any vector space $F\subset\mathbb{R}^{n\times p}$. Then, the orthogonal projections onto $\tilde{F}$ and $F$ are identified via the equality $\bPi_{\tilde{F}}=\bI_p\otimes \bPi_{F}$.


\begin{proof}[Proof of Theorem~\ref{th:null_dist_XC}] We prove Theorem~\ref{th:null_dist_XC} in two steps. First, we show that the conditioned vector~\eqref{cond_mean_vector} has a $p$-dimensional normal distribution under~\eqref{h0}, explicitly deriving its mean and covariance matrix. Then, we will show that such distribution is centered. To shed light on the objects introduced in this proof, we keep the notation of the proof of Proposition~\ref{prop_conditions}. In particular, we denote by $\Lambda\subset\mathbb{R}^{n\times p}$ the kernel of the linear mapping $\nu^T: \bM \in \R^{n\times p} \mapsto \nu^T \bM$ and by $\knu^\perp$ its orthogonal complement. This means that $\Piknu=\boldsymbol{\pi}^\perp_{\nu}$ and $\Piknuperp= \boldsymbol{\pi}_{\nu}$, respectively. The idea is to find a matrix $\bA_\bx$ and a vector $y_\bx$ such that the conditioned vector
\begin{equation}\label{vec_cond}
    \vect{\bX}\,|\, \lbrace \bPi_\Lambda\mathbf{X} = \bPi_\Lambda\mathbf{x}, \mathrm{dir}(\nu^T\mathbf{X})=\pm\mathrm{dir}(\nu^T\mathbf{x})\rbrace
\end{equation}
can be rewritten as $\vect{\bX}\,|\,\lbrace\mathbf{A}_\bx\vect{\bX}=y_\bx\rbrace$. Then, applying Theorem~\ref{thm:cond:gaussian} would yield an explicit Gaussian distribution for
\begin{equation}
  (\bI_p\otimes\nu)\vect{\bX}\,|\,\lbrace \mathbf{A}_\bx\vect{\bX}=y_\bx\rbrace =  \nu^T\bX\,|\,\lbrace \mathbf{A}_\bx\vect{\bX}=y\rbrace=\bar{\bX}_\nu(\bx).
\end{equation}

We start by rewriting the condition $\mathrm{dir}(\nu^T\mathbf{X})=\pm\mathrm{dir}(\nu^T\mathbf{x})$ as follows. First, we have:
\begin{eqnarray*}
\mathrm{dir}({\nu^T}\bX)=\pm\mathrm{dir}(\nu^T\bx) 
&\iff& 
\nu^T\bX \in \spanned(\nu^T\bx) \\
&\iff& 
\bX \in  V_{\bx} := \left(\knu \oplus \spanned(\bx)\right)\\
&\iff& \vect{\bX}\in \tilde{V}_{\bx} = \tilde{\Lambda}\oplus \spanned(\vect{\bx})\\
&\iff& 
\bPi_{\tilde{V}_{\bx}^\perp} \vect{\bX}= 0.
\end{eqnarray*}
Writing ${\bx_\nu}:=\vect{\bPi_{\Lambda^\perp}\bx}$, we have that
\begin{equation*}
\tilde{V}_\bx^\perp = \left(\tilde{\Lambda} \oplus \spanned(\vect{\bx})\right)^\perp =
\left(\tilde{\Lambda} \oplus \spanned({\bx_\nu})\right)^\perp
=
\tilde{\Lambda}^\perp \cap {\bx_\nu^\perp},
\end{equation*}
where ${\bx_\nu^\perp}$ denotes the orthogonal complement of ${\bx_\nu}$. Since ${\bx_\nu}\in\tilde{\Lambda}^\perp$, we can write $\bPi_{\tilde{V}_{\bx}^\perp}= \bPi_{{\bx_\nu^\perp}}\circ\bPi_{\tilde{\Lambda}^\perp}$. This yields:
\begin{equation}\label{eq:mat:direct:equiv}
\mathrm{dir}({\nu}^T\bX)=\pm\mathrm{dir}({\nu}^T\bx)  \iff
\bigl(\bPi_{{\bx_\nu^\perp}} \circ \bPi_{\tilde{\Lambda}^\perp}\bigr)\vect{\bX} = \bPi_{{\bx_\nu^\perp}}\bigl(\bI_p\otimes \bPi_{\Lambda^\perp}\bigr)\vect{\bX}=0. 
\end{equation}
Finally, using that
\begin{equation}\label{eq:mat:interp:equiv}
\Piknu\bX=\Piknu\bx \iff (\bI_p\otimes \bPi_\Lambda)\vect{\bX}=(\bI_p\otimes \bPi_\Lambda)\vect{\bx},
\end{equation}
we can characterize the conditioning set in~\eqref{vec_cond} as follows:
\begin{equation}
\mathrm{dir}({\nu}\bX)= \pm \mathrm{dir}({\nu}\bx) \text{ and } \Piknu\bX=\Piknu\bx\iff
{\bA_\bx}\vect{\bx} = y_\bx,
\end{equation}
where ${\bA_\bx}$ and $y_\bx$ are defined as
\begin{equation}\label{eq:bAx:by}
{\bA_\bx} =
\begin{bNiceArray}{c}
\bPi_{{\bx_\nu^\perp}}(\bI_p\otimes\bPi_{\Lambda^\perp})\\
\bI_p\otimes\bPi_{\Lambda} 
\end{bNiceArray}, 
\hspace{1cm}
y_\bx= \begin{bNiceArray}{c}
0_{np}\\
(\bI_p\otimes\bPi_{\Lambda})\vect{\bx}
\end{bNiceArray},
\end{equation}
and $\bPi_{\bx_\nu^\perp}$ corresponds to the object $\boldsymbol{\pi}^\perp_{\bx_\nu}$ defined in Theorem~\ref{th:null_dist_XC}. Finally, using Theorem~\ref{thm:cond:gaussian} and the properties of the multivariate Gaussian distribution, we have that
\begin{equation*}
   \bar{\bX}_\nu(\bx)\sim\mathcal{N}_p(\bar{\mu}_\nu(\bx), \bGamma_\bx),
\end{equation*}
where
\begin{equation}
    \bar{\mu}_\nu(\bx)=(\bI_p\otimes\nu^T)\bigl(\vect{\bmu} + (\bSigma\otimes\bU){\bA_\bx^T}({\bA_\bx} (\bSigma\otimes\bU) {\bA_\bx^T})^\dagger(y_\bx-{\bA_\bx}\vect{\bmu})\bigr),
\end{equation}
and $\bGamma_\bx$ is defined in~\eqref{gamma_x}.

We conclude by showing that $\bar{\mu}_\nu(\bx)=0_p$ under~\eqref{h0} for all $\bx\in\mathbb{R}^{n\times p}$. In what follows, we assume that~\eqref{h0} holds. First, note that \eqref{h0} implies
\begin{equation}
\label{eq:nHequiv}
(\bI_p\otimes \nu^T)\vect{\bmu} =0_{np}\quad\textrm{and}\quad
{\bA_\bx}\vect{\bmu} = (0_{np},\vect{\bmu})^T,
\end{equation} 
yielding
\begin{equation}
    y_\bx-{\bA_\bx}\vect{\bmu} = (0_{np},\,(\bI_p\otimes\bPi_{\Lambda})\vect{\bx} - \vect{\bmu}). 
\end{equation}
Consequently, proving $\bar{\mu}_\nu(\bx)=0$ comes down to show that $0_{np}\times \R^{np}$ is included in the kernel of the linear operator defined by the matrix
\begin{equation}
   (\bI_p\otimes\nu^T)(\bSigma\otimes\bU){\bA_\bx^T}({\bA_\bx} (\bSigma\otimes\bU) {\bA_\bx^T})^\dagger,
\end{equation}
or, equivalently, in the kernel of the linear operator associated to
\begin{equation}\label{eq:muc:equiv1} 
\bPi_{\tilde{\Lambda}^\perp}(\bSigma\otimes\bU){\bA_\bx^T}({\bA_\bx} (\bSigma\otimes\bU) {\bA_\bx^T})^\dagger.
\end{equation}
Let us consider the matrix $\bA = (\bPi_{\tilde{\Lambda}^\perp},
\bPi_{\Tilde{\Lambda}})^T$. Then, if the following statements hold:
\begin{enumerate}
\item[(S1)] $\bPi_{\tilde{\Lambda}^\perp}(\bSigma\otimes\bU) \bA^T(\bA(\bSigma\otimes\bU) \bA^T)^\dagger (0_{np}\times \R^{np}) = 0_{np}$,
\item[(S2)] $\bPi_{\tilde{\Lambda}^\perp}(\bSigma\otimes\bU) {\bA_\bx^T}(\bA(\bSigma\otimes\bU) {\bA_\bx^T})^\dagger (0_{np}\times \R^{np})
\subset
\bPi_{\tilde{\Lambda}^\perp}(\bSigma\otimes\bU) \bA^T({\bA_\bx}(\bSigma\otimes\bU) \bA^T)^\dagger (0_{np}\times \R^{np})$,
\end{enumerate}
the subspace $0_{np}\times \R^{np}$ is included in the kernel of~\eqref{eq:muc:equiv1} and the result follows.

Since $\bPi_{\tilde{\Lambda}^\perp}$ is a sub-block of $\bA$, (S1) is equivalent to the equality:
\begin{equation}
\label{eq:projection:Sigma:knu}
\bA(\bSigma\otimes\bU) \bA^T(\bA(\bSigma\otimes\bU) \bA^T)^\dagger (0_{np}\times \R^{np})  = (0_{np},V)^T,
\end{equation}
for a subspace $V \subset \R^{np}$. From the properties of the Moore-Penrose pseudo-inverse, we have that
\begin{equation*}
\bA(\bSigma\otimes\bU) \bA^T(\bA(\bSigma\otimes\bU) \bA^T)^\dagger = \bPi_{\range(\bA(\bSigma\otimes\bU) \bA^T)}.
\end{equation*}
This, together with Lemma~\ref{lemma:proj:prod:spaces:general}, yields~\eqref{eq:projection:Sigma:knu}.

To prove (S2), it suffices to show that:
\[
{\bA_\bx^T}({\bA_\bx}(\bSigma\otimes\bU) {\bA_\bx^T})^\dagger(0_{np}\times \R^{np})
\subset
\bA^T(\bA(\bSigma\otimes\bU) \bA^T)^\dagger (0_{np}\times \R^{np}).
\]
Inclusion $(ii)$ in Lemma~\ref{cor:inclusion:prod:space} yields:
\begin{equation*}
({\bA_\bx}(\bSigma\otimes\bU) {\bA_\bx^T})^\dagger(0_{np}\times \R^{np})
\subset
(\bA(\bSigma\otimes\bU) \bA^T)^\dagger (0_{np}\times \R^{np}).
\end{equation*}
Finally, following the same strategy as in the proof of Lemma~\ref{cor:inclusion:prod:space}, we can show that the previous inclusion is stable when composed by ${\bA_\bx}$ on the left side and $\bA$ on the right side, which yields (S2).
\end{proof}

\begin{proof}[Proof of Proposition~\ref{prop:gamma_chi}] We keep the notation of the proof of Theorem~\ref{th:null_dist_XC}. To show~\eqref{pvalue_gamma_chi1}, the key idea is to prove that the rank of $\bGamma_\bx$ is $\bx$-a.s. constant equal to one. From the condition $\mathrm{dir}(\nu^T \bX)=\pm\mathrm{dir}(\nu^T \bx)$, we have clearly that the rank of $\bGamma_\bx$ is upper bounded by one. Moreover, since the mapping $\nu : \R^{1\times p}\to \R^{n\times p}$ defined by $z \mapsto \nu z$ is injective, the rank of the covariance matrix of $\bar{\bX}_\nu(\bx)$ is the same as the rank of $\nu\bar{\bX}_\nu(\bx)$ and, from the proof of Theorem~\ref{th:null_dist_XC}, the same as the rank of the matrix
\begin{equation*}
\bPiVperp \left(\bX\, |\, \lbrace\Piknu\bX=\Piknu\bx,\,
\bPi_{\bx_\nu^\perp}\bPiVperp\bX = \mathbf{0}\rbrace\right).
\end{equation*}
Following the steps of the proof of Theorem~\ref{thm:cond:gaussian} (Proposition 3.13 in~\cite{Eaton2007}), we can decompose $\bPiVperp \bX$ as the sum of two independent Gaussian vectors $\bY$ and $\mathbf{Z}$, with  with $\bY=\bPiV \bX + \bPi_{\bx_\nu^\perp}\bPiVperp\bX$. Thus, $\mathbf{Z}$ must be non-zero since otherwise $\bPi_{\bx_\nu}\bX=0$, and $\bX$ is non degenerated. As $\bGamma_\bx$ is the covariance matrix of $\mathbf{Z}$, its rank is $\bx$-a.s. equal to one. This implies that $\norm{\bar{\bX}_\nu(\bx)}_{\bGamma_\bx}\sim\chi_1$ $\bx$-a.s. under~\eqref{h0}, where $\norm{\cdot}_{\bGamma_\bx}$ is defined in~\eqref{norm_gamma_x}. Following the same steps as in the proofs of Theorem~\ref{th:pvalue_V} and Lemma~\ref{equivalence_sets}, we have~\eqref{pvalue_gamma_chi1} and~\eqref{S_gamma}.
\end{proof}

\subsection{Proofs of Section~\ref{sec:unknown_sigma}}\label{proofs_2}

\begin{proof}[Proof of Theorem~\ref{th:over_estimate}] We follows the steps of the proof of Theorem 4 in \cite{Gao}. For simplicity, we use $\hat{p}_n$ to denote $p_{\hat{\mathbf{V}}_{\mG_1^{(n)},\mG_2^{(n)}}}\bigl(\mathbf{X}^{(n)};\bigl\lbrace \mG_1^{(n)},\mG_2^{(n)}\bigr\rbrace\bigr)$, $p_n$ to denote $p_{\mathbf{V}_{\mG_1^{(n)},\mG_2^{(n)}}}\bigl(\mathbf{X}^{(n)};\bigl\lbrace \mG_1^{(n)},\mG_2^{(n)}\bigr\rbrace\bigr)$, $\hat{\mathbf{V}}_n$ to denote $\hat{\mathbf{V}}_{\mG_1^{(n)},\mG_2^{(n)}}$, $\mathbf{V}_n$ to denote $\mathbf{V}_{\mG_1^{(n)},\mG_2^{(n)}}$ and $\nu_n$ to denote $\nu_{\mG_1^{(n)},\mG_2^{(n)}}$. We will also write the difference of cluster means as the row vector $\nu_n^T\mathbf{X}^{(n)}$ for the sake of a clearer notation. If we show that
\begin{equation}\label{condition}
    \hat{\mathbf{\Sigma}}\Bigl(\mathbf{X}^{(n)}\Bigr)\succeq \mathbf{\Sigma}\,\Rightarrow\,\hat{p}_n\geq p_n,
\end{equation}
then the result follows using the same reasoning as in the proof of \cite[Theorem 4]{Gao}, replacing the usual order $\geq$ in $\mathbb{R}$ by the Loewner partial order $\succeq$ between matrices. Consequently, we only need to prove \eqref{condition}. First note that, as the Kronecker product is distributive, we have
\begin{equation}
    \hat{\mathbf{\Sigma}}\Bigl(\mathbf{X}^{(n)}\Bigr)\succeq \mathbf{\Sigma}\,\Rightarrow\, \hat{\mathbf{V}}_n\succeq \mathbf{V}_n.
\end{equation}
Next, by Corollary 7.7.4(a) and Theorem 7.7.2(a) in \cite{horn2013matrix}, we can write
\begin{align}
    \hat{\mathbf{V}}_n\succeq \mathbf{V}_n\,\Leftrightarrow\,\mathbf{V}_n^{-1}\succeq \hat{\mathbf{V}}_n^{-1}
    \Rightarrow \left(\nu_n^T\mathbf{X}^{(n)}\right)\,\mathbf{V}^{-1}_n\,\left(\nu_n^T\mathbf{X}^{(n)}\right)^T \nonumber \\ 
    \geq \left(\nu_n^T\mathbf{X}^{(n)}\right)\,\hat{\mathbf{V}}^{-1}_n\,\left(\nu_n^T\mathbf{X}^{(n)}\right)^T 
    \Leftrightarrow\norm{\nu_n^T\mathbf{X}^{(n)}}_{\mathbf{V}_n} \geq \norm{\nu_n^T\mathbf{X}^{(n)}}_{\hat{\mathbf{V}}_n}.
\end{align}
Let us then state that, if $\mathbb{F}_p(t,c,\mathcal{S})$ denotes the cumulative distribution function of a $c\cdot\chi_p$ distribution truncated to the set $\mathcal{S}$, for $c>0$, it follows that
\begin{equation}\label{statement}
    \mathbb{F}_p(t,c,a\,\mathcal{S})=\mathbb{F}_p\left(\frac{t}{a}, \frac{c}{a}, \mathcal{S}\right),
\end{equation}
for any $a>0$. We prove \eqref{statement} as a technical lemma after the proof. Consequently, taking
\begin{equation}
    a=\frac{\norm{\nu_n^T\mathbf{X}^{(n)}}_{\hat{\mathbf{V}}_n} }{\norm{\nu_n^T\mathbf{X}^{(n)}}_{\mathbf{V}_n}}\leq 1,
\end{equation}
we have
\begin{align}
    1-\hat{p}_n=\mathbb{F}_p\left( \norm{\nu_n^T\mathbf{X}^{(n)}}_{\hat{\mathbf{V}}_n}\,,\,\mathcal{S}_{\hat{\mathbf{V}}_n} \right)=\mathbb{F}_p\left(\norm{\nu_n^T\mathbf{X}^{(n)}}_{\hat{\mathbf{V}}_n}\,,\,a\,\mathcal{S}_{\mathbf{V}_n} \right)\nonumber\\
    =\mathbb{F}_p\left(\frac{1}{a}\,\norm{\nu_n^T\mathbf{X}^{(n)}}_{\hat{\mathbf{V}}_n}\,,\,\frac{1}{a}\,,\,\mathcal{S}_{\mathbf{V}_n} \right)=\mathbb{F}_p\left(\norm{\nu_n^T\mathbf{X}^{(n)}}_{\mathbf{V}_n}\,,\,\frac{1}{a}\,,\,\mathcal{S}_{\mathbf{V}_n} \right)\nonumber\\
    \leq \mathbb{F}_p\left(\norm{\nu_n^T\mathbf{X}^{(n)}}_{\mathbf{V}_n}\,,\,1\,,\,\mathcal{S}_{\mathbf{V}_n} \right)=1-p_n,
\end{align}
where the last inequality follows from Lemma A.3 in \cite{Gao}. This shows \eqref{condition}. 
\end{proof}

\begin{lemma}\label{technical_lemma} For $c>0$ and $\emptyset\neq\mathcal{S}\subset\mathbb{R}$, let $\mathbb{F}_p(t,c,\mathcal{S})$ denote the cumulative distribution function of a $c\cdot\chi_p$ distribution truncated to $\mathcal{S}$. Then, for any $a>0$, it holds
\begin{equation*}
    \mathbb{F}_p(t,c,a\,\mathcal{S})=\mathbb{F}_p\left(\frac{t}{a}, \frac{c}{a}, \mathcal{S}\right).
\end{equation*}
\end{lemma}

\begin{proof}[Proof of Lemma \ref{technical_lemma}]
     First, if we denote by $f(t,c,\mathcal{S)}$ the probability density function of a $c\cdot\chi_p$ distribution truncated to the set $\mathcal{S}$, we have
\begin{equation}
    f(t,c,a\,\mathcal{S})=\frac{1}{a}\,f\left(\frac{t}{a},\frac{c}{a},\mathcal{S}\right).
\end{equation}
Indeed, following the first lines of the proof of \cite[Lemma A.3]{Gao}, we can rewrite $f(t,c,a\,\mathcal{S})$ as
\begin{equation}
    f(t,c, a\,\mathcal{S})=\frac{t^{p-1}\,\mathds{1}\lbrace t\in a\,\mathcal{S}\rbrace}{\int u^{p-1}\exp(-\frac{u^2}{2c^2}),\mathds{1}\lbrace t\in a\,\mathcal{S}\rbrace\,du }\exp\left(-\frac{t^2}{2c^2}\right),
\end{equation}
that we can easily express in terms of $t/a$ as
\begin{equation}
    f(t,c, a\,\mathcal{S})=\frac{\left(\frac{t}{a}\right)^{p-1}\,\mathds{1}\lbrace \frac{t}{a}\in \mathcal{S}\rbrace}{\int \left(\frac{u}{a}\right)^{p-1}\exp(-\frac{(u/a)^2}{2(c/a)^2}),\mathds{1}\lbrace \frac{t}{a}\in \mathcal{S}\rbrace\,du}\exp\left(-\frac{(t/a)^2}{2(c/a)^2}\right)=\frac{1}{a}\,f\left(\frac{t}{a},\frac{c}{a},\mathcal{S}\right),
\end{equation}
where the last equality follows from taking the variable change $y=u/a$ in the integral. Finally, we have
\begin{equation}\label{cdf_change}
    \mathbb{F}_p(t,c,a\,\mathcal{S})=\int_0^t f(x,c,a\,\mathcal{S})\,dx=\frac{1}{a}\int_0^t f\left(\frac{x}{a},\frac{c}{a},\mathcal{S}\right)\,dx = \int_0^{\frac{t}{a}}f\left(u,\frac{c}{a},\mathcal{S}\right)\,du\\
    =\mathbb{F}_p\left(\frac{t}{a},\frac{c}{a},\mathcal{S}\right),\nonumber
\end{equation}
which concludes the proof.
\end{proof}

\begin{proof}[Alternative formulation of Assumption~\ref{as_3}] Assumption~\ref{as_3} can be formulated in terms of strong mixing of measure-preserving dynamical systems~\cite[Chapter 20]{Klenke}. To show this, let us consider the sets $A_k=\lbrace i\in\mathbb{N}\,:\,\mu_i^{(i)}=\theta_k\rbrace$ for any $k=1,\ldots,K^*$. This makes the family $\mathcal{F}=\mathcal{P}(\mathcal{A})$ with $\mathcal{A}=\lbrace A_k\rbrace_{k=1}^{K^*}$ a $\sigma$-algebra on $\mathbb{N}$. Next, let $P_n$ denote the measure defined by $P_n(A)=\frac{1}{n}\lvert A\cap[n]\rvert$ for any $A\in\mathcal{F}$ and $P$ denote the measure defined by $P(\cup_{s\in S}A_s)=\sum_{s \in S}\pi_s$ for any $S\in\mathcal{P}(\lbrace 1,\ldots,K^*\rbrace)$. Note that the pair $(\mathbb{N}, \mathcal{F})$ can be provided with either $P_n$ or $P$ to form a measure space. Besides, Assumption~\ref{as_1} states the setwise convergence of $P_n$ to $P$ when $n\rightarrow\infty$. Finally, for any $k=1,\ldots,K^*$, we can define the transformation $T(A_k)=\lbrace i\in\mathbb{N}\,:\,\mu_{i-1}^{(i-1)}=\theta_k\rbrace$, which is measure-preserving on $(\mathbb{N},\mathcal{F},P)$~\cite[Definition 20.6]{Klenke}. Then, Equation~\eqref{mixing} can be rewritten as:
\begin{equation}\label{reformulate_mixing}
    P_n\left(T^{-r}(A_k)\cap A_{k'}\right)\underset{n\rightarrow\infty}{\longrightarrow} P\left(T^{-r}(A_k)\cap A_{k'}\right)\underset{r\rightarrow\infty}{\longrightarrow} P(A_k)\,P(A_{k'}),
\end{equation}
for any $k,k'\in\lbrace 1,\ldots,K^*\rbrace$. The first limit in \eqref{reformulate_mixing} follows from Assumption~\ref{as_1}, whereas the second one is equivalent to state that the measure-preserving dynamical system $(\mathbb{N},\mathcal{F},P,T)$ is (strong) mixing (see ~\cite[Definition 20.04]{Klenke}). 
\end{proof}

\begin{proof}[Proof of Remark~\ref{remark_compound}] Let $\mathbf{U}^{(n)}=b\mathbf{1}_{n\times n}+(a-b)\,\mathbf{I}_n$. As $b\in(-\frac{a}{n-1},a)$ if and only if $\mathbf{U}^{(n)}$ is positive definite, condition $0<b<a$ is needed to ensure positive definiteness for all $n\in\mathbb{N}$. Following the Sherman–Morrison formula \cite{bartlett}, we can derive an explicit expression for the sequence of inverse matrices:
\begin{equation}
   \left(\mathbf{U}^{(n)}\right)^{-1}=\frac{1}{a-b}\,\mathbf{I}_n+\frac{-b}{(a-b)(nb+a-b)},\quad\forall\,n\in\mathbb{N}.
\end{equation}
Consequently, for every $r\geq 0$ and every $i\in\mathbb{N}$, we have
\begin{equation*}
    \left(\mathbf{U}^{(n)}\right)^{-1}_{i\,i+r}=
    \begin{cases}
        \frac{1}{a-b}+\frac{-b}{(a-b)(nb+a-b)} & \mathrm{if }\,\,r=0,\\
        \frac{-b}{(a-b)(nb+a-b)} & \mathrm{if }\,\,r>0,
    \end{cases}
\end{equation*}
which are monotone, so condition $(ii)$ in Assumption~\ref{as_2} is satisfied. Then, we have
\begin{equation*}
     \Lambda_{i\,i+r}=
    \begin{cases}
        \frac{1}{a-b} & \mathrm{if }\,\,r=0,\\
        0 & \mathrm{if }\,\,r>0,
    \end{cases}
\end{equation*}
for all $i\in\mathbb{N}$, $\lambda_0 = 1/(a-b)$ and $\lambda_r=0$ for $r>0$. Consequently, Assumption~\ref{as_2} holds.
\end{proof}

\begin{proof}[Proof of Remark~\ref{remark_diagonal}] The case of diagonal matrices is straightforward as both $\mathbf{U}^{(n)}$ and $\left(\mathbf{U}^{(n)}\right)^{-1}$ are defined by a sequence $\lbrace\lambda_i\rbrace_{i\in\mathbb{N}}$. Every diagonal entry of the inverse satisfies $\left(U^{(n)}\right)^{-1}_{ii}=\frac{1}{\lambda_i}$ for all $n\in\mathbb{N}$ and, as we asked the $\lambda_i$ to converge to $\lambda$, which is strictly positive due to the positive definiteness of $\mathbf{U}^{(n)}$, Assumption~\ref{as_2} is satisfied.    
\end{proof}

\begin{proof}[Proof of Remark~\ref{remark_ar}] The inverse of an auto-regressive covariance matrix of lag $P\geq 1$ is banded with $2P-1$ non-zero diagonals. Its explicit form is derived in \cite{inverse_AR} for a stationary process of any lag, and the cases $P\leq 3$ are discussed in detail in \cite{inverse_AR_123}. From these results we can derive the behavior of the sequences $\lbrace \left(U^{(n)}\right)^{-1}_{i\,i+r}\rbrace$ as $n$ increases. The diagonal elements define the sequences
\begin{equation*}
    \sigma^2\,\left\lbrace \left(U^{(n)}\right)^{-1}_{ii}\right\rbrace_{n\in\mathbb{N}}=
    \begin{cases} 
      \lbrace 1 + \sum_{k=1}^{i-1} \beta_k^2,1 + \sum_{k=1}^{i-1} \beta_k^2,\ldots\rbrace& \mathrm{if }\,\,i \leq P+1,\\
      \lbrace 0,\overset{i-1}{\ldots}, 0, 1, 1+\beta_1^2,1, 1+\beta_1^2\beta_2^2,\ldots, 1+ \sum_{k=1}^{P} \beta_k^2, 1+ \sum_{k=1}^{P} \beta_k^2, \ldots \rbrace & \mathrm{if }\,\,i>P+1,
   \end{cases}      
\end{equation*}
where the sums are taken as zero if the upper limit of summation is zero. Note that these sequences do not satisfy condition $(i)$ in Assumption~\ref{as_2} as, even if each sequence reaches its limit after a finite number of terms, the index of the term where the limit is reached diverges with $i$. In other words, we can dominate the sequence, but not by a summable one. However, for all $i\in\mathbb{N}$ the series are non-decreasing so condition $(ii)$ is satisfied and we have 
\begin{equation*}
    \sigma^2\,\Lambda_{ii}=
    \begin{cases}
         1 + \sum_{k=1}^{i-1} \beta_k^2 & \mathrm{if }\,\,i \leq P+1\\
         1 + \sum_{k=1}^{P} \beta_k^2 & \mathrm{if }\,\,i > P+1.
    \end{cases}
\end{equation*}
Then, $\sigma^2\,\lambda_0= 1 + \sum_{k=1}^{P} \beta_k^2$. The sequences outside the main diagonal show a similar behavior, but they are not positive in general. As, following the same reasoning, they do not satisfy condition $(i)$ in Assumption~\ref{as_2}, we force them to satisfy condition $(ii)$. For any $0<r\leq P$, we have
\begin{equation}\label{inverse_arp}
    \sigma^2\,\left\lbrace \left(U^{(n)}\right)^{-1}_{i\,i+r}\right\rbrace_{n\in\mathbb{N}}=
    \begin{cases} 
      \lbrace -\beta_r + \sum_{k=1}^{i-(r+1)} \beta_k\beta_{k+r},\,-\beta_r + \sum_{k=1}^{i-(r+1)} \beta_k\beta_{k+r},\ldots\rbrace& \mathrm{if }\,\,i \leq P+1,\\
      \lbrace 0,\overset{i-1}{\ldots}, 0, -\beta_r + \beta_1\beta_{1+r} , -\beta_r + \beta_1\beta_{1+r} + \beta_2\beta_{2+r},\ldots,\\
      -\beta_r + \sum_{k=1}^{P-r}\beta_k\beta_{k+r}, -\beta_r + \sum_{k=1}^{P-r}\beta_k\beta_{k+r}, \ldots \rbrace & \mathrm{if }\,\,i>P+1.
   \end{cases}      
\end{equation}
For these sequences to satisfy condition $(ii)$ we need them to be non-decreasing or non-increasing. For $P\leq 2$ this is always satisfied but, for $P>2$, we need to require all the $\beta_k$ to have the same sign. In that case, condition $(ii)$ holds and we have
 \begin{equation*}
     \sigma^2\Lambda_{i\,i+r}=
     \begin{cases}
        -\beta_r + \sum_{k=1}^{i-(r+1)}\beta_k\beta_{k+r} & \mathrm{if }\,\, i\leq P+1,\\
        -\beta_r + \sum_{k=1}^{P-r}\beta_k\beta_{k+r} & \mathrm{if }\,\, i> P+1,
     \end{cases}
 \end{equation*}
and, consequently, $\sigma^2\lambda_r=-\beta_r + \sum_{k=1}^{P-r}\beta_k\beta_{k+r}$. As the sequence $\lbrace\lambda_r\rbrace_{r=1}^{\infty}$ is non-zero for for a finite number of terms (due to the bandedness of the inverse matrix), its sum converges and Assumption~\ref{as_2} is satisfied.
\end{proof}

\begin{proof}[Proof of Lemma~\ref{lemma_prop}] We start by rewriting the sum in \eqref{lemma_eq} as a sum along each diagonal. Using the symmetry of $\left(\mathbf{U}^ {(n)}\right)^{-1}$ we have,
\begin{align}
    \underset{n\rightarrow\infty}{\lim}\,\frac{1}{n}\sum_{l,s=1}^n \left(U^{(n)}\right)^{-1}_{ls}\,\mathds{1}\lbrace \mu_l^{(n)}=\theta_k\rbrace\,\mathds{1}\lbrace \mu_s^{(n)}=\theta_{k'}\rbrace\nonumber\\
    =\underset{n\rightarrow\infty}{\lim}\,\frac{1}{n}\sum_{r=1}^{n-1} \sum_{i=1}^{n-r}\left(U^{(n)}\right)^{-1}_{i\,i+r}\,\mathds{1}\lbrace \mu_i^{(n)}=\theta_k\rbrace\,\mathds{1}\lbrace \mu_{i+r}^{(n)}=\theta_{k'}\rbrace\label{super_diag}\\
    +\underset{n\rightarrow\infty}{\lim}\,\frac{1}{n}\sum_{r=1}^{n-1} \sum_{i=1}^{n-r}\left(U^{(n)}\right)^{-1}_{i\,i+r}\,\mathds{1}\lbrace \mu_{i+r}^{(n)}=\theta_k\rbrace\,\mathds{1}\lbrace \mu_{i}^{(n)}=\theta_{k'}\rbrace\label{sub_diag}\\
    +\underset{n\rightarrow\infty}{\lim}\,\frac{1}{n} \sum_{i=1}^{n}\left(U^{(n)}\right)^{-1}_{i\,i}\,\mathds{1}\lbrace \mu_i^{(n)}=\theta_k\rbrace\,\mathds{1}\lbrace \mu_{i}^{(n)}=\theta_{k'}\rbrace,\label{main_diag}
\end{align}
where \eqref{super_diag},\eqref{sub_diag} and \eqref{main_diag} are respectively the sums along all the superdiagonals, subdiagonals and along the main diagonal. Let us detail the general reasoning that we use to show that the three quantities converge. Let $\lbrace a_i^{(n)}\rbrace_{i\in\mathbb{N}}$ be a double sequence such that $\lim_{n\rightarrow\infty}a_i^{(n)}=a_i\in\mathbb{R}$, and let $\lbrace b_i^{(n)}\rbrace_{i\in\mathbb{N}}$ be a binary Cesàro summable double sequence, i.e. such that $\lim_{n\rightarrow\infty} \frac{1}{n}\sum_{i=1}^n b_i^{(n)} = b$ and $b_i^{(n)}\in\lbrace 0,1\rbrace$ for all $i,n\in\mathbb{N}$. Let us first show that, if $\lbrace a_i^{(n)}\rbrace_{n\in\mathbb{N}}$ satisfies any of the conditions $(i)$ or $(ii)$, and the sequence $\lbrace a_i^{(1)}-a_i\rbrace_{i=1}^\infty\in\ell_1$, we can write

\begin{equation}\label{conv_result}
    \underset{n\rightarrow\infty}{\lim}\,\frac{1}{n}\sum_{i=1}^{n}a_i^{(n)}\,b_i^{(n)} = \underset{n\rightarrow\infty}{\lim}\frac{1}{n}\sum_{i=1}^{n}a_i\,b_i^{(n)}.
\end{equation}
First, note that
\begin{equation}\label{eq_eq}
    \underset{n\rightarrow\infty}{\lim}\,\frac{1}{n}\sum_{i=1}^{n}a_i^{(n)}\,b_i^{(n)}=\underset{n\rightarrow\infty}{\lim}\,\frac{1}{n}\sum_{i=1}^{n}\left(a_i^{(n)}-a_i\right)\,b_i^{(n)} + \underset{n\rightarrow\infty}{\lim}\,\frac{1}{n}\sum_{i=1}^{n}a_i\,b_i^{(n)}.
\end{equation}
Therefore, it suffices to show that the first term in \eqref{eq_eq} is zero to have \eqref{conv_result}. Using Hölder's inequality, we have
\begin{align*}
    \underset{n\rightarrow\infty}{\lim}\frac{1}{n}\left\lvert \sum_{i=1}^n \left(a_i^{(n)}-a_i\right)\,b_i^{(n)}\right\rvert\leq \underset{n\rightarrow\infty}{\lim}\frac{1}{n}\sum_{i=1}^n \left\lvert \left(a_i^{(n)}-a_i\right)b_i^{(n)} \right\rvert \\
    \leq \underset{n\rightarrow\infty}{\lim}\left(\sum_{i=1}^n \left(a_i^{(n)}-a_i\right)^2\right)^{\frac{1}{2}}\,\underset{n\rightarrow\infty}{\lim}\,\frac{1}{n}\left(\sum_{i=1}^n b_i^{(n)}\right)^{\frac{1}{2}}.
\end{align*}
On one hand,
\begin{equation*}
    \underset{n\rightarrow\infty}{\lim}\,\frac{1}{n}\left(\sum_{i=1}^n b_i^{(n)}\right)^{\frac{1}{2}} = \underset{n\rightarrow\infty}{\lim} \frac{1}{\sqrt{n}}\, \underset{n\rightarrow\infty}{\lim} \left(\frac{1}{n} \sum_{i=1}^n b_i^{(n)}\right)^{\frac{1}{2}}=0.
\end{equation*}
On the other hand, let us show that
\begin{equation}\label{square_limit}
    \underset{n\rightarrow\infty}{\lim}\sum_{i=1}^n \left(a_i^{(n)}-a_i\right)^2=0
\end{equation}
if $\lbrace a_i^{(n)}\rbrace_{n\in\mathbb{N}}$ satisfies any of the conditions $(i)$ or $(ii)$. If $\lbrace a_i^{(n)}\rbrace_{n\in\mathbb{N}}$ satisfies $(i)$, the sequence $\lbrace (a_i^{(n)}-a_i)^2\rbrace_{n\in\mathbb{N}}$ is dominated by the sequence $\lbrace\alpha_i^2\rbrace_{i\in\mathbb{N}}$, which is summable as $\ell_1\subset\ell_2$. Then, \eqref{conv_result} holds following the Dominated Convergence Theorem~\cite[Theorem 9.20]{Yeh}. If $\lbrace a_i^{(n)}\rbrace_{n\in\mathbb{N}}$ is non-increasing, then $a_i^{(n+1)}-a_i\leq a_i^{(n)}-a_i$ implies $(a_i^{(n+1)}-a_i)^2\leq (a_i^{(n)}-a_i)^2$ and 
$\Tilde{a}_i^{(n)}:=(a_i^{(n)}-a_i)^2$ is a non-increasing and non-negative sequence. Similarly, if $\lbrace a_i^{(n)}\rbrace_{n\in\mathbb{N}}$ is non-decreasing, then $a_i^{(n+1)}-a_i\geq a_i^{(n)}-a_i$ implies $(a_i^{(n+1)}-a_i)^2\leq (a_i^{(n)}-a_i)^2$ and 
$\Tilde{a}_i^{(n)}$ is again a non-increasing and non-negative sequence. Then, the sequence $z_i^{(n)}:= \Tilde{a}_i^{(1)} - \Tilde{a}_i^{(n)}$ is non-negative and non-decreasing. Thus, following the Monotone Convergence Theorem~\cite[Theorem 8.5]{Yeh}, we have
\begin{equation}\label{monotone_conv}
    \underset{n\rightarrow\infty}{\lim}\sum_{i=1}^n z_i^{(n)} = \underset{n\rightarrow\infty}{\lim}\sum_{i=1}^n (a_i^{(1)} - a_i)^2,
\end{equation}
which implies \eqref{square_limit} if the limit in the right side of \eqref{monotone_conv} exists and is finite. This is guaranteed if we ask the sequence $\lbrace a_i^{(1)}-a_i\rbrace_{i=1}^{\infty}$ to be summable. This always holds in our case as we can arbitrarily define the entries $\left({U}^{(n)}\right)^{-1}_{i\,i+r}$ for $i>n$. Consequently, if we write $\lbrace\left({U}^{(1)}\right)^{-1}_{i\,i+r}\rbrace_{i=1}^\infty=\lbrace\left({U}^{(1)}\right)^{-1}_{1\,1+r}, \Lambda_{2\,2+r},\Lambda_{3\,3+r},\ldots\rbrace$, the sequence $\lbrace\left({U}^{(1)}\right)^{-1}_{i\,i+r} - \Lambda_{i\,i+r}\rbrace_{i=1}^\infty$ is trivially summable. This proves \eqref{conv_result}. 

Now, if we have that $\underset{i\rightarrow\infty}{\lim}a_i=a$, let us show that
\begin{equation}\label{conv_result_2}
   \underset{n\rightarrow\infty}{\lim}\frac{1}{n}\sum_{i=1}^{n}a_i\,b_i^{(n)} = ab. 
\end{equation}
First, let separate the sum in \eqref{conv_result_2} as
\begin{equation}
    \frac{1}{n}\sum_{i=1}^{n}a_i\,b_i^{(n)} =   \frac{1}{n}\sum_{i=1}^{n}\left(a_i-a\right)\,b_i^{(n)} +  \frac{a}{n}\sum_{i=1}^{n}\,b_i^{(n)}.
\end{equation}
The right term tends to $ab$ when $n\rightarrow\infty$. Let's show that the first term tends to zero. For any $i_0\in\mathbb{N}$, we can write
\begin{align}
    \left\lvert\frac{1}{n}\sum_{i=1}^{n}\left(a_i-a\right)\,b_i^{(n)}\right\rvert \leq \left\lvert\frac{1}{n}\sum_{i=1}^{i_0-1}\left(a_i-a\right)\,b_i^{(n)}\right\rvert + \left\lvert\frac{1}{n}\sum_{i=i_0}^{n}\left(a_i-a\right)\,b_i^{(n)}\right\rvert\\
    \leq \underset{i < i_0}{\sup}\abs{a_i-a}\frac{1}{n}\sum_{i=1}^{i_0-1} b_i^{(n)} + \underset{i \geq i_0}{\sup}\abs{a_i-a}\frac{1}{n}\sum_{i=i_0}^{n} b_i^{(n)}
    \leq \frac{C}{n} + \underset{i \geq i_0}{\sup}\abs{a_i-a}\frac{1}{n}\sum_{i=i_0}^{n} b_i^{(n)},
\end{align}
where $C$ is a real constant. Then, following the definition of limit, when can choose $i_0$ as the one such that for all $i\geq i_0$ we have ${\abs{a_i-a}}\leq \frac{1}{n}$. Therefore,
\begin{equation}
   \left\lvert \frac{1}{n}\sum_{i=1}^{n}\left(a_i-a\right)\,b_i^{(n)}\right\rvert\leq\frac{C}{n}+\frac{1}{n^2}\sum_{i=i_0}^n b_i^{(n)}, 
\end{equation}
which tends to zero when $n\rightarrow\infty$ using that $\lbrace b_i^{(n)}\rbrace_i\in\mathbb{N}$ has Cesàro sum $b$. Thus, we have \eqref{conv_result_2}. As the sequences $\left(U^{(n)}\right)^{-1}_{i\,i+r}$ have limits $\Lambda_{i\,i+r}$ when $i\rightarrow\infty$, following Assumption~\ref{as_3}, and the products of indicator functions are Cesàro summable thanks to Assumptions~\ref{as_1} and \ref{as_3}, we can use \eqref{conv_result} and \eqref{conv_result_2} to rewrite the three limits in \eqref{super_diag}, \eqref{sub_diag}, \eqref{main_diag} as
\begin{align}\label{limit_conv}
    \underset{n\rightarrow\infty}{\lim}\,\frac{1}{n}\sum_{l,s=1}^n \left(U^{(n)}\right)^{-1}_{ls}\,\mathds{1}\lbrace \mu_l^{(n)}=\theta_k\rbrace\,\mathds{1}\lbrace \mu_s^{(n)}=\theta_{k'}\rbrace\nonumber\\
    =\underset{n\rightarrow\infty}{\lim}\sum_{r=1}^{n-1} \lambda_r\left(\pi_{kk'}^r+\pi_{k'k}^r\right)+\lambda_0\pi_{k}\delta_{kk'}=2(\lambda-\lambda_0)\pi_{k}\pi_{k'}+\lambda_0\pi_k\delta_{kk'},
\end{align}
where the last limit is derived following the same reasoning as to prove \eqref{conv_result_2}. This concludes the proof.
\end{proof}

\begin{proof}[Proof of Proposition~\ref{prop_overestimate_general}] We start by proving the element-wise convergence in probability of \eqref{estSigma}. More precisely, we show that
\begin{equation}\label{conv_proba}
    \hat{\Sigma}_{ij}^{(n)}\overset{p}{\rightarrow}\Sigma_{ij}+ \lambda_0\sum_{k=1}^{K^*}\pi_k\left(\theta_{ki}-\Tilde{\theta}_{i}\right)\left(\theta_{kj}-\Tilde{\theta}_{j}\right),
\end{equation}
for all $i,j\in\lbrace 1,\ldots,p\rbrace$, where $\hat{\Sigma}_{ij}^{(n)}$ is the $ij$ entry of $\hat{\mathbf{\Sigma}}\left(\mathbf{X}^{(n)}\right)$, that is,
\begin{equation}\label{element_wise_estimator}
    \hat{\Sigma}_{ij}=\frac{1}{n-1}\sum_{l,s=1}^n\left(X_{li}-\Bar{X}_i\right)\left(U^{-1}\right)_{ls}\left(X_{sj}-\bar{X}_j\right),\quad\forall\,i,j\in\lbrace 1,\ldots,p\rbrace,
\end{equation}
where $\Bar{X}_i=\frac{1}{n}\sum_{k=1}^n X_{ki}$, and we have defined $\Tilde{\theta_i}=\sum_{k=1}^{K^*}\pi_k\theta_{ki}$. Recall that all the quantities in \eqref{conv_proba} have been defined in Assumptions~\ref{as_1} and \ref{as_2}. To prove \eqref{conv_proba}, it suffices to show, following the same reasoning as in the proof of~\cite[Lemma C.1]{Gao}, that
\begin{equation}\label{exp_and_var}
    \underset{n\rightarrow\infty}{\lim}\mathbb{E}\left(\hat{\Sigma}_{ij}^{(n)}\right) = \Sigma_{ij} + \lambda_0\sum_{k=1}^{K^*}\pi_k\left(\theta_{ki}-\Tilde{\theta}_{i}\right)\left(\theta_{kj}-\Tilde{\theta}_{j}\right)\quad\textrm{ and }\quad\underset{n\rightarrow\infty}{\mathrm{Var}}\left(\hat{\Sigma}_{ij}^{(n)}\right)=0.
\end{equation}
Indeed, \eqref{exp_and_var} implies convergence in mean of $\hat{\Sigma}_{ij}^{(n)}$ towards the limit of its expectation and, following Markov's inequality, convergence in probability. Let start by rewriting $\hat{\Sigma}_{ij}^{(n)}$. Following \eqref{element_wise_estimator}, we can write
\begin{eqnarray}\label{sigma_ij_terms}
   \hat{\Sigma}_{ij}^{(n)} &=&  \frac{1}{n-1}\sum_{l,s=1}^n X_{li}^{(n)}\,X_{js}^{(n)}\,\left(U^{(n)}\right)^{-1}_{ls}-
   \frac{1}{n-1}\,\bar{X}_{j}^{(n)}\,\sum_{l,s=1}^n X_{li}^{(n)}\,\left(U^{(n)}\right)^{-1}_{ls}\nonumber\\ & & 
-\frac{1}{n-1}\,\bar{X}_{i}^{(n)}\,\sum_{l,s=1}^n X_{sj}^{(n)}\left(U^{(n)}\right)^{-1}_{ls} + \frac{1}{n-1}\,\bar{X}_{i}^{(n)}\,\bar{X}_{j}^{(n)}\,\sum_{l,s=1}^n \left(U^{(n)}\right)^{-1}_{ls}.
\end{eqnarray}
For simplicity, we denote as $A_{ij}^{(n)}$, $B_{ij}^{(n)}$, $C_{ij}^{(n)}$ and $D_{ij}^{(n)}$ the four terms in \eqref{sigma_ij_terms} respectively. First, let us derive their asymptotic expectations. 
\begin{eqnarray*}
\mathbb{E}\left(A_{ij}^{(n)}\right) 
&=& 
\frac{1}{n-1}\sum_{l,s=1}^n \left(U^{(n)}\right)^{-1}_{ls}\mathbb{E}\left(X_{li}^{(n)}\,X_{sj}^{(n)}\right)\\
&=&
\frac{1}{n-1}\sum_{l,s=1}^n \left(U^{(n)}\right)^{-1}_{ls}\mu_{li}^{(n)}\mu_{sj}^{(n)}+\frac{\Sigma_{ij}}{n-1}\sum_{l,s=1}^n\left(U^{(n)}\right)^{-1}_{ls}\,U^{(n)}_{sl}\\
&=&
\sum_{k,k'=1}^{K^*}\frac{1}{n-1}\sum_{l,s=1}^n \left(U^{(n)}\right)^{-1}_{ls} \mathds{1}\lbrace \mu_l^{(n)}=\theta_k\rbrace\,\mathds{1}\lbrace \mu_s^{(n)}=\theta_{k'}\rbrace\,\theta_{ki}\theta_{k'j}+\frac{n}{n-1}\Sigma_{ij}.
\end{eqnarray*}

Using Lemma~\ref{lemma_prop}, we have
\begin{equation}\label{exp_a}
    \underset{n\rightarrow\infty}{\lim}\,\mathbb{E}\left(A_{ij}^{(n)}\right)=2(\lambda-\lambda_0)\sum_{k=1}^{K^*}\pi_k\theta_{ki}\sum_{k=1}^{K^*}\pi_k\theta_{kj}+\lambda_0\sum_{k=1}^{K^*}\pi_k\theta_{ki}\theta_{kj}+\Sigma_{ij}.
\end{equation}
Then,
\begin{eqnarray*}
    \mathbb{E}\left(B_{ij}^{(n)}\right)&=& \frac{1}{n(n-1)}\sum_{l,s,r=1}^n \left(U^{(n)}\right)^{-1}_{ls}\mathbb{E}\left(X_{li}^{(n)}\,X_{rj}^{(n)}\right)
    \\
    &=&\frac{1}{n(n-1)}\sum_{l,s,r=1}^n \left(U^{(n)}\right)^{-1}_{ls}\mu_{li}^{(n)}\mu_{rj}^{(n)}+\frac{\Sigma_{ij}}{n-1}\\
    &=&\frac{1}{n}\sum_{r=1}^n\mu_{rj}^{(n)}\frac{1}{n-1}\sum_{l,s}^n\left(U^{(n)}\right)^{-1}_{ls}\mu_{li}^{(n)}+\frac{\Sigma_{ij}}{n-1}\\
    &=& \sum_{k=1}^{K^*}\frac{1}{n}\sum_{r=1}^n\mathds{1}\lbrace \mu_r^{(n)}=\theta_k\rbrace\theta_{kj}\,\sum_{k=1}^{K^*}\frac{1}{n-1}\sum_{l,s=1}^n\left(U^{(n)}\right)^{-1}_{ls}\mathds{1}\lbrace \mu_{l}^{(n)}=\theta_{k}\rbrace\theta_{ki}+\frac{\Sigma_{ij}}{n-1}.
\end{eqnarray*}
Using the same reasoning as to prove Lemma~\ref{lemma_prop}, we have
\begin{equation*}
    \underset{n\rightarrow\infty}{\lim}\frac{1}{n-1}\sum_{l,s=1}^n\left(U^{(n)}\right)^{-1}_{ls}\mathds{1}\lbrace \mu_l^{(n)}=\theta_k\rbrace = (2(\lambda-\lambda_0)+\lambda_0)\pi_k.
\end{equation*}
This, together with Assumption~\ref{as_1}, yields
\begin{equation}\label{exp_bc}
   \underset{n\rightarrow\infty}{\lim} \mathbb{E}\left(B_{ij}^{(n)}\right)=\underset{n\rightarrow\infty}{\lim} \mathbb{E}\left(C_{ij}^{(n)}\right)=
    (2(\lambda-\lambda_0)+\lambda_0)\sum_{k=1}^{K^*}\pi_k\theta_{kj}\sum_{k=1}^{K^*}\pi_k\theta_{ki},
\end{equation}
where $B_{ij}^{(n)}$ and $C_{ij}^{(n)}$ have the same expectation by symmetry. Finally,
\begin{eqnarray*}
    \mathbb{E}\left(D_{ij}^{(n)}\right)&=&
    \frac{1}{n^2(n-1)}\sum_{l,s=1}^n\left(U^{(n)}\right)^{-1}_{ls}\sum_{r,r'=1}^n \mathbb{E}\left(X_{ri}^{(n)}\,X_{r'j}^{(n)}\right)\\ &=&
    \frac{1}{n-1}\sum_{l,s=1}^n\left(U^{(n)}\right)^{-1}_{ls}\left[\frac{1}{n^2}\sum_{r,r'=1}^{n}\mu_{ri}^{(n)}\mu_{r'j}^{(n)}+\frac{\Sigma{ij}}{n^2}\sum_{r,r'=1}^{n}U_{rr'}^{(n)}\right].
\end{eqnarray*}
Using the same reasoning as to prove Lemma~\ref{lemma_prop}, we have
\begin{equation}\label{claim_0}
    \underset{n\rightarrow\infty}{\lim}\,\frac{1}{n-1}\sum_{l,s=1}^n\left(U^{(n)}\right)^{-1}_{ls}=2(\lambda-\lambda_0)+\lambda_0.
\end{equation}
Moreover, we state that
\begin{equation}\label{claim}
    \underset{n\rightarrow\infty}{\lim}\,\frac{1}{n^2}\sum_{l,s=1}^n U^{(n)}_{ls}=0.
\end{equation}
We prove \eqref{claim} at the end of the proof. This claim, together with \eqref{claim_0} and Assumption~\ref{as_1}, yields
\begin{equation}\label{exp_d}
    \underset{n\rightarrow\infty}{\lim}\mathbb{E}\left(D_{ij}^{(n)}\right)=\left(2(\lambda-\lambda_0)+\lambda_0\right)\sum_{k=1}^{K^*}\pi_k\theta_{ki}\,\sum_{k=1}^{K^*}\pi_k\theta_{kj}.
\end{equation}
Consequently, following \eqref{exp_a}, \eqref{exp_bc} and \eqref{exp_d}, we have
\begin{eqnarray}
    \underset{n\rightarrow\infty}{\lim}\mathbb{E}\left(\hat{\Sigma}_{ij}^{(n)}\right)&=&\Sigma_{ij}+\lambda_0\left[ \sum_{k=1}^{K^*}\pi_k\theta_{ki}\theta_{kj}-\sum_{k=1}^{K^*}\pi_k\theta_{ki}\sum_{k=1}^{K^*}\pi_k\theta_{kj}\right]\nonumber\\
    &=& \Sigma_{ij}+\lambda_0\sum_{k=1}^{K^*}\pi_k\left(\theta_{ki}-\Tilde{\theta}_{i}\right)\left(\theta_{kj}-\Tilde{\theta}_{j}\right).
\end{eqnarray}
This is the first statement in \eqref{exp_and_var}. To prove the second one, we show that the variance of each term in \eqref{sigma_ij_terms} tends to zero. To do so, we need the explicit form of the non-centered $4$-th moments of a Gaussian distribution. More precisely, if $X_1,\ldots,X_4$ are four Gaussian random variables with $\mathbb{E}(X_i)=\mu_i$ and $\mathrm{Cov}(X_i,X_j)=\sigma_{ij}$, for $i,j\in\lbrace1,\ldots,4\rbrace$, we need the explicit form of the quantity
\begin{equation}
    \mathbb{E}\left(X_1\,X_2\,X_3\,X_4\right)-\mathbb{E}\left(X_1\,X_2\right)\,\mathbb{E}\left(X_3\,X_4\right).
\end{equation}
The first term can be derived using the moment generating function of a $4$-dimensional normal distribution
\begin{equation*}
    M_{(X_1,\ldots,X_4)}(t_1,\ldots,t_4)=\exp\left(\sum_{i=1}^4 \mu_i\,t_i+\frac{1}{2}\sum_{i,j=1}^n \sigma_{ij}\,t_i\,t_j\right),
\end{equation*}
and computing
\begin{equation*}
    \mathbb{E}\left(X_1\,X_2\,X_3\,X_4\right)=\frac{\partial M_{(X_1,\ldots,X_4)}(t_1,\ldots,t_4)}{\partial\,t_1\cdots\partial\,t_4}\Bigg\rvert_0.
\end{equation*}
Doing so, and using $\mathbb{E}(X_i\,X_j)=\mu_i\mu_j+\sigma_{ij}$, we can derive
\begin{equation}\label{exp_product}
    \mathbb{E}\left(X_1\,X_2\,X_3\,X_4\right)-\mathbb{E}\left(X_1\,X_2\right)\,\mathbb{E}\left(X_3\,X_4\right)=\sigma_{13}\sigma_{24}+\sigma_{14}\sigma_{23}+\mu_1\mu_4\sigma_{23}+\mu_1\mu_3\sigma_{24}+\mu_{2}\mu_{3}\sigma_{14}+\mu_2\mu_4\sigma_{13}.
\end{equation}
We are ready to prove that $\mathrm{Var}\left(\hat{\Sigma}_{ij}^{(n)}\right)$ tends to zero. First, using $\mathrm{Var}(X)=\mathbb{E}(X^2)-\mathbb{E}(X)^2$, we have
\begin{align}\label{var_a_expression}
    \mathrm{Var}\left(A_{ij}^{(n)}\right)=\frac{1}{(n-1)^2}\sum_{l,s,k,r=1}^n\left(U^{(n)}\right)^{-1}_{sl}\left(U^{(n)}\right)^{-1}_{kr}\left[\mathbb{E}\left(X_{li}\,X_{sj}\,X_{ri}\,X_{kj}\right)-\mathbb{E}\left(X_{li}\,X_{sj}\right)\mathbb{E}\left(X_{ki}\,X_{rj}\right)\right].
\end{align}
Using \eqref{exp_product}, we can separate \eqref{var_a_expression} into the following six terms:
\begin{eqnarray}
 \mathrm{Var}\left(A_{ij}^{(n)}\right)
 &=&\frac{\Sigma_{ii}\Sigma_{jj}}{(n-1)^2}\sum_{l,s,k,r=1}^n\left(U^{(n)}\right)^{-1}_{ls}\left(U^{(n)}\right)^{-1}_{kr}\,U^{(n)}_{lk}\,U^{(n)}_{sr}\label{i}\\
 &+&\frac{\Sigma_{ij}^2}{(n-1)^2}\sum_{l,s,k,r=1}^n\left(U^{(n)}\right)^{-1}_{ls}\left(U^{(n)}\right)^{-1}_{kr}\,U^{(n)}_{lr}\,U^{(n)}_{sk}\label{ii}\\
 &+&\frac{\Sigma_{jj}}{(n-1)^2}\sum_{l,s,k,r=1}^n\left(U^{(n)}\right)^{-1}_{ls}\left(U^{(n)}\right)^{-1}_{kr}\,U^{(n)}_{sr}\,\mu_{li}^{(n)}\,\mu_{ki}^{(n)}\label{iii}\\
 &+&\frac{\Sigma_{ij}}{(n-1)^2}\sum_{l,s,k,r=1}^n\left(U^{(n)}\right)^{-1}_{ls}\left(U^{(n)}\right)^{-1}_{kr}\,U^{(n)}_{sk}\,\mu_{li}^{(n)}\,\mu_{rj}^{(n)}\label{iv}\\
 &+&\frac{\Sigma_{ij}}{(n-1)^2}\sum_{l,s,k,r=1}^n\left(U^{(n)}\right)^{-1}_{ls}\left(U^{(n)}\right)^{-1}_{kr}\,U^{(n)}_{lr}\,\mu_{ki}^{(n)}\,\mu_{sj}^{(n)}\label{v}\\
 &+&\frac{\Sigma_{ii}}{(n-1)^2}\sum_{l,s,k,r=1}^n\left(U^{(n)}\right)^{-1}_{ls}\left(U^{(n)}\right)^{-1}_{kr}\,U^{(n)}_{lk}\,\mu_{sj}^{(n)}\,\mu_{rj}^{(n)}. \label{vi}
\end{eqnarray}
Each of these terms tend to zero when $n\rightarrow\infty$. For \eqref{i}, we have
\begin{align*}
    \frac{\Sigma_{ii}\Sigma_{jj}}{(n-1)^2}\sum_{l,s,k,r=1}^n\left(U^{(n)}\right)^{-1}_{ls}\left(U^{(n)}\right)^{-1}_{kr}\,U^{(n)}_{lk}\,U^{(n)}_{sr}=
    \frac{\Sigma_{ii}\Sigma_{jj}}{(n-1)^2}\sum_{l,s,r=1}^n\left(U^{(n)}\right)^{-1}_{ls}\,U^{(n)}_{sr}\,\delta_{lr}\\
    =\frac{\Sigma_{ii}\Sigma_{jj}}{(n-1)^2}\sum_{l,s=1}^n \left(U^{(n)}\right)^{-1}_{ls}\,U^{(n)}_{sl} =
    \frac{\Sigma_{ii}\Sigma_{jj}}{(n-1)^2}\sum_{l=1}^n \delta_{ll} 
    =\frac{n}{(n-1)^2}\Sigma_{ii}\Sigma_{jj}\underset{n\rightarrow\infty}{\longrightarrow}0.
\end{align*}
Identically we can show that \eqref{ii} tends to zero. For \eqref{iii}, we have
\begin{eqnarray*}
    \frac{\Sigma_{jj}}{(n-1)^2}\sum_{l,s,k,r=1}^n\left(U^{(n)}\right)^{-1}_{ls}\left(U^{(n)}\right)^{-1}_{kr}\,U^{(n)}_{sr}\,\mu_{li}^{(n)}\,\mu_{ki}^{(n)} \\
 = \frac{\Sigma_{jj}}{(n-1)^2}\sum_{l,k,r=1}^n\left(U^{(n)}\right)^{-1}_{kr}\, \delta_{lr}\,\mu_{li}^{(n)}\,\mu_{ki}^{(n)}
    = \frac{\Sigma_{jj}}{(n-1)^2}\sum_{l,k=1}^n\left(U^{(n)}\right)^{-1}_{kl}\, \,\mu_{li}^{(n)}\,\mu_{ki}^{(n)}\\ =
    \sum_{r,r'=1}^{K^*}\frac{\Sigma_{jj}}{(n-1)^2}\sum_{l,k=1}^n\left(U^{(n)}\right)^{-1}_{kl}\, \,\mathds{1}\lbrace \mu_l^{(n)}=\theta_r\rbrace\,\mathds{1}\lbrace \mu_k^{(n)}=\theta_{r'}\rbrace\,\mu_{li}^{(n)}\,\mu_{ki}^{(n)}\,\theta_{ri}\,\theta_{r'i}\underset{n\rightarrow\infty}{\longrightarrow}0,
\end{eqnarray*}
where the limit is derived using Lemma~\ref{lemma_prop}. The same reasoning is used to show that \eqref{iv}, \eqref{v} and \eqref{vi} tend to zero when $n\rightarrow\infty$. Therefore, we have $\lim_{n\rightarrow\infty}\mathrm{Var}\bigl(A_{ij}^{(n)}\bigr)=0$. The same strategy, together with \eqref{claim_0} and \eqref{claim}, is used to show that $\underset{n\rightarrow\infty}{\lim}\mathrm{Var}\bigl(B_{ij}^{(n)}\bigr)= \underset{n\rightarrow\infty}{\lim}\mathrm{Var}\bigl(C_{ij}^{(n)}\bigr)=\underset{n\rightarrow\infty}{\lim}\mathrm{Var}\bigl(D_{ij}^{(n)}\bigr)=0$. Thus, we have \eqref{conv_proba}. Note that the sum in \eqref{conv_proba} can be written as the $ij$ term of a matrix. Indeed, we have
\begin{equation}
    \hat{\Sigma}_{ij}^{(n)}-\Sigma_{ij}\overset{p}{\rightarrow} \lambda_0\left(\Theta^T\,\mathrm{diag}(\pi_1,\ldots,\pi_{K^*})\,\Theta\right)_{ij},
\end{equation}
where $\Theta$ is a $p\times K^*$ matrix having as entries $\Theta_{ij}=\theta_{ij}-\Tilde{\theta}_{j}$. As $\lambda_0,\pi_1,\ldots,\pi_{K^*}\geq 0$, the matrix $\lambda_0(\Theta^T\,\mathrm{diag}(\pi_1,\ldots,\pi_{K^*})\,\Theta)$ is positive semi-definite, so the entries of $\hat{\mathbf{\Sigma}}\left(\mathbf{X}^{(n)}\right)-\mathbf{\Sigma}$ converge in probability to the entries of a positive semi-definite matrix. Note that, as both $\hat{\mathbf{\Sigma}}\left(\mathbf{X}^{(n)}\right)$ and $\mathbf{\Sigma}$ are positive definite, the eigenvalues of their difference are real. Finally, since the eigenvalues depend continuously on the entries of the matrix, the eigenvalues of $\hat{\mathbf{\Sigma}}\Bigl(\mathbf{X}^{(n)}\Bigr)-\mathbf{\Sigma}$ converge in probability to the eigenvalues of a positive semi-definite matrix, which are non-negative. Therefore, we have \eqref{asymp_overest}.

Let us conclude by showing \eqref{claim}. To do show, note that we can write,
\begin{equation*}
    1=\frac{1}{n}\sum_{k,l,s=1}^n \left(U^{(n)}\right)^{-1}_{lk}\,U^{(n)}_{ks}=\frac{2}{n}\sum_{s=1}^n\sum_{r=1}^{n-1}\sum_{i=1}^{n-r}\left(U^{(n)}\right)^{-1}_{i\,i+r}\,U^{(n)}_{i+r\,s}+\frac{1}{n}\sum_{s,i=1}^n\left(U^{(n)}\right)^{-1}_{ii}\,U^{(n)}_{is}.
\end{equation*}
Using the same reasoning as in the proof of Lemma~\ref{lemma_prop}, we have
\begin{equation*}
1 = 2\underset{n\rightarrow\infty}{\lim}\sum_{r=1}^{n-1}\,\lambda_r\left(\underset{n\rightarrow\infty}{\lim}\,\frac{1}{n}\sum_{i,s=1}^n U^{(n)}_{i+r\,s}\right)+\lambda_0\,\underset{n\rightarrow\infty}{\lim}\,\frac{1}{n}\sum_{i,s=1}^n U^{(n)}_{is},
\end{equation*}
which diverges unless the third limit is finite, which implies \eqref{claim}.
\end{proof}

\section{Non-maximal conditioning sets}\label{sec:finer_cond}

The methodology presented in Section~\ref{sec:CS_known_U} sets up the framework to perform selective inference after hierarchical clustering. Exploring its adaptation  to further clustering algorithms involves, as shown in \cite{chen2022selective}, the redefinition of $p$-values by constraining the conditional event that define \eqref{pvalue_gao} and \eqref{pvalue_V}. In this section, we revisit the procedure of post-clustering inference introduced in Section~\ref{sec:CS_known_U} and rewrite it in a more general form that allows its straightforward adaptation to the scenario where more conditioning is imposed.

When defining a $p$-value for \eqref{h0} that controls the selective type I error \eqref{sel_typeI}, one may think of conditioning only on having selected the pair of clusters that define the null hypothesis, i.e. on the event
\begin{equation}\label{M_set}
  M_{\mG_1,\mG_2}(\mathbf{X}) =\lbrace \mG_1,\mG_2\in\mathcal{C}(\mathbf{X})\rbrace.
\end{equation}
However, this is generally not enough to ensure the analytical tractability of the $p$-value. When considering a matrix normal distribution for the $p$-dimensional observations, two further conditions are imposed as shown in~\cite{Gao}. Following Section~\ref{sec:CS_known_U}, this corresponds to conditioning on the event
\begin{equation}\label{maximal_set}
    M_{\mG_1,\mG_2}(\mathbf{X})\,\cap\,\left\lbrace\boldsymbol\pi_{\nu}^\perp \mathbf{X}=\boldsymbol\pi_{\nu}^\perp \mathbf{x} \, ,\, \mathrm{dir}_{\mathbf{V}_{\mathcal{G}_1,\mathcal{G}_2}}\left(\mathbf{X}^T\nu\right)=\mathrm{dir}_{\mathbf{V}_{\mathcal{G}_1,\mathcal{G}_2}}\left(\mathbf{x}^T\nu\right) \right\rbrace,
\end{equation}
which is the maximal event for which any analytically tractable $p$-value has been shown to control \eqref{sel_typeI} under the general model \eqref{model}. If we denote by $T_{\mG_1,\mG_2}(\mathbf{X},\mathbf{x})$ the second set in \eqref{maximal_set}, we can rewrite \eqref{pvalue_V} as
\begin{equation}\label{M_and_T}
    p_{\mathbf{V}_{\mathcal{G}_1,\mathcal{G}_2}}(\mathbf{x};\lbrace \mG_1,\mG_2\rbrace)=\mathbb{P}_{H_0^{\lbrace \mG_1, \mG_2\rbrace}}\left(\norm{\mathbf{X}^T\nu}_{\mathbf{V}_{\mathcal{G}_1,\mathcal{G}_2}}\geq \norm{\mathbf{x}^T\nu}_{\mathbf{V}_{\mathcal{G}_1,\mathcal{G}_2}}\,\, \biggr\rvert\, M_{\mG_1,\mG_2}(\mathbf{X})\cap T_{\mG_1,\mG_2}(\mathbf{X}, \mathbf{x})\right).
\end{equation}
Then, from Theorem~\ref{th:pvalue_V} and its proof we can rewrite the truncation set in \eqref{pvalue_F} as
\begin{equation}\label{trunc_M}
    \mathcal{S}_{\mathbf{V}_{\mathcal{G}_1,\mathcal{G}_2}}(\mathbf{x};\lbrace \mG_1,\mG_2\rbrace) = \left\lbrace \phi\in\mathbb{R}\,:\, M_{\mG_1,\mG_2}\left(\mathbf{x}'_{{\mathbf{V}}_{\mG_1,\mG_2}}(\phi)\right)\right\rbrace,
\end{equation}
where $\mathbf{x}'_{{\mathbf{V}}_{\mG_1,\mG_2}}(\phi)$ is defined in \eqref{perturbed_x}. Consequently, in the conditions of Theorem~\ref{th:pvalue_V}, \eqref{pvalue_V} is analytically tractable as
\begin{equation}\label{pvalue_V_M}
  p_{\mathbf{V}_{\mathcal{G}_1,\mathcal{G}_2}}(\mathbf{x};\lbrace \mG_1,\mG_2\rbrace) = 1-\mathbb{F}_p\left(\norm{\mathbf{x}^T\nu}_{\mathbf{V}_{\mathcal{G}_1,\mathcal{G}_2}},\,\left\lbrace \phi\geq 0\,:\, M_{\mG_1,\mG_2}\left(\mathbf{x}'_{{\mathbf{V}}_{\mG_1,\mG_2}}(\phi)\right)\right\rbrace \right),
\end{equation}
where $\mathbb{F}_p$ is defined in Theorem~\ref{th:pvalue_V}. Uncoupling $M_{\mG_1,\mG_2}(\mathbf{X})$ and $T_{\mG_1,\mG_2}(\mathbf{X},\mathbf{x})$ in \eqref{M_and_T} allows us to characterize the null distribution of the $p$-value in terms of the conditioning event \eqref{M_set}. This is useful to study the scenarios where, for technical reasons, subsets of \eqref{M_set} are chosen to define the $p$-value for \eqref{h0}. This is the case in \cite{chen2022selective}, where the framework of \cite{Gao} under model \eqref{model_gao} has been adapted to perform selective inference after $k$-means clustering. To allow the efficient computation of their truncation set, the authors condition on $T_{\mG_1,\mG_2}(\mathbf{X},\mathbf{x})$ but also on all the intermediate clustering assignments for the $n$ observations \cite[Equation (9)]{chen2022selective}, which is a subset of \eqref{M_set}. In accordance with \eqref{trunc_M} and \eqref{pvalue_V_M}, this more restrictive conditioning yielded the same $p$-value \eqref{pvalue_gao} as in \cite{Gao} except from a different truncation set, based on the finer conditioning event. The following result characterizes this framework under our general model \eqref{model} and for an arbitrary non-maximal conditioning event. As such, it is a generalization of Theorem~\ref{th:pvalue_V}.

\begin{theorem}\label{th:pvalue_V_E} In the conditions of Theorem~\ref{th:pvalue_V}, let $\emptyset\neq E_{\mG_1,\mG_2}(\mathbf{X})\subset M_{\mG_1,\mG_2}(\mathbf{X})$ for any $(\mG_1,\mG_2)\in\mathcal{C}_{[n]}$. Then, the quantity
\begin{equation}\label{pvalue_V_E}
  p_{\mathbf{V}_{\mathcal{G}_1,\mathcal{G}_2}}(\mathbf{x};\lbrace \mathcal{G}_1,\mathcal{G}_2\rbrace;E_{\mG_1,\mG_2}) = \mathbb{P}_{H_0^{\lbrace \mathcal{G}_1, \mathcal{G}_2\rbrace}}\left(\norm{\mathbf{X}^T\nu}_{\mathbf{V}_{\mathcal{G}_1,\mathcal{G}_2}}\geq \norm{\mathbf{x}^T\nu}_{\mathbf{V}_{\mathcal{G}_1,\mathcal{G}_2}}\,\, \biggr\rvert\, E_{\mG_1,\mG_2}(\mathbf{X})\cap T_{\mG_1,\mG_2}(\mathbf{X},\mathbf{x})\right)
\end{equation}
is a $p$-value for~\eqref{h0} that controls the selective type I error for clustering \eqref{sel_typeI} at level $\alpha$. Furthermore, it satisfies
\begin{equation}\label{pvalue_V_E_dist}
    p_{\mathbf{V}_{\mathcal{G}_1,\mathcal{G}_2}}(\mathbf{x};\lbrace \mathcal{G}_1,\mathcal{G}_2\rbrace;E_{\mG_1,\mG_2}) = 1-\mathbb{F}_p\left(\norm{\mathbf{x}^T\nu}_{\mathbf{V}_{\mathcal{G}_1,\mathcal{G}_2}},\,\left\lbrace \phi\geq 0\,:\, E_{\mG_1,\mG_2}\left(\mathbf{x}'_{{\mathbf{V}}_{\mathcal{G}_1,\mathcal{G}_2}}(\phi)\right)\right\rbrace \right),
\end{equation}
where $\mathbb{F}_p(t,\mathcal{S})$ is the cumulative distribution function of a $\chi_p$ random variable truncated to the set $\mathcal{S}$ and $\mathbf{x}'_{{V}_{\mathcal{G}_1,\mathcal{G}_2}}(\phi)$ is defined in \eqref{perturbed_x}.
\end{theorem}

\begin{proof}[Proof of Theorem~\ref{th:pvalue_V_E}] We omit the proof of \eqref{pvalue_V_E_dist} as it is identical to the one of \eqref{pvalue_F}. Here, we show that the $p$-values defined using a non-maximal conditioning set $E(\mathbf{X})\subset M(\mathbf{X})$ as \eqref{pvalue_V_E} control the selective type I error for clustering \eqref{sel_typeI}. First, note that we have
\begin{equation}\label{pvalue_alpha}
    \mathbb{P}_{H_0^{\lbrace \mathcal{G}_1,\mathcal{G}_2}}\left( p_{\mathbf{V}_{\mathcal{G}_1,\mathcal{G}_2}}(\mathbf{x};\lbrace \mathcal{G}_1,\mathcal{G}_2\rbrace;E) \leq \alpha\,\biggr\vert\, E(\mathbf{X})\cap T(\mathbf{X})\right)=\alpha
\end{equation}
following \eqref{pvalue_V_E}, for any $\alpha\in(0,1)$. For simplicity, we will denote
\begin{equation}
    A=\mathds{1}\left\lbrace p_{\mathbf{V}_{\mathcal{G}_1,\mathcal{G}_2}}(\mathbf{x};\lbrace \mathcal{G}_1,\mathcal{G}_2\rbrace;E)\leq\alpha\right\rbrace.
\end{equation}
Then, following a similar reasoning as in the proof of \cite[Theorem 1]{Gao} and the tower property of conditional expectation, we can write
\begin{align}
    \mathbb{P}_{H_0^{\lbrace \mathcal{G}_1,\mathcal{G}_2\rbrace}}\left( p_{\mathbf{V}_{\mathcal{G}_1,\mathcal{G}_2}}(\mathbf{x};\lbrace \mathcal{G}_1,\mathcal{G}_2\rbrace;E) \leq \alpha\,\biggr\vert\, M(\mathbf{X})\right) = \mathbb{E}_{H_0^{\lbrace \mathcal{G}_1,\mathcal{G}_2\rbrace}}\left(A\,\biggr\vert\, M(\mathbf{X})\right)\\
    = \mathbb{E}_{H_0^{\lbrace \mathcal{G}_1,\mathcal{G}_2\rbrace}}\left[\mathbb{E}_{H_0^{\lbrace \mathcal{G}_1,\mathcal{G}_2\rbrace}}\left(A\,\biggr\vert\, M(\mathbf{X})\cap E(\mathbf{X})\cap T(\mathbf{X})\right)\,\biggr\vert\,M(\mathbf{X})\right]\\
    = \mathbb{E}_{H_0^{\lbrace \mathcal{G}_1,\mathcal{G}_2\rbrace}}\left[\mathbb{E}_{H_0^{\lbrace \mathcal{G}_1,\mathcal{G}_2\rbrace}}\left(A\,\biggr\vert\, E(\mathbf{X})\cap T(\mathbf{X})\right)\,\biggr\vert\,M(\mathbf{X})\right] =  \mathbb{E}_{H_0^{\lbrace \mathcal{G}_1,\mathcal{G}_2\rbrace}}\left[\alpha\,\biggr\vert\,M(\mathbf{X})\right]=\alpha,
\end{align}
where the third equality follows from the fact $E(\mathbf{X})\subset M(\mathbf{X})$ and the last equality follows from \eqref{pvalue_alpha}.
\end{proof}

Note that, following \eqref{pvalue_V_M}, replacing $E_{\mG_1,\mG_2}(\mathbf{X})$ by $M_{\mG_1,\mG_2}(\mathbf{X})$ yields exactly Theorem~\ref{th:pvalue_V}. Once again, the efficient computation of \eqref{pvalue_V_E_dist} depends on the efficient computation of the truncation set $E_{\mG_1,\mG_2}(\mathbf{x}'_{{\mathbf{V}}_{\mathcal{G}_1,\mathcal{G}_2}}(\phi))$. As shown for the maximal conditioning event in Lemma~\ref{equivalence_sets}, it suffices to characterize the truncation set when the perturbed data set $\mathbf{x}'$ is defined with respect to any norm. 

\begin{lemma}\label{equivalence_sets_E} Let $\mathbf{x}$ be a realization of $\mathbf{X}$ and $\mG_1,\mG_2$ an arbitrary pair of clusters in $\mathcal{C}(\mathbf{x})$. Let $\mathbf{x}'$ denote the set \eqref{perturbed_2} defined in \cite[Equation (12)]{Gao}. Then,
\begin{equation}
    E_{\mG_1,\mG_2}\left(\mathbf{x}'_{{\mathbf{V}}_{\mG_1,\mG_2}}(\phi)\right) = \frac{\norm{\mathbf{x}^T\nu}_{\mathbf{V}_{\mathcal{G}_1,\mathcal{G}_2}}}{\norm{\mathbf{x}^T\nu}_2}\,E_{\mG_1,\mG_2}\left(\mathbf{x}'(\phi)\right).
\end{equation}
\end{lemma}
The proof of Lemma~\ref{equivalence_sets_E} is omitted as it is identical to that of Lemma~\ref{equivalence_sets}. In \cite{chen2022selective}, the authors characterized $E_{\mG_1,\mG_2}(\mathbf{x}'(\phi))$ when $E_{\mG_1,\mG_2}$ corresponds to all intermediate clustering assignments of a $k$-means algorithm. Therefore, we can benefit from their efficient computation procedure and compute the truncation set under model~\eqref{model} using Lemma~\ref{equivalence_sets_E}. As such, we are able to perform selective inference after $k$-means clustering for $\bU\in\mathcal{CS}(n)$ and arbitrary $\bSigma$. The estimation procedure presented in Section~\ref{sec:unknown_sigma} remains identical for this case.  
\section{Supplementary Figures}

\begin{figure}[ht!]
    \centering
    \includegraphics[width=0.6\textwidth]{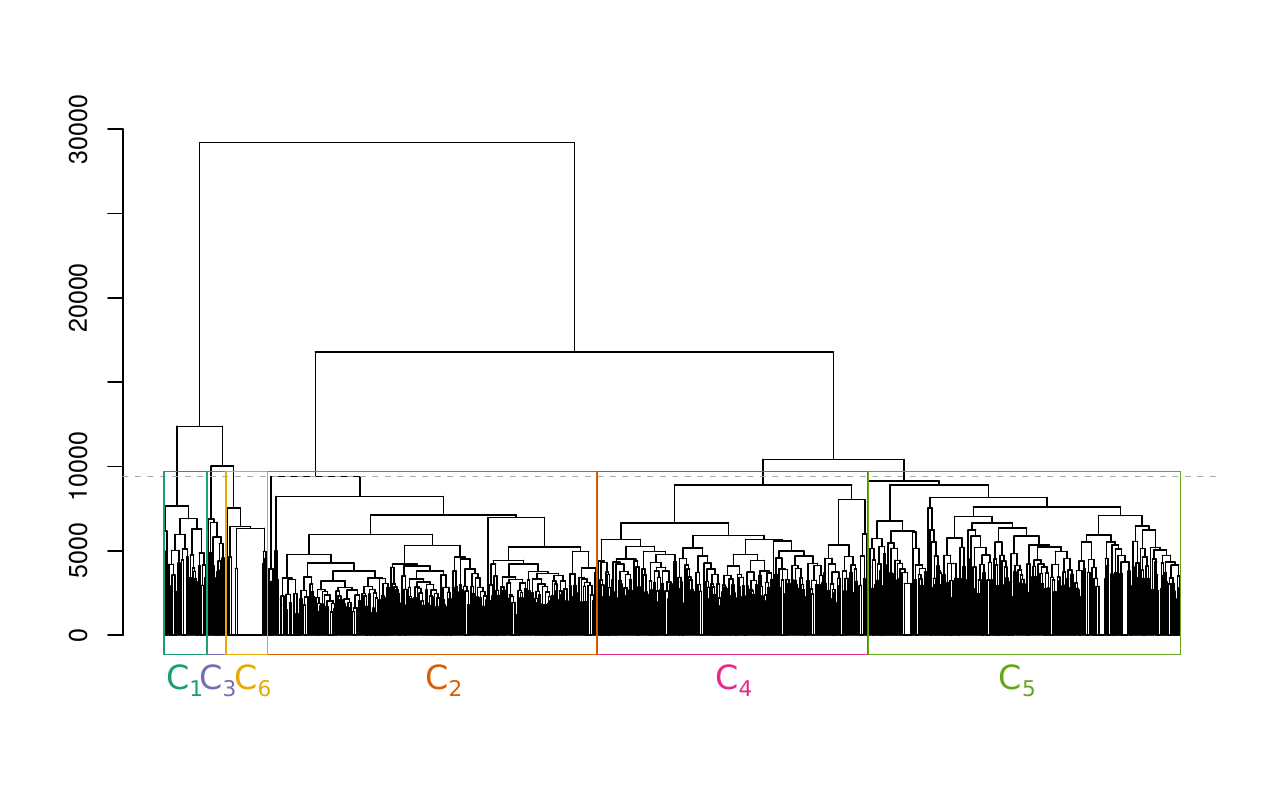}
    \caption{HAC dendrogram for the Hst5 protein ensemble data, with the six estimated clusters marked with colored rectangles.}
    \label{fig:dendrogram}
\end{figure}

\section{Additional numerical simulations}

In this section we describe the numerical experiments illustrated in Figures~\ref{fig:ignoredep} and \ref{fig:whitening} and present the results of the simulations described in Section~\ref{sec:numerical_experiments} when $\mathcal{C}$ is a $k$-means or a hierarchical agglomerative clustering (HAC) algorithm with centroid, single and complete linkages.

\subsection{Numerical simulation of Figure~\ref{fig:ignoredep}}\label{sec:ignoredep}

Figure~\ref{fig:ignoredep} simulates the null distribution of $p$-values defined in \cite{Gao} when data present dependence structures between observations and features, and $p$-values are computed assuming \eqref{model_gao}. We consider the general matrix normal model $\mathbf{X} \sim \mathcal{MN}_{n\times p}(\boldsymbol\mu,\mathbf{U},\mathbf{\Sigma})$, where we set $\boldsymbol\mu=\mathbf{0}_{n\times p}$, that is, the global null hypothesis. The matrices $\mathbf{U}\in\mathcal{M}_{n\times n}(\mathbb{R})$ and $\mathbf{\Sigma}\in\mathcal{M}_{p\times p}(\mathbb{R})$ encode the dependence structure between observations and features respectively. We choose $\mathbf{U}$ the covariance matrix of a stationary auto-regressive process of first order, AR(1), whose entries are given by $U_{ij}=\phi\rho^{\abs{i-j}}$, for $\phi>0$ and $\abs{\rho}<1$. The dependence between features is given by a Toeplitz matrix with entries $\Sigma_{ij}=1+1/\abs{i-j}$. We choose $\phi=1$, $\rho=0.2$ and generate $M=2000$ realizations of $\mathbf{X}$. For each one, we set the HAC algorithm with average linkage to choose three clusters and test for the difference in means of a pair of randomly selected clusters. The $p$-values are computed using the approach defined in \cite{Gao} assuming that $\mathbf{X}$ follows \eqref{model_gao} with $\sigma^2 = 2$, that is, neglecting the off-diagonal entries of the covariance matrices $\mathbf{U}$ and $\mathbf{\Sigma}$.

\subsection{Numerical simulation of Figure~\ref{fig:whitening}}\label{sec:whitening}

Figure~\ref{fig:whitening} illustrates the effect of whitening matrix normal data with dependent observations and features and performing post-clustering inference assuming \eqref{model_gao} afterwards. Data were first simulated from the general model~\eqref{model} with $n=100$, $p=2$. We set $\mathbf{U}$ as the covariance matrix of a AR(1) process, that is, $U_{ij}=\phi\rho^{\abs{i-j}}$ for $\phi>0$ and $\abs{\rho}<1$. We chose $\phi=1$ and $\rho=0.2$ The dependence between features was encoded by a Toeplitz matrix $\mathbf{\Sigma}$ with entries $\Sigma_{ij}=1+1/\abs{i-j}$. The mean matrix $\boldsymbol{\mu}$ divided the observations into three clusters and its entries were given by:
\begin{equation*}
    \mu_{i}=
      \begin{cases} 
      \left(-5,0,\ldots,0\right)& \mathrm{if }\,\,i \leq \lfloor \frac{n}{3}\rfloor,\\
      \left(0,\ldots,0,5\sqrt{3}\right) & \mathrm{if }\,\,\lfloor \frac{n}{3}\rfloor<i\leq \lfloor \frac{2n}{3}\rfloor,\\
      \left(5,0,\ldots,0\right) & \mathrm{otherwise},
   \end{cases}
     \quad\forall\,i\in[n].
\end{equation*}
The sample drawn from this model is presented in Figure~\ref{fig:whitening}(a). Its observations are classified into three groups using the $k$-means algorithms and compared using the $p$-values \eqref{pvalue_V} presented in this work, that account for the dependence structures $\mathbf{U}$ and $\mathbf{\Sigma}$. In panels (b,c), data is whitened by taking the transformation $(\mathbf{\Sigma}\otimes\mathbf{U})^{-\frac{1}{2}}\mathrm{vec}(\mathbf{X})$ and de-vectorizing the resulting random vector into a $n\times p$ matrix. Then, observations are classified into three groups using $k$-means (b) and HAC with average linkage (c) algorithms and the differences between cluster means are tested using the approaches proposed in \cite{chen2022selective} (b) and \cite{Gao} (c), that assume model \eqref{model_gao}.

\subsection[Numerical analysis of (p-Gamma)]{Numerical analysis of~\eqref{pvalue_gamma}}\label{app:simulations_gamma}

In this Section, we simulate the distribution of~\eqref{pvalue_gamma} under a global null hypothesis, that is, setting $\boldsymbol\mu=\mathbf{0}_{n\times p}$. Following Proposition~\ref{prop:gamma_chi}, the quantity~\eqref{pvalue_gamma} has the closed form~\eqref{pvalue_gamma_chi1}, allowing its implementation in practice. We follow the same strategy as in Section~\ref{sec:ignoredep}, generating $M=2000$ realizations of $\bX\sim\mathcal{MN}_{n\times p}(\mathbf{0}_{n\times,p},\bU,\bSigma)$, setting the HAC algorithm to choose three clusters and computing~\eqref{pvalue_gamma_chi1} for a pair of randomly selected groups. We choose $\bSigma$ to be a diagonal matrix with entries $\Sigma_{ii}=1+1/i$, and repeat the simulation under the following three settings:
\begin{enumerate}
    \item[(D4)] $\mathbf{U}$ is a diagonal matrix with entries $U_{ii}=1+1/i$.
    \item[(D5)] $\mathbf{U}$ is the covariance matrix of an AR(1) model with $\sigma=1$ and $\rho = 0.1$.
    \item[(D6)] $\mathbf{U}$ is the covariance matrix of an AR(2) model with $\sigma=1$, $\beta_1=0.4$ and $\beta_2=0.1$.
\end{enumerate}
Note that the truncation set in~\eqref{pvalue_gamma_chi1} has slightly changed with respect to~\eqref{set_S_hat}, due to the relaxation of the direction equality in~\eqref{pvalue_gamma}, that now includes the event $\lbrace \mathrm{dir}(\bX^T\nu)=-\mathrm{dir}(\bx^T\nu)\rbrace$. As shown in Proposition~\ref{prop:gamma_chi}, this yields a broader truncation set~\eqref{S_gamma} including also perturbations in the sense of $-\bx^T\nu$. Adapting the efficient characterization of~\eqref{set_S_hat} to this setting is not straightforward. However, this is immediate under a Monte Carlo computation of~\eqref{pvalue_gamma_chi1}, as we only need to replace $\mathcal{C}(\mathbf{x}'(\omega_i))$ by $\mathcal{C}(\mathbf{x}'(\pm\omega_i))$ in~\eqref{monte_carlo_pv}. As this is sufficient for the purpose of this analysis, we limit this experience to HAC clustering with complete linkage. Results, showing that selective type I error is not controlled in any of the previous settings, are presented in Figure~\ref{fig:gamma-sim}.

\begin{figure}[t]
    \centering
    \includegraphics[width=0.95\textwidth]{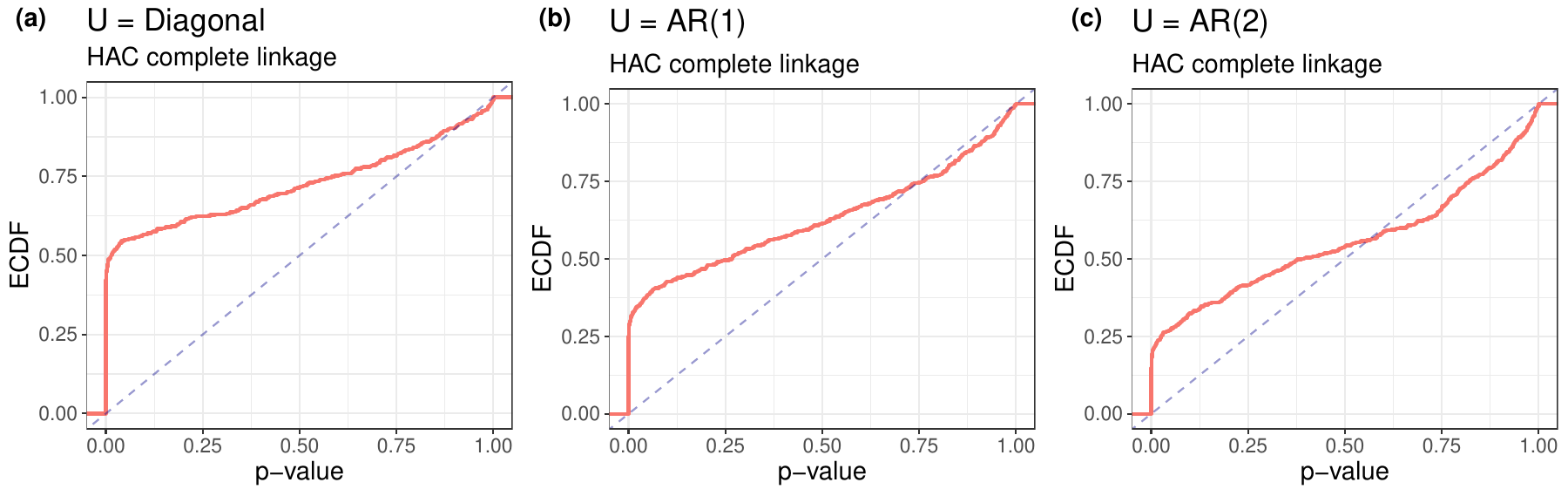}
    \caption{Empirical cumulative distribution functions (ECDF) of quantities \eqref{pvalue_gamma_chi1} with $\mathcal{C}$ being a hierarchical agglomerative clustering algorithm (HAC) with complete linkage. The ECDF were computed from $M=2000$ realizations of \eqref{model} under the three dependence settings (D4), (D5) and (D6) with $\boldsymbol{\mu}=\mathbf{0}_{n\times p}$, $n=20$ and $p=5$.}
    \label{fig:gamma-sim}
\end{figure}

\subsection{Additional numerical simulations of Section~\ref{sec:numerical_experiments}}\label{sec:extra_sim}

In this section, we present the counterparts of Figures~\ref{fig:average_global_h0}, \ref{fig:sim_overest}, \ref{fig:non_CS_U_average}, \ref{fig:non_CS_U_est_average}, \ref{fig:non_ad_U_average} and \ref{fig:ignoredep_average} for $k$-means and HAC with centroid, single and complete linkage. In Figures~\ref{fig:non_admissible_U_extra} and \ref{fig:ignoredep_others}, the simulation for $k$-means was performed for $\delta\in\lbrace 6,8,10\rbrace$, as the proportion of samples for which the null hypothesis held was very low for $\delta=4$.

\begin{figure}[ht!]
    \centering
    \includegraphics[width=0.95\textwidth]{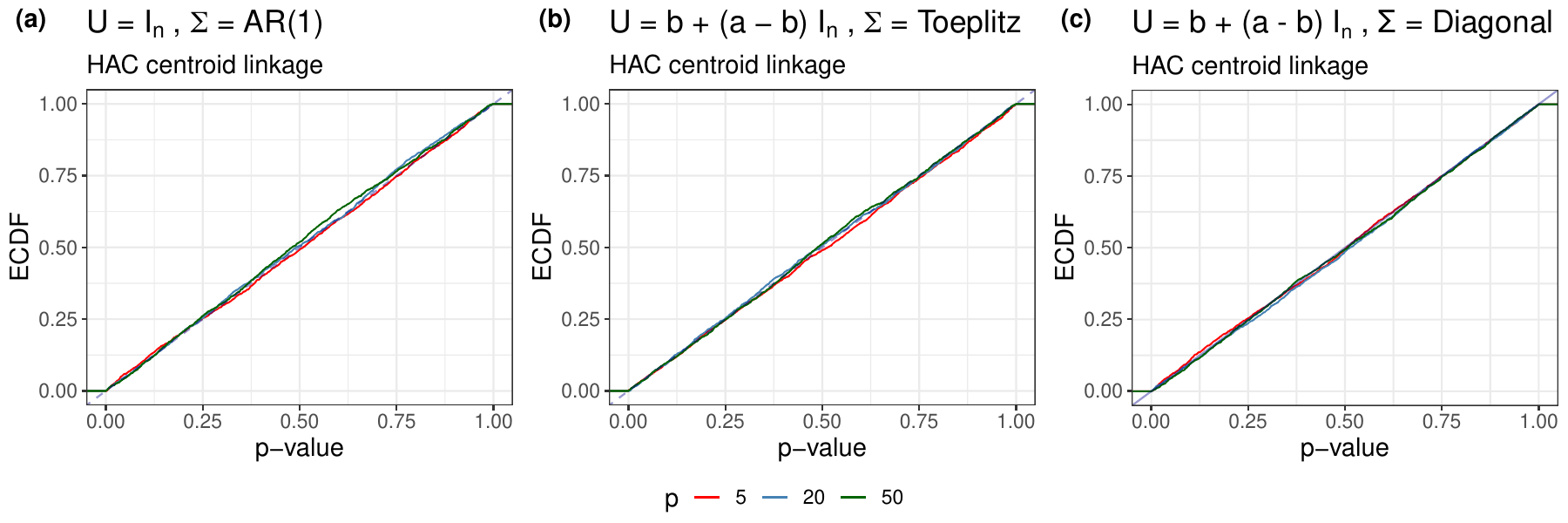}
    \includegraphics[width=0.95\textwidth]{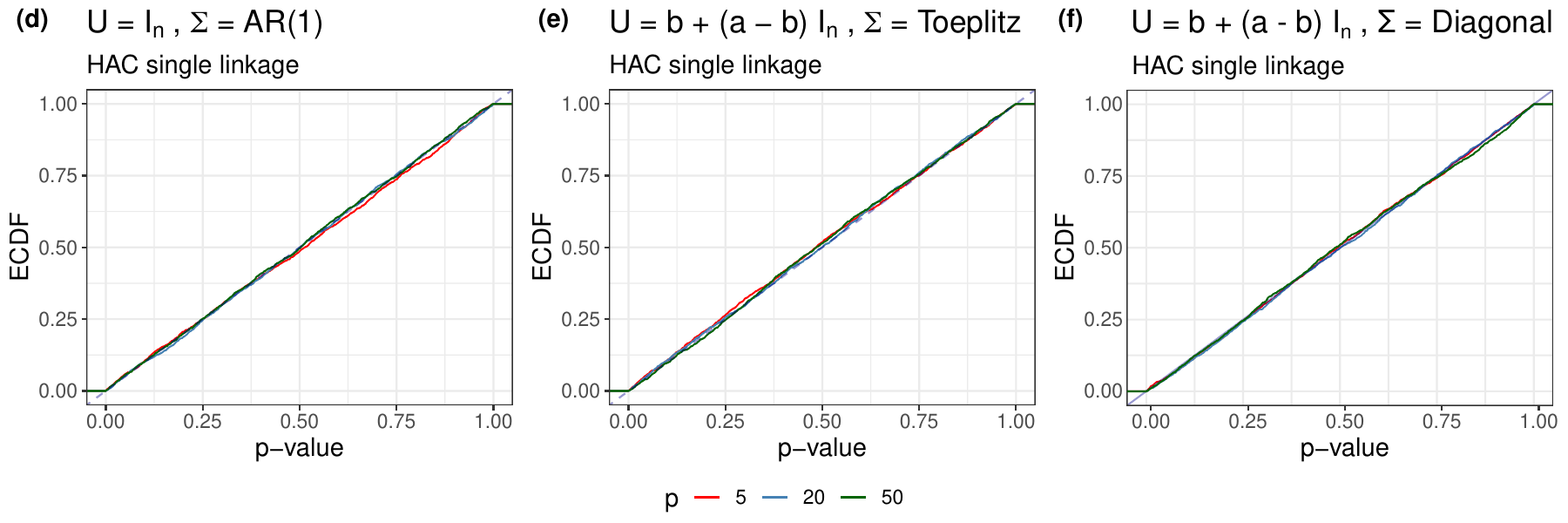}
    \includegraphics[width=0.95\textwidth]{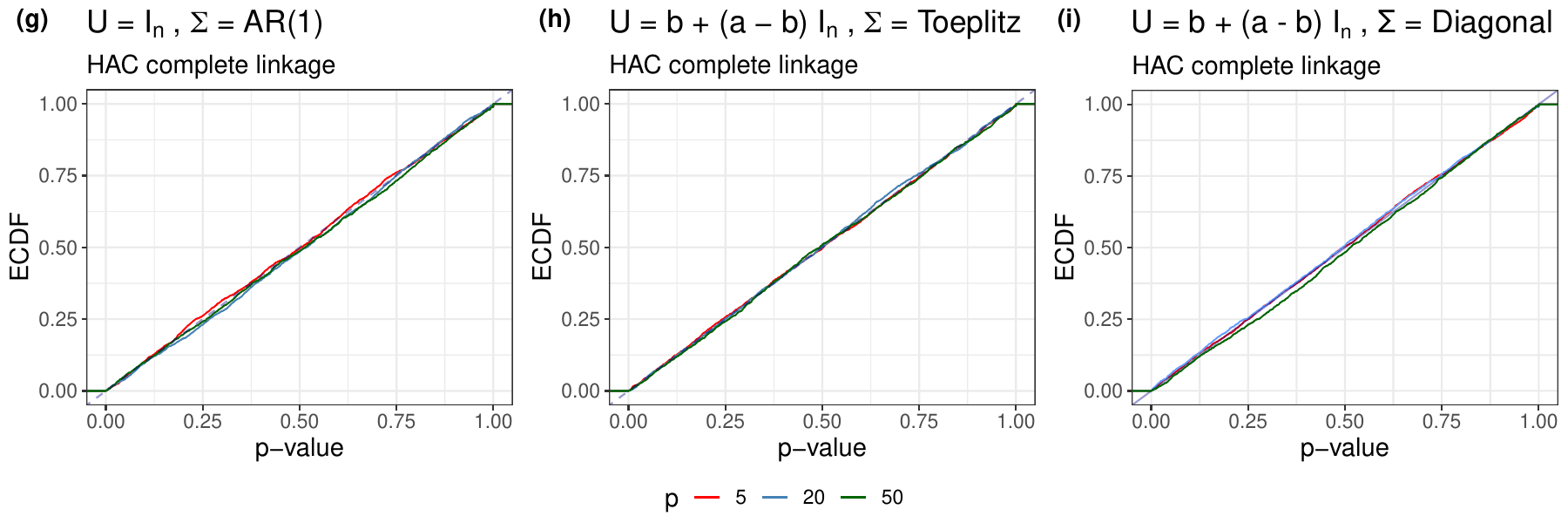}
    \includegraphics[width=0.95\textwidth]{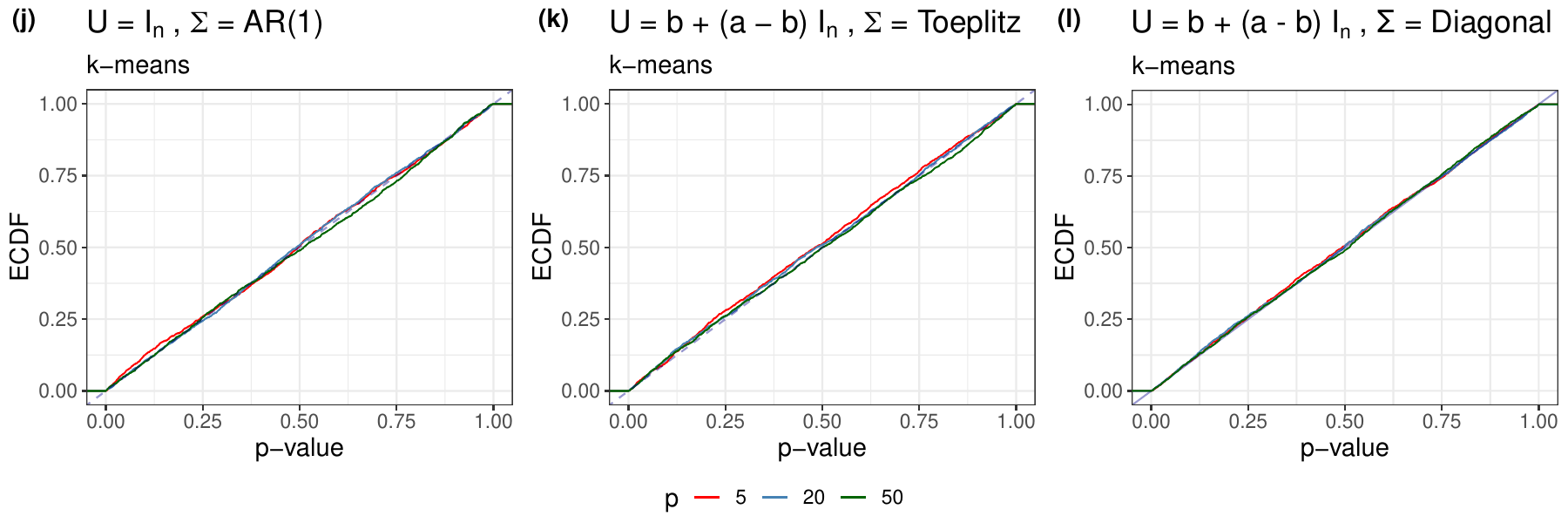}
    \caption{Empirical cumulative distribution functions (ECDF) of $p$-values \eqref{pvalue_V} with $\mathcal{C}$ being a hierarchical agglomerative clustering algorithm (HAC) with centroid (a-c), single (d-f) and complete (g-i) linkage and a $k$-means algorithm (j-l). The ECDF were computed from $M=2000$ realizations of \eqref{model} under the three dependence settings $(D1)$, $(D2)$ and $(D3)$ with $\boldsymbol{\mu}=\mathbf{0}_{n\times p}$, $n=100$ and $p\in\lbrace 5,20,50\rbrace$.}
    \label{fig:global_others}
\end{figure}

\begin{figure}[ht!]
    \centering
    \includegraphics[width=0.95\textwidth]{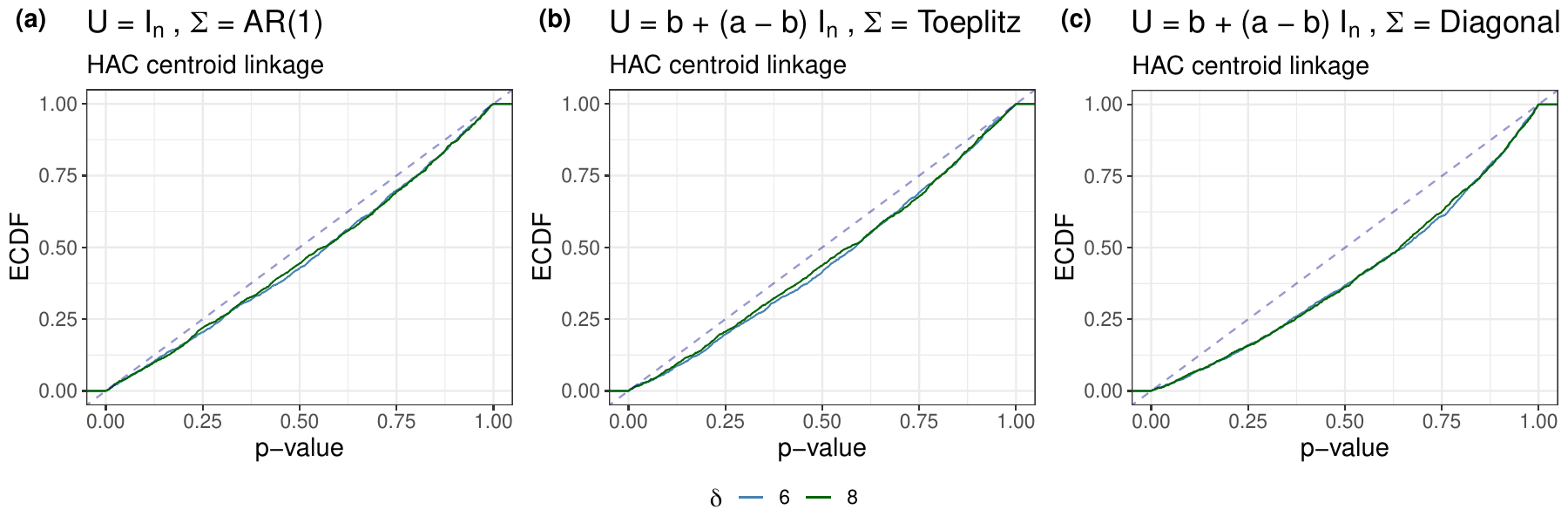}
    \includegraphics[width=0.95\textwidth]{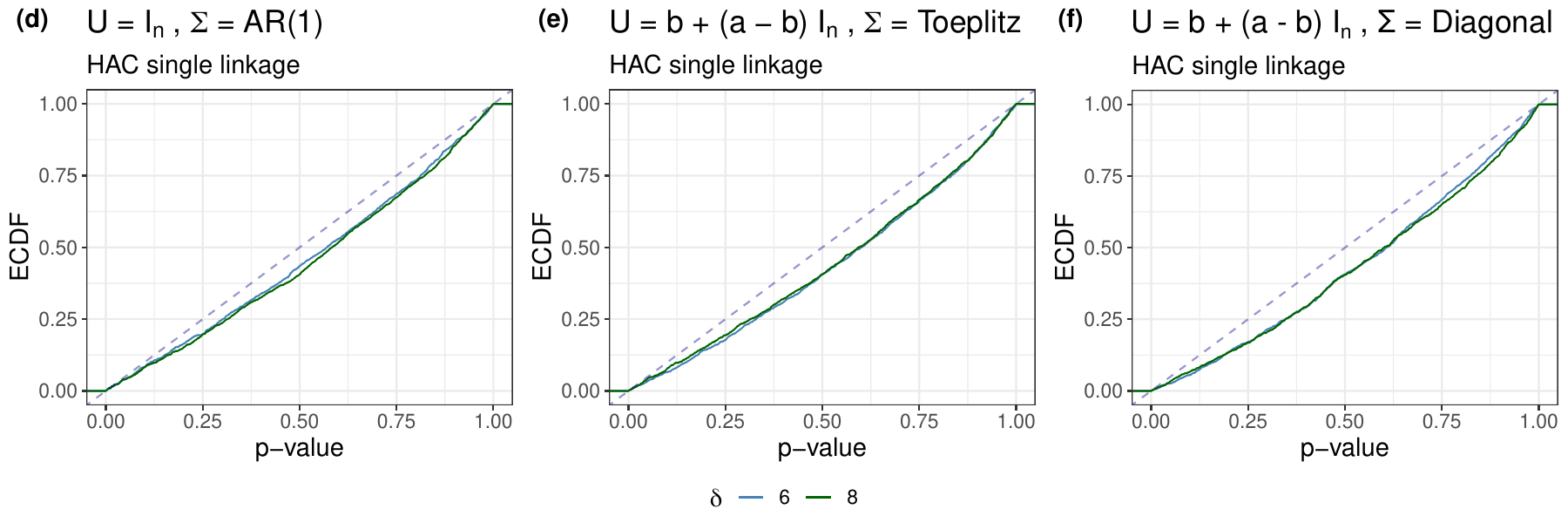}
    \includegraphics[width=0.95\textwidth]{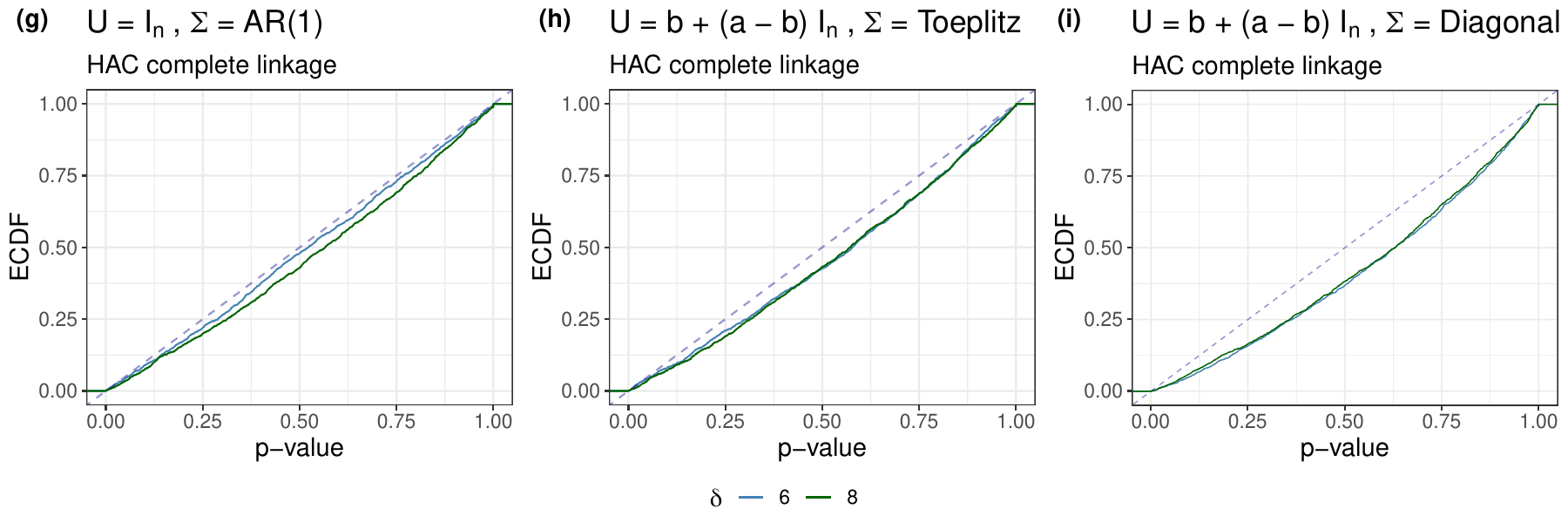}
    \includegraphics[width=0.95\textwidth]{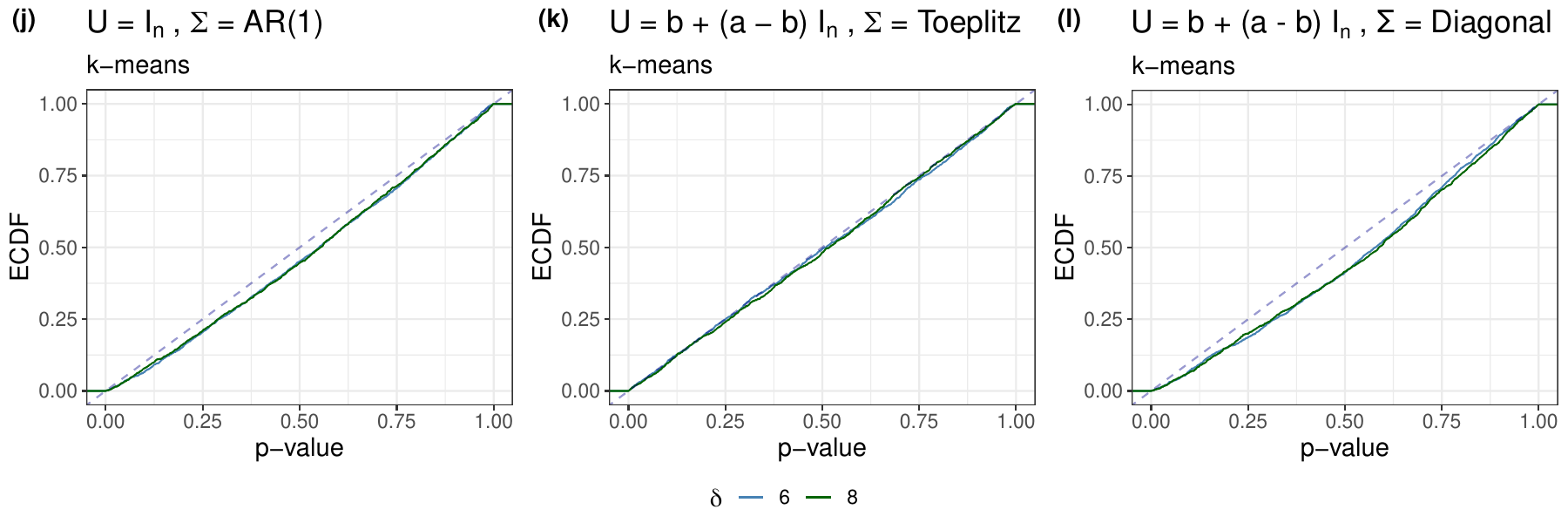}
    \caption{Empirical cumulative distribution functions (ECDF) of $p$-values \eqref{hat_pvalue} with $\mathcal{C}$ being a HAC algorithm with centroid (a-c), single (d-f) and complete (g-i) linkage and a $k$-means algorithm (j-l). The ECDF were computed from $M=5000$ realizations of \eqref{model} under the three dependence settings $(D1)$, $(D2)$ and $(D3)$ with $n=100$, $p=5$ and $\boldsymbol\mu$ given by \eqref{mu_est}. Only samples for which the null hypothesis held were kept, as described in Section~\ref{sec:sigma_est}.}
    \label{fig:est_others}
\end{figure}

\begin{figure}[ht!]
    \centering
    \includegraphics[width=0.95\textwidth]{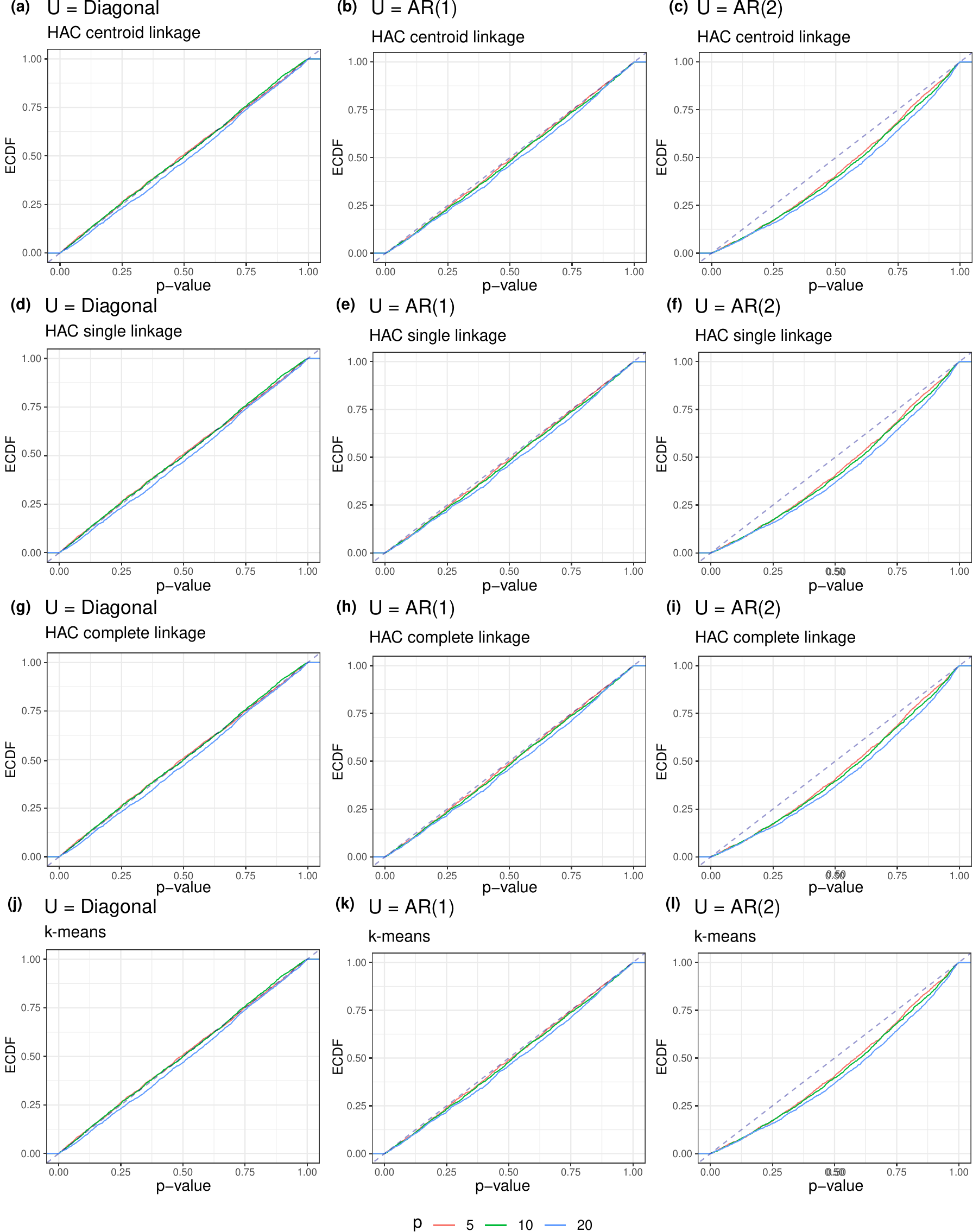}
    \caption{Empirical cumulative distribution functions (ECDF) of $p$-values \eqref{pvalue_V} with $\mathcal{C}$ being a hierarchical agglomerative clustering algorithm (HAC) with centroid (a-c), single (d-f) and complete (g-i) linkage and a $k$-means algorithm (j-l). The ECDF were computed from $M=2000$ realizations of \eqref{model} under the three dependence settings $(D4)$, $(D5)$ and $(D6)$ with $\boldsymbol{\mu}=\mathbf{0}_{n\times p}$, $n=100$ and $p\in\lbrace 5,20,50\rbrace$.}
    \label{fig:nonCS_global_others}
\end{figure}

\begin{figure}[ht!]
    \centering
    \includegraphics[width=0.95\textwidth]{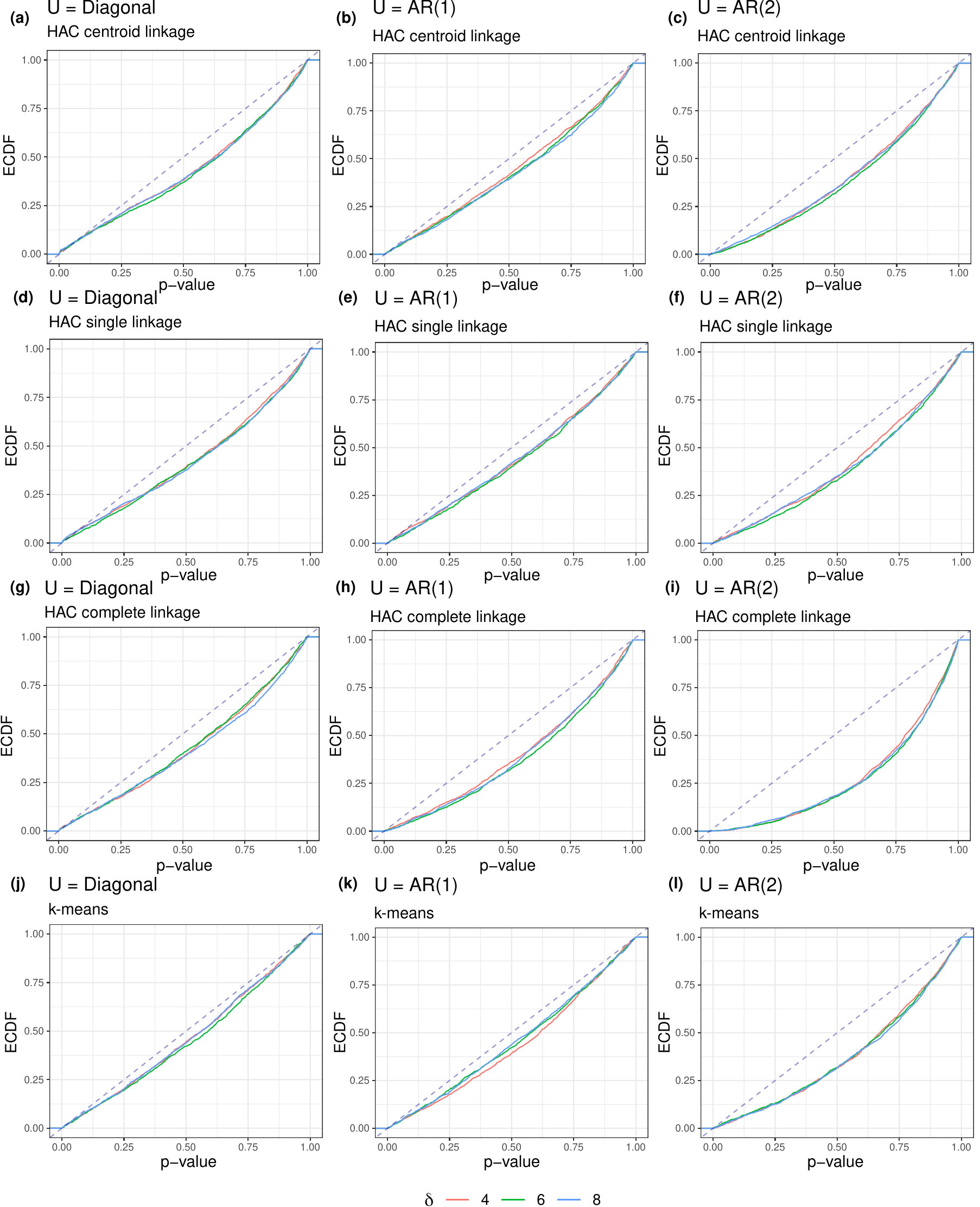}
    \caption{Empirical cumulative distribution functions (ECDF) of $p$-values \eqref{hat_pvalue} with $\mathcal{C}$ being a HAC algorithm with centroid (a-c), single (d-f) and complete (g-i) linkage and a $k$-means algorithm (j-l). The ECDF were computed from $M=5000$ realizations of \eqref{model} under the three dependence settings $(D4)$, $(D5)$ and $(D6)$ with $n=100$, $p=5$ and $\boldsymbol\mu$ given by \eqref{mu_est}. Only samples for which the null hypothesis held were kept, as described in Section~\ref{sec:sigma_est}.}
    \label{fig:nonCS_overest_others}
\end{figure}

\begin{figure}[ht!]
    \centering
    \includegraphics[width=0.85\textwidth]{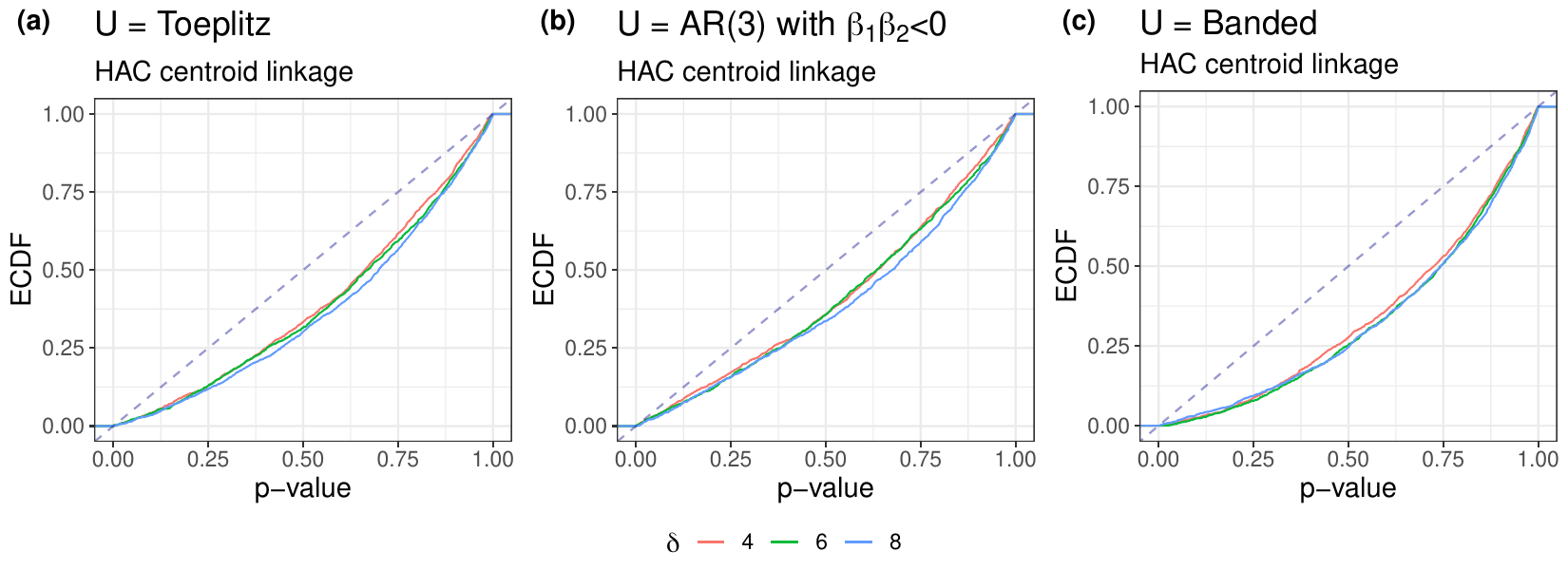}
    \includegraphics[width=0.85\textwidth]{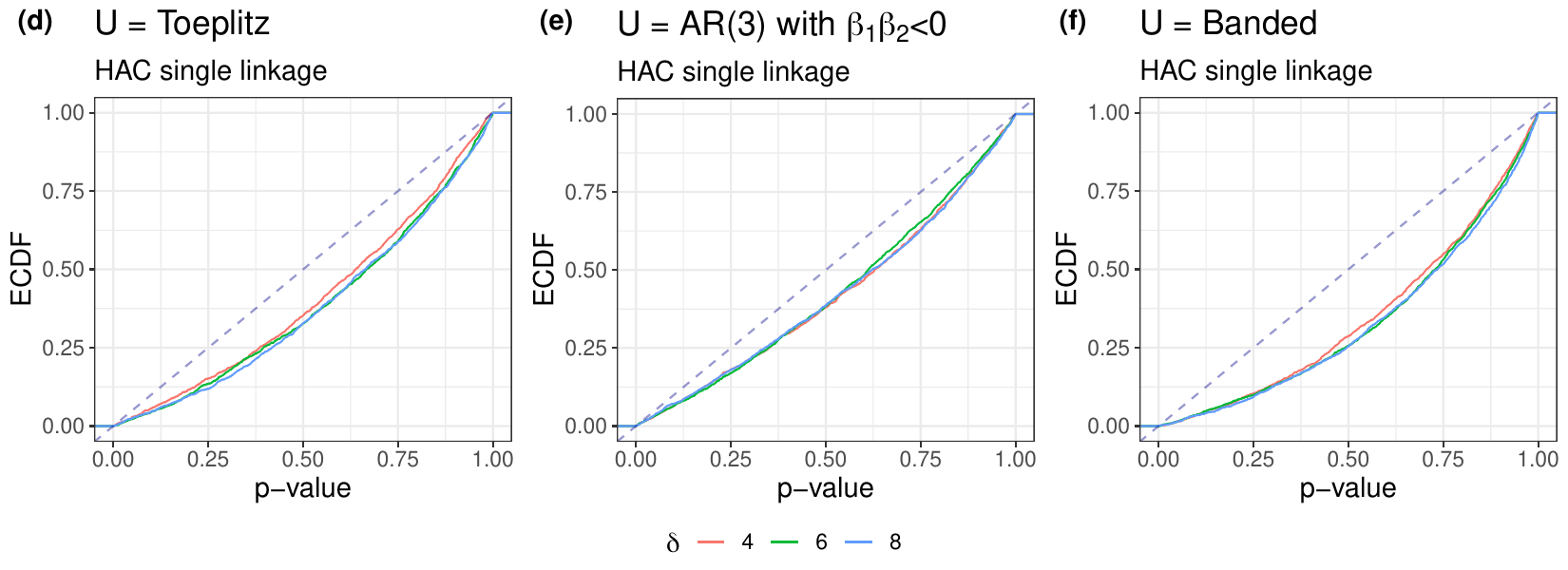}
    \includegraphics[width=0.85\textwidth]{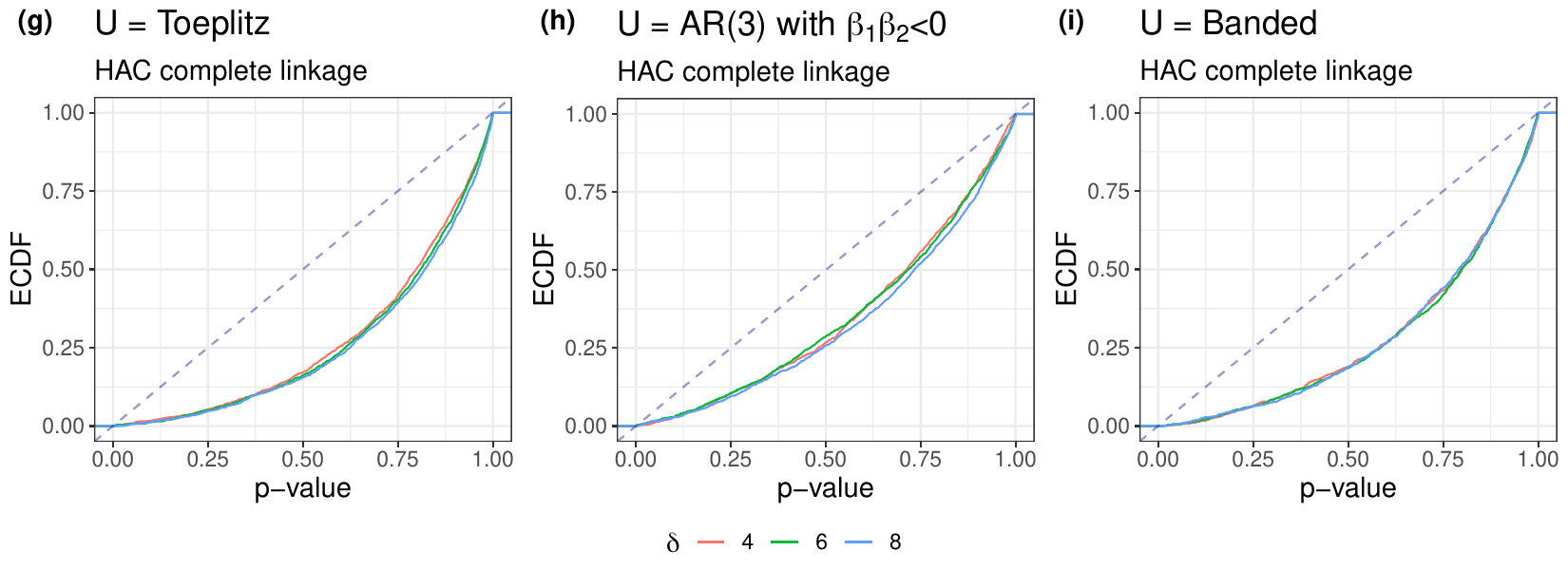}
    \includegraphics[width=0.85\textwidth]{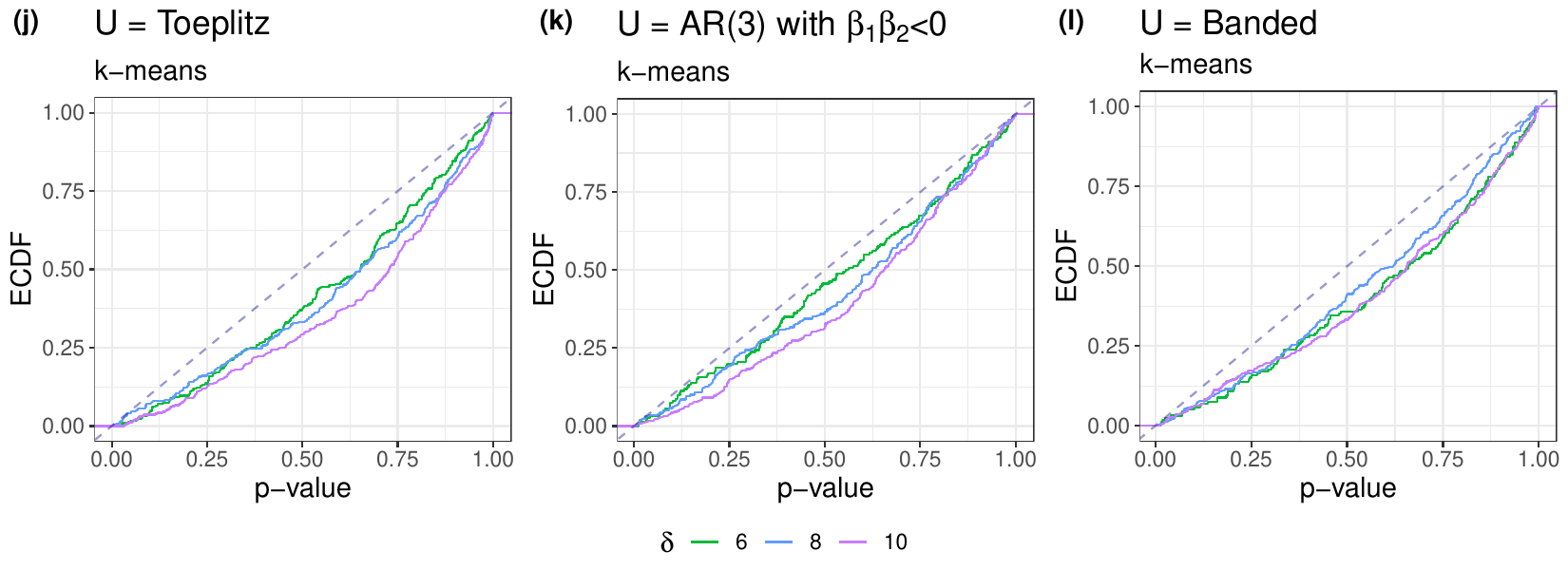}
    \caption{Empirical cumulative distribution functions (ECDF) of $p$-values \eqref{hat_pvalue} with $\mathcal{C}$ being a HAC algorithm with centroid (a-c), single (d-f) and complete (g-i) linkage and a $k$-means algorithm (j-l). The ECDF were computed from $M=5000$ realizations of \eqref{model} under the three dependence settings $(D7)$, $(D8)$ and $(D9)$ with $n=50$, $p=5$ and $\boldsymbol\mu$ given by \eqref{mu_est}. Only samples for which the null hypothesis held were kept, as described in Section~\ref{sec:no_ad_U}.}
    \label{fig:non_admissible_U_extra}
\end{figure}

\begin{figure}[ht!]
    \centering
    \includegraphics[width=0.85\textwidth]{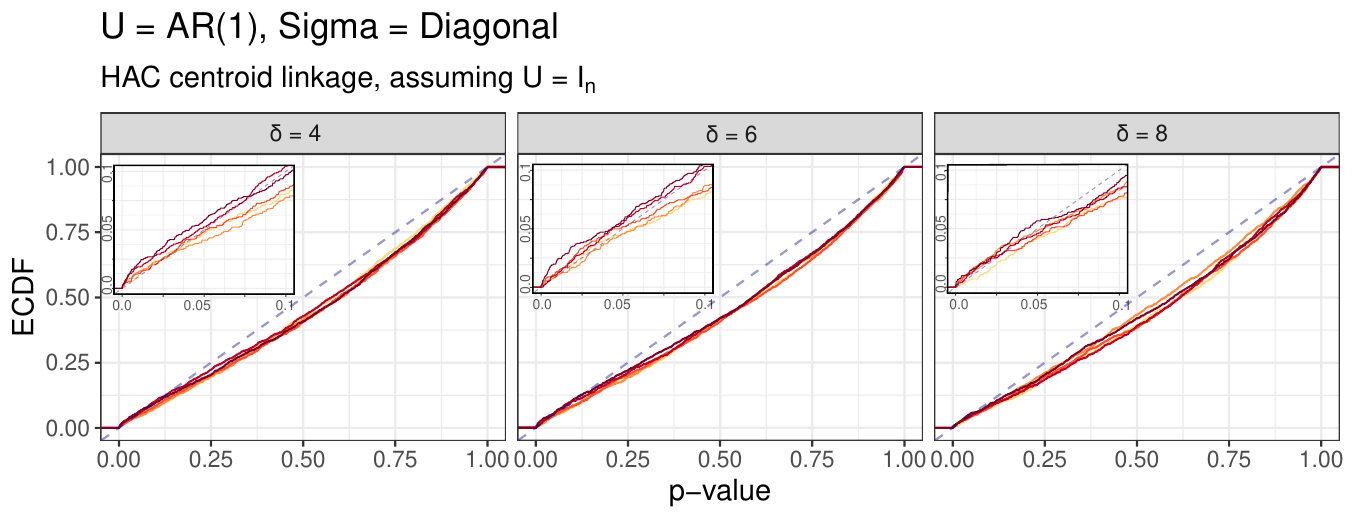}
    \includegraphics[width=0.85\textwidth]{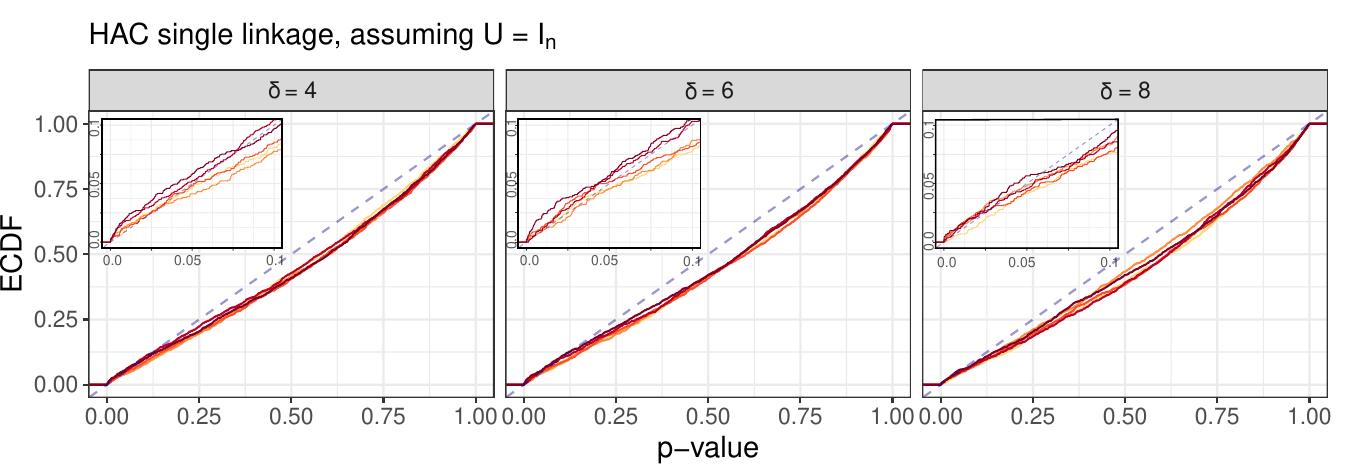}
    \includegraphics[width=0.85\textwidth]{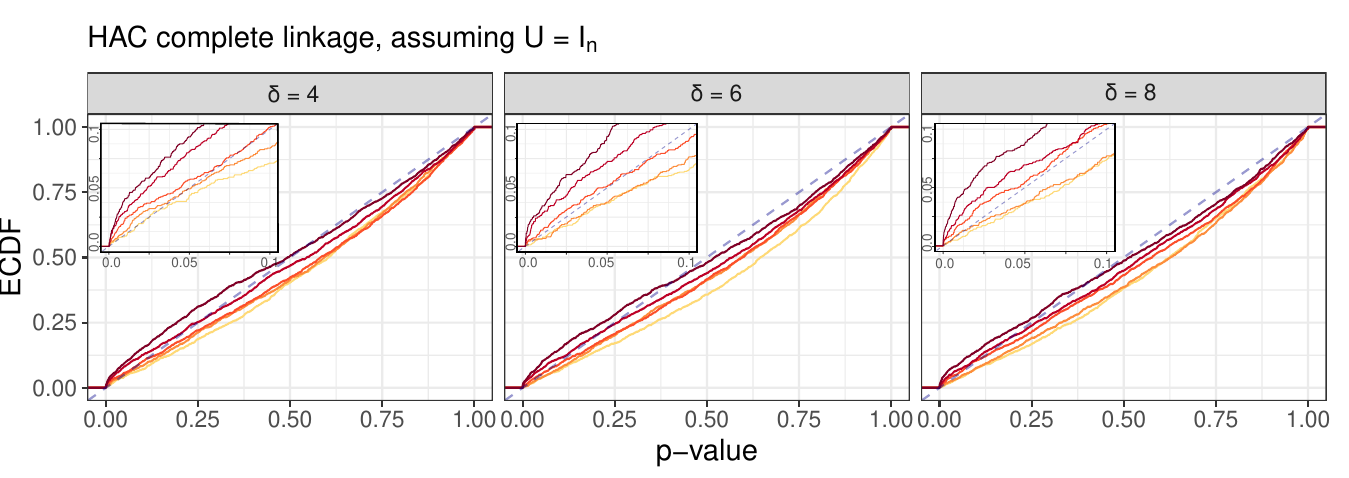}
   \includegraphics[width=0.85\textwidth]{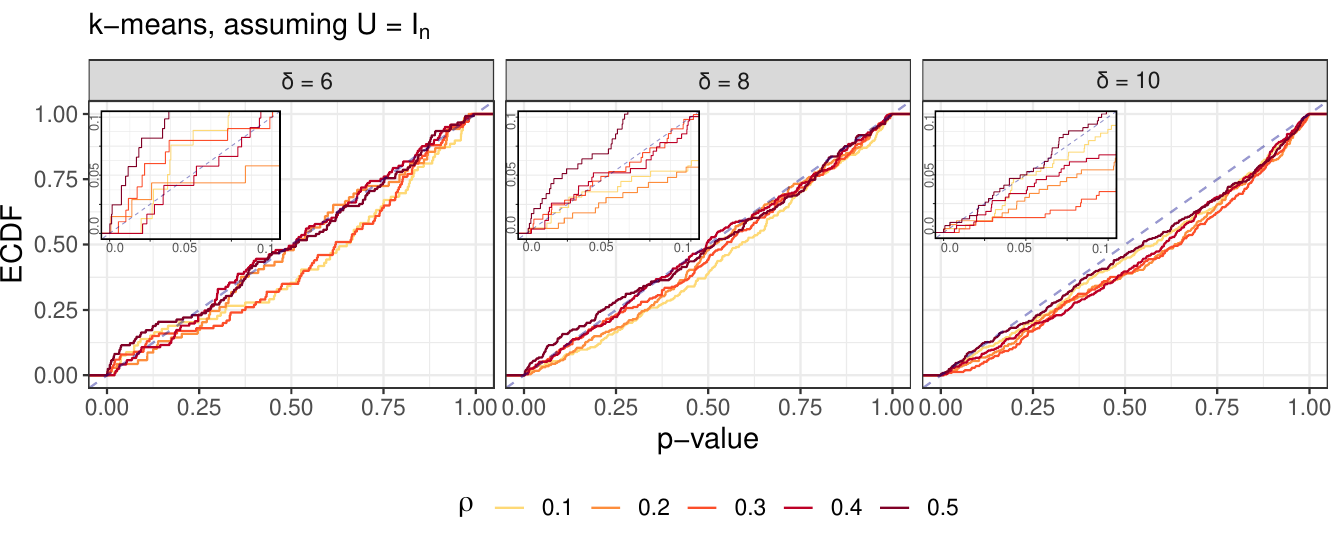}
    \caption{Empirical cumulative distribution functions (ECDF) of $p$-values \eqref{hat_pvalue} with $\mathcal{C}$ being a HAC algorithm with centroid, single and complete linkage and a $k$-means algorithm. The ECDF were computed from $M=5000$ realizations of \eqref{model} as described in Section~\ref{sec:ignoredep_U} with $n=50$, $p=5$ and $\boldsymbol{\mu}$ given by \eqref{mu_est} with $\delta\in\lbrace 4,6,8\rbrace$ for HAC and $\delta\in\lbrace 6,8,10\rbrace$ for $k$-means. Only samples for which the null hypothesis held were kept, as described in Section~\ref{sec:ignoredep_U}.}
    \label{fig:ignoredep_others}
\end{figure}

\clearpage
\bibliographystyle{abbrv} 
\bibliography{PCI_dependence.bib}

\begin{thebibliography}{10}

\bibitem{Ahmed2015Oct}
S.~E. Ahmed, S.~Fallahpour, D.~von Rosen, and T.~von Rosen.
\newblock {Estimation of Several Intraclass Correlation Coefficients}.
\newblock {\em Comm. Statist. Simulation Comput.}, Oct. 2015.

\bibitem{allen}
M.~P. Allen and D.~J. Tildesley.
\newblock {\em {Computer Simulation of Liquids}}.
\newblock Oxford University Press, 06 2017.

\bibitem{Appadurai2022}
R.~Appadurai, J.~K. Koneru, M.~Bonomi, P.~Robustelli, and A.~Srivastava.
\newblock Clustering heterogeneous conformational ensembles of intrinsically disordered proteins with t-distributed stochastic neighbor embedding.
\newblock {\em Journal of Chemical Theory and Computation}, June 2023.

\bibitem{bartlett}
M.~S. Bartlett.
\newblock {An Inverse Matrix Adjustment Arising in Discriminant Analysis}.
\newblock {\em The Annals of Mathematical Statistics}, 22(1):107 -- 111, 1951.

\bibitem{Bernado:2005}
P.~Bernadó, L.~Blanchard, P.~Timmins, D.~Marion, R.~W.~H. Ruigrok, and M.~Blackledge.
\newblock A structural model for unfolded proteins from residual dipolar couplings and small-angle {x}-ray scattering.
\newblock {\em Proceedings of the National Academy of Sciences of the United States of America}, 102(47):17002--17007, 2005.

\bibitem{CamachoZarco2020}
A.~R. Camacho-Zarco, S.~Kalayil, D.~Maurin, N.~Salvi, E.~Delaforge, S.~Milles, M.~R. Jensen, D.~J. Hart, S.~Cusack, and M.~Blackledge.
\newblock Molecular basis of host-adaptation interactions between influenza virus polymerase {PB}2 subunit and {ANP}32a.
\newblock {\em Nature Communications}, 11(1), July 2020.

\bibitem{chen2022powerful}
Y.~Chen, S.~Jewell, and D.~Witten.
\newblock More powerful selective inference for the graph fused lasso.
\newblock {\em Journal of Computational and Graphical Statistics}, 32(2):577--587, 2023.

\bibitem{Chen2023Nov}
Y.~T. Chen and L.~L. Gao.
\newblock {Testing for a difference in means of a single feature after clustering}.
\newblock {\em arXiv}, 2023.

\bibitem{chen2022selective}
Y.~T. Chen and D.~M. Witten.
\newblock Selective inference for k-means clustering.
\newblock {\em Journal of Machine Learning Research}, 24(152):1--41, 2023.

\bibitem{Cochran1934Apr}
W.~G. Cochran.
\newblock {The distribution of quadratic forms in a normal system, with applications to the analysis of covariance}.
\newblock {\em Math. Proc. Cambridge Philos. Soc.}, 30(2):178--191, Apr. 1934.

\bibitem{engens}
A.~Conev, M.~M. Rigo, D.~Devaurs, A.~F. Fonseca, H.~Kalavadwala, M.~V. de~Freitas, C.~Clementi, G.~Zanatta, D.~A. Antunes, and L.~E. Kavraki.
\newblock {EnGens: a computational framework for generation and analysis of representative protein conformational ensembles}.
\newblock {\em Briefings in Bioinformatics}, 24(4):bbad242, 07 2023.

\bibitem{elementwise_toep}
P.~M. Crespo and J.~Gutierrez-Gutierrez.
\newblock On the elementwise convergence of continuous functions of hermitian banded toeplitz matrices.
\newblock {\em IEEE Transactions on Information Theory}, 53(3):1168--1176, 2007.

\bibitem{expdecay}
S.~Demko, W.~F. Moss, and P.~Smith.
\newblock Decay rates for inverses of band matrices.
\newblock {\em Mathematics of Computation}, 43:491--499, 1984.

\bibitem{Derksen2021}
H.~Derksen and V.~Makam.
\newblock Maximum likelihood estimation for matrix normal models via quiver representations.
\newblock {\em SIAM Journal on Applied Algebra and Geometry}, 5(2):338--365, 2021.

\bibitem{Drton2024Jan}
M.~Drton, A.~Grosdos, and A.~McCormack.
\newblock {Rational Maximum Likelihood Estimators of Kronecker Covariance Matrices}.
\newblock {\em arXiv}, Jan. 2024.

\bibitem{Drton2021Oct}
M.~Drton, S.~Kuriki, and P.~Hoff.
\newblock {Existence and uniqueness of the Kronecker covariance MLE}.
\newblock {\em Ann. Stat.}, 49(5):2721--2754, Oct. 2021.

\bibitem{MLE}
P.~Dutilleul.
\newblock The {MLE} algorithm for the matrix normal distribution.
\newblock {\em Journal of Statistical Computation and Simulation}, 64(2):105--123, 1999.

\bibitem{Dyson:2005}
H.~J. Dyson and P.~E. Wright.
\newblock Intrinsically unstructured proteins and their functions.
\newblock {\em Nature Reviews Molecular Cell Biology}, 6:197--208, 2005.

\bibitem{Eaton2007}
M.~L. Eaton.
\newblock {\em Multivariate Statistics}.
\newblock SPIE, Jan. 2007.

\bibitem{fithian2017optimal}
W.~Fithian, D.~Sun, and J.~Taylor.
\newblock Optimal inference after model selection, 2017.
\newblock arXiv:1410.2597.

\bibitem{Gao}
L.~L. Gao, J.~Bien, and D.~Witten.
\newblock Selective inference for hierarchical clustering.
\newblock {\em Journal of the American Statistical Association}, 0(0):1--11, 2022.

\bibitem{Gray}
R.~M. Gray.
\newblock Toeplitz and circulant matrices: A review.
\newblock {\em Foundations and Trends® in Communications and Information Theory}, 2(3):155--239, 2006.

\bibitem{Gupta2018}
A.~Gupta and D.~Nagar.
\newblock {\em Matrix Variate Distributions}.
\newblock Chapman and Hall/CRC, 2018.

\bibitem{hivert2022post-clustering}
B.~Hivert, D.~Agniel, R.~Thi{\ifmmode\acute{e}\else\'{e}\fi}baut, and B.~P. Hejblum.
\newblock {Post-clustering difference testing: Valid inference and practical considerations with applications to ecological and biological data}.
\newblock {\em Comput. Statist. Data Anal.}, 193:107916, May 2024.

\bibitem{holm}
S.~Holm.
\newblock A simple sequentially rejective multiple test procedure.
\newblock {\em Scandinavian Journal of Statistics}, 6(2):65--70, 1979.

\bibitem{horn2013matrix}
R.~Horn and C.~Johnson.
\newblock {\em Matrix Analysis}.
\newblock Matrix Analysis. Cambridge University Press, 2013.

\bibitem{Jewell2022}
S.~Jewell, P.~Fearnhead, and D.~Witten.
\newblock Testing for a change in mean after changepoint detection.
\newblock {\em Journal of the Royal Statistical Society Series B: Statistical Methodology}, 84(4):1082--1104, Apr. 2022.

\bibitem{Kessel2018}
A.~Kessel and N.~Ben-Tal.
\newblock {\em Introduction to Proteins}.
\newblock Chapman and Hall/{CRC}, Mar. 2018.

\bibitem{Klenke}
A.~Klenke.
\newblock {\em Probability Theory}.
\newblock Springer International Publishing, Cham, Switzerland, 2020.

\bibitem{Lazar}
T.~Lazar, M.~Guharoy, W.~Vranken, S.~Rauscher, S.~J. Wodak, and P.~Tompa.
\newblock Distance-based metrics for comparing conformational ensembles of intrinsically disordered proteins.
\newblock {\em Biophysical Journal}, 118(12):2952--2965, 2020.

\bibitem{datafission}
J.~Leiner, B.~Duan, L.~Wasserman, and A.~Ramdas.
\newblock {Data Fission: Splitting a Single Data Point}.
\newblock {\em J. Am. Stat. Assoc.}, 2023.

\bibitem{Liljas:2009}
A.~Liljas, L.~Liljas, J.~Piskur, G.~Lindblom, P.~Nissen, and M.~Kjeldgaard.
\newblock {\em Textbook Of Structural Biology}.
\newblock World Scientific Publishing, Singapore, 2009.

\bibitem{liu2018powerful}
K.~Liu, J.~Markovic, and R.~Tibshirani.
\newblock More powerful post-selection inference, with application to the lasso, 2018.
\newblock arXiv:1801.09037.

\bibitem{mahalanobis}
P.~Mahalanobis.
\newblock On the generalized distance in statistics.
\newblock {\em Proceedings of the National Academy of Sciences, India (Calcutta)}, 2(1):44--55, 1936.

\bibitem{Nishikawa}
K.~Nishikawa, T.~Ooi, Y.~Isogai, and N.~Sait\^{o}.
\newblock Tertiary structure of proteins. i. representation and computation of the conformations.
\newblock {\em Journal of the Physical Society of Japan}, 32(5):1331--1337, 1972.

\bibitem{Ntranos2019}
V.~Ntranos, L.~Yi, P.~Melsted, and L.~Pachter.
\newblock A discriminative learning approach to differential expression analysis for single-cell {RNA}-seq.
\newblock {\em Nature Methods}, 16(2):163--166, Jan. 2019.

\bibitem{Oldfield:2014}
C.~J. Oldfield and A.~K. Dunker.
\newblock Intrinsically disordered proteins and intrinsically disordered protein regions.
\newblock {\em Annual Review of Biochemistry}, 83(1):553--584, 2014.

\bibitem{Ozenne:2012}
V.~Ozenne, F.~Bauer, L.~Salmon, J.-r. Huang, M.~R. Jensen, S.~Segard, P.~Bernadó, C.~Charavay, and M.~Blackledge.
\newblock {Flexible-meccano: a tool for the generation of explicit ensemble descriptions of intrinsically disordered proteins and their associated experimental observables}.
\newblock {\em Bioinformatics}, 28(11):1463--1470, 2012.

\bibitem{Pearce2021}
R.~Pearce and Y.~Zhang.
\newblock Deep learning techniques have significantly impacted protein structure prediction and protein design.
\newblock {\em Current Opinion in Structural Biology}, 68:194--207, June 2021.

\bibitem{phillips1970}
D.~Phillips.
\newblock British biochemistry, past and present.
\newblock In {\em London Biochemical Society Symposia}, page~11. Academic Press, 1970.

\bibitem{rasines}
D.~G. Rasines and G.~A. Young.
\newblock {Splitting strategies for post-selection inference}.
\newblock {\em Biometrika}, 12 2022.
\newblock asac070.

\bibitem{sagar2021}
A.~Sagar, C.~M. Jeffries, M.~V. Petoukhov, D.~I. Svergun, and P.~Bernadó.
\newblock Comment on the optimal parameters to derive intrinsically disordered protein conformational ensembles from small-angle {X}-ray scattering data using the ensemble optimization method.
\newblock {\em Journal of Chemical Theory and Computation}, 17(4):2014--2021, 2021.

\bibitem{Shao2007}
J.~Shao, S.~W. Tanner, N.~Thompson, and T.~E. Cheatham.
\newblock Clustering molecular dynamics trajectories: 1. characterizing the performance of different clustering algorithms.
\newblock {\em Journal of Chemical Theory and Computation}, 3(6):2312--2334, Oct. 2007.

\bibitem{Soloveychik2016Jul}
I.~Soloveychik and D.~Trushin.
\newblock {Gaussian and robust Kronecker product covariance estimation: Existence and uniqueness}.
\newblock {\em J. Multivariate Anal.}, 149:92--113, July 2016.

\bibitem{TRENCH}
W.~F. Trench.
\newblock Asymptotic distribution of the spectra of a class of generalized kac–murdock–szegö matrices.
\newblock {\em Linear Algebra and its Applications}, 294(1):181--192, 1999.

\bibitem{Vandenbon2020}
A.~Vandenbon and D.~Diez.
\newblock A clustering-independent method for finding differentially expressed genes in single-cell transcriptome data.
\newblock {\em Nature Communications}, 11(1), Aug. 2020.

\bibitem{inverse_AR}
A.~P. Verbyla.
\newblock A note on the inverse covariance matrix of the autoregressive process1.
\newblock {\em Australian Journal of Statistics}, 27(2):221--224, 1985.

\bibitem{inverse_AR_123}
J.~Wise.
\newblock The autocorrelation function and the spectral density function.
\newblock {\em Biometrika}, 42(1/2):151--159, 1955.

\bibitem{Yeh}
J.~Yeh.
\newblock {\em Real Analysis}.
\newblock World Scientific, 3rd edition, 2014.

\bibitem{Yun2023Jan}
Y.-J. {Yun} and R.~{Foygel Barber}.
\newblock {Selective inference for clustering with unknown variance}.
\newblock {\em Electronic Journal of Statistics}, 17(2):1923 -- 1946, 2023.

\end{thebibliography}





\end{document}